\documentclass[manuscript,screen]{acmart}

\usepackage{mathtools}
\usepackage{xfrac}
\usepackage{physics}
\usepackage{tensor}
\usepackage{thm-restate}
\usepackage{stmaryrd}
\usepackage{booktabs}
\usepackage{tikz}
\usetikzlibrary{arrows.meta,cd,calc,positioning,decorations.pathmorphing,fit,shapes.misc,patterns,patterns.meta,quotes,intersections}
\usepackage{subcaption}
\usepackage{apxproof}
\usepackage{soul}
\usepackage{extarrows}
\usepackage{mathpartir}
\usepackage{centernot}
\usepackage{amsmath}

\input{macro}

\newcommand{\ltssize}{\small}
\tikzset{arrow/.style={-stealth}}
\tikzset{darrow/.style={>->}}
\tikzset{tarrow/.style={-{stealth}{stealth}}}
\tikzset{parrow/.style={-stealth,decoration={snake,amplitude=1pt,segment
length=8pt,post length=4pt},decorate,thick}}
\tikzset{conf/.style={font=\ltssize}}
\tikzset{ci/.style={font=\ltssize,draw, circle}}
\tikzset{psplit/.style={draw, circle, fill=black, inner sep=1.5pt}}
\tikzset{hsplit/.style={inner sep=0pt}}
\tikzset{csplit/.style={fill, circle, inner sep=2pt}}
\tikzset{weight/.style={font=\ltssize,midway}}
\tikzset{csarrow/.style={-stealth,decorate,decoration={zigzag,amplitude=0.7pt,segment
length=1.2mm,pre length=4pt,post length=2pt}}}

\setcopyright{acmlicensed}
\copyrightyear{2018}
\acmYear{2018}
\acmDOI{XXXXXXX.XXXXXXX}
\acmConference[Conference acronym 'XX]{Make sure to enter the correct
conference title from your rights confirmation email}{June 03--05,
2018}{Woodstock, NY}
\acmISBN{978-1-4503-XXXX-X/2018/06}

\theoremstyle{acmdefinition}
\newtheorem{fact}{Fact}
\newtheorem{assumption}{Assumption}

\begin{document}

\title{Verification of Quantum Protocols Adopting Physically Admissible Schedulers}

\author{Lorenzo Ceragioli}
\email{lorenzo.ceragioli@imtlucca.it}
\orcid{0000-0002-1288-9623}
\affiliation{%
	\institution{IMT School for Advanced Studies Lucca}
	\city{Lucca}
	\country{Italy}
}

\author{Fabio Gadducci}
\email{fabio.gadducci@unipi.it}
\orcid{0000-0003-0690-3051}
\affiliation{%
	\institution{University of Pisa}
	\city{Pisa}
	\country{Italy}}

\author{Giuseppe Lomurno}
\email{giuseppe.lomurno@phd.unipi.it}
\orcid{0009-0000-0573-7974}
\affiliation{%
	\institution{University of Pisa}
	\city{Pisa}
	\country{Italy}}

\author{Gabriele Tedeschi}
\email{gabriele.tedeschi@phd.unipi.it}
\orcid{0009-0002-5345-9141}
\affiliation{%
	\institution{IRIF, CNRS -- Université Paris Cité}
	\city{Paris}
	\country{France}}

\renewcommand{\shortauthors}{Ceragioli et al.}

%
%
%

\begin{abstract}
	Recent developments in the field of quantum communication demonstrate that secure communication protocols based on quantum features are already practical.
	Reliable verification techniques are of paramount importance for these technologies,
	given their high implementation cost and critical contexts of application.
	Extensions of process calculi such as CCS and $\pi$-calculus have been proposed in the literature, together with various notions of behavioural equivalence.
	However, their standard probabilistic models turn out to introduce some non-deterministic capabilities that are not aligned with the observational properties of physical quantum systems, leading to bisimilarity notions that distinguish processes that should be physically equivalent.
	Nonetheless, we argue that non-deterministic features are fundamental to account for inputs, environments and adversarial behaviour.
	
	To address this issue, we propose lqCCS, a process calculus that integrates concurrency, non-determinism and quantum capabilities.
	The calculus is enriched with a linear type system that enforces the no-cloning principle and resolves ambiguities in the visibility of ancillary qubits.
	We introduce a novel semantics in terms of distributions, where explicit physically admissible schedulers constrain probabilistic composition and forbid ill-defined non-deterministic moves, while preserving the expressivity needed to model real-world protocols.
	We investigate a scheduled version of saturated bisimilarity, deeming two lqCCS processes behaviourally equivalent if no observer can tell them apart.
	The adequacy of the approach is verified by lifting a known result from quantum mechanics to lqCCS, i.e. that
	equivalent processes acting on indistinguishable mixtures of quantum states are correctly recognized as bisimilar.
	
	Finally, we give an alternative semantics and a labelled bisimilarity based on a quantum generalization of probability distributions.
	This provides an equivalent characterization of our behavioural equivalence that is a congruence with respect to the parallel operator, enabling compositional reasoning without the need to explicitly check all possible contexts.
	We describe a rich class of lqCCS processes for which equivalence is decidable
	using standard techniques, and we analyse real-world quantum communication protocols.
\end{abstract}

\ccsdesc[500]{Theory of computation~Process calculi}
\ccsdesc[500]{Theory of computation~Quantum computation theory}
\ccsdesc[500]{Software and its engineering~Formal software verification}
\ccsdesc[100]{Networks~Protocol correctness}
\ccsdesc[100]{Theory of computation~Probabilistic computation}

\keywords{Quantum Communication, Linear Process Calculi, Behavioural
	Equivalence, Probabilistic Bisimulation.}

\received{20 February 2007}
\received[revised]{12 March 2009}
\received[accepted]{5 June 2009}

\maketitle

\section{Introduction}
The recent advances in \emph{quantum computation} and \emph{quantum communication} mark a shift in the capabilities of quantum technologies, representing what many consider the second quantum revolution~\cite{revolution}.
This wave of development promises to harness quantum phenomena, particularly superposition and entanglement, for achieving groundbreaking results in algorithmic complexity and efficient implementation of secure data transmission.

At the forefront of this technological transformation lies \emph{quantum communication}, one of the four pillars of quantum technologies~\cite{pillars}.
The unique properties of quantum systems, where information cannot be observed without stochastically altering the system state, provide physical safeguards against eavesdropping and enable the development of protocols that require exponentially less communication resources with respect to their classical counterparts~\cite{brassard2003quantum}.
Protocols have been developed targeting several use cases, such as key distribution, cryptographic coin tossing, private information retrieval, byzantine agreement~\cite{bb84,gaoQuantumPrivateQuery2019,byzantine2005}.
The benefits of these developments also impact the field of quantum algorithms.
By leveraging multiple local systems connected both with quantum and classical channels, distributed implementations of algorithms can alleviate some scalability and fidelity issues still present in state-of-the-art quantum computers.
In a rapidly evolving landscape, quantum communication channels are already practical, with both academic and real-world deployments.
The vision of a worldwide network of quantum-capable machines, the \emph{Quantum Internet}, capable of performing both cryptographic protocols and distributed computing is no longer confined to theoretical speculation~\cite{zhangFutureQuantumCommunications2022, kimblequantum2008}.

This flourishing development calls for theoretically grounded modelling and verification techniques, especially due to the substantial implementation costs of quantum technologies and their expected deployment in security-critical contexts.
However, while the research on quantum circuits is mature, the verification of interactive quantum protocols received less attention.
Following the well-established pattern of investigation for classical distributed systems with probabilistic and cryptographic features~\cite{CCS, larsenbisimulation1991, spicalc}, we aim to develop a verification framework made of three components: $(i)$ a description language for concurrent or distributed quantum processes; $(ii)$ an adequate semantics; and $(iii)$ a behavioural equivalence for comparing processes and verify their expected properties.
In line with our intention, this framework will serve as a minimal starting point, needed as a stable foundation for building more advanced tools and verification mechanisms.
The adequacy of the framework requires two properties:
the \emph{correctness} with respect to the prescriptions of quantum mechanics,
and the \emph{expressivity} needed to model concrete real-world protocols.
We thus prove the correctness of our proposal, in particular with respect to the no-cloning theorem and the uncertainty principle, so that the system state can be inspected only by performing measurements.
To guarantee the needed expressivity, we allow combining quantum and classical features, since quantum protocols rely on standard control flow and interaction with the external environment, usually modelled as non-deterministic agents.

The algebraic formalisms of process calculi, like CCS~\cite{CCS} and CSP~\cite{CSP}, have proven successful as modelling languages for classical distributed and concurrent systems.
This success has naturally led to proposals for quantum counterparts, with several quantum process calculi sharing common features to model quantum capable participants of protocols~\cite{lalireprocess2004,gaycommunicating2005,fengprobabilistic2007}.
Despite these efforts, the community has not yet reached a consensus on the standard approach for quantum process calculi, and quantum protocols continue to be predominantly described in natural language, which is both error-prone and incompatible with automated verification techniques~\cite{zoo}.

As a modelling formalism, we propose lqCCS, a process calculus
that seamlessly integrates classical communication, parallelism, and non-determinism with quantum capabilities.
We enrich our calculus with a linear type system to enforce the no-cloning principle, and to resolve ambiguities regarding the visibility of ancillary qubits~\cite{ceragioliQuantumBisimilarityBarbs2024}.
We demonstrate the expressivity of lqCCS by modelling and analysing three real-world protocols.

As a critical aspect of our modelling, we stress the necessity in our language of purely non-deterministic features, which are used to encode choices that are not specified at modelling time.
When no assumption can be made about the system behaviour, we must consider each possible alternative as a separate branch, thus expressing our properties in terms of an implicit universal quantification.
This is fundamentally different from stochastic behaviour, where we assume that the available options come from a known probability distribution, each with its own likelihood.
Non-deterministic capabilities are useful in modelling protocols where the user input is arbitrary, and in which implementation details are not made explicit.
In addition, most proposed quantum protocols are cryptographic in nature.
Therefore, we must account for attackers that are capable of tricking the system into following an adversarial execution path, even if it is unlikely when considering random probabilities.
Lastly, non-determinism helps in modelling the behaviour of parallel processes,  especially in dealing with race condition, i.e. when the sequence or timing of uncontrollable events may lead to unexpected or inconsistent results.
The latter is even more critical in the quantum setting, where concurrent transformation of entangled states can influence the behaviour of other agents even without any explicit interaction or synchronization.

Regarding semantics and behavioural equivalences, the ones of existing quantum processes calculi are based on probabilistic labelled transition systems (pLTSs) and probabilistic bisimilarity.
The choice seems natural, given the stochastic nature of quantum measurements;
however, their use proves inadequate for quantum systems due to their failure to respect the observational limitations inherent in quantum theory.
The core issues lie in the interplay between non-determinism, quantum features, and probabilistic behaviour.
Previous proposals implicitly allow resolving non-determinism according to the unknown state of qubits, thus permitting a form of errorless quantum state inspection that causes no measurement-induced disturbance.
The violation of this fundamental principle of quantum information manifests in practical terms as the presence of spurious moves in the semantics of processes, and causes behavioural equivalence to distinguish between quantum states that are indistinguishable according to quantum mechanics~\cite{kubotaapplicationnoyear,davidsonformal2012}.
To ensure that our framework strictly adheres to the prescriptions of quantum mechanics, we propose a new semantics of processes that makes explicit the scheduling choices performed by protocol participants, as done in the classical setting~\cite{chatzikokolakisBisimulationDemonicSchedulers2009}.
Then, we substitute the traditional pLTS-based semantics for one based on a labelled transition system where states hold distributions of processes.
In doing so, we employ a constrained version of probabilistic composition, which only allows composing moves characterized by the same non-deterministic choices.
These well-behaved schedulers strictly prohibit the unfeasible state inspections found in previous models, ensuring that the system state can only be accessed via proper measurement operators.

To give a natural notion of indistinguishability, we resort to saturated bisimilarity, a form of contextual equivalence that deems processes equivalent if there is no observer capable of distinguishing between the two~\cite{bonchigeneral2014}.
More in detail, termination is the only property of lqCCS processes that is immediately observable, while other properties are inspected indirectly by putting the processes in parallel with other lqCCS processes called observers.
This approach favours correctness and intuitiveness over decidability, reducing equivalence to an immediate notion of operational comparison on one hand, while at the same time requiring a universal quantification over all possible observers on the other hand.
In this context, we discuss the adequacy of our approach, and we prove that using our physically feasible schedulers is sufficient to lift a known indistinguishability results from quantum states to lqCCS processes, namely that probabilistic mixtures of quantum states are indistinguishable if they can be represented as the same density operator.
The same lifted property does not hold for the traditional approach from probabilistic systems that is often employed in quantum process calculi~\cite{dengopen2012,fengtoward2015-1,dengbisimulations2018,long24}.

Whereas contextual equivalences, like our saturated bisimilarity, have a clear operational interpretation that simplify reasoning on their correctness with respect to the real-world discriminating power, traditional labelled bisimilarities in terms of atomic observables offer an elegant mathematical formalization, and are easier to compute~\cite{bonchigeneral2014}.
Therefore, we give an alternative succinct semantics that relies on a quantum generalization of probability distributions, and we characterize
behavioural equivalence in terms of labelled bisimilarity.
Finally, we prove the two characterizations equivalent, relieving the prover from the burden of manually considering each possible distinguishing context.
Noticeably, our labelled bisimilarity is closed under parallel composition, a feature often missing in literature that permits compositional reasoning on protocols by possibly verifying the behaviour of each agent in isolation.

To demonstrate the expressivity and practical utility of our framework, we apply lqCCS to the formalization and verification of real-world cryptographic systems. Using the Quantum Coin Flip (QCF) protocol as a primary example, we illustrate how our proposed bisimilarity can be utilized to rigorously model both the implementation and the specification of the system, successfully proving essential properties such as protocol fairness, data confidentiality, and the reliable detection of malicious cheating strategies.

\paragraph*{Synopsis}
In~\autoref{sec:overview} we introduce our process calculi and exemplify our approach by analysing the quantum coin flip protocol.
In~\autoref{sec:back} we give some background about probability distributions and quantum computing.
We formally define lqCCS with schedulers in~\autoref{sec:calc}, and discuss the adequacy of saturated bisimilarity in~\autoref{sec:sat}.
In~\autoref{sec:lab} we give an alternative semantics and equivalent labelled bisimilarity.
In~\autoref{sec:real} we apply our approach to three real-world protocols.
Finally, we compare related works with ours in~\autoref{sec:rel}, and we wrap-up the paper in~\autoref{sec:conc}.
Appendices contain the full proofs of our results.

\paragraph*{Remark}
This is an extended version of~\cite{aplas24};
this paragraph discusses the novel contributions.
In~\autoref{sec:calc}, we relax our assumption about deterministic schedulers by allowing external non-determinism when receiving values on a channel.
In~\autoref{sec:sat}, we show that the problems caused by non-determinism also
affect the modelling of isolated parallel processes that do not communicate (see~\autoref{ex:broken-ftl}).
In a nutshell, the measurement-induced collapse is detected by a non-deterministic	observer as if a message was sent faster than light.
As a main novel contribution,~\autoref{sec:lab} describes a new semantics for lqCCS based on a quantum generalization of probability distributions, which internalize the equivalence stated in~\autoref{thm:propertyA}, and proves it fully abstract with respect to the probabilistic one (see~\autoref{thm:demonicfa}).
Our labelled bisimilarity requires relations to be closed for the application of arbitrary quantum superoperators,
which hinder its decidability.
Here, we prove that it cannot be omitted in general (\autoref{ex:superopneed}),
but also that it is not needed for a rich class of processes, for which our bisimilarity is decidable (\autoref{thm:demonicgfa}).
Finally, in~\autoref{sec:real} we apply our approach beyond the standard verification of equivalence between specification and implementation:
we model and verify the quantum coin flip protocol, showing how we can resort to our formal verification technique
to characterize the desired security properties of a cryptographic protocol.

\section{Preview of Our Verification Framework}\label{sec:overview}
In this section, we provide an overview of our formal approach to the verification of quantum protocols by presenting our modelling language and its intended usage on a concrete real-world example, the quantum coin flip protocol.

\subsection{The {lqCCS} Process Calculus}

As a modelling language, we propose lqCCS, a process calculus inspired by qCCS~\cite{fengbisimulation2012} where the primitive features of communication, parallelism and non-determinism of (value passing) CCS~\cite{vpCCS} are extended with quantum capabilities.

\begin{definition}[LqCCS Process Syntax]\label{def:lqccs_syntax}
	Let $\btype \ni b, \ntype \ni n, \qtype \ni q$ be atomic boolean values, natural numbers and qubit names, respectively.
	Let $\chanset \ni c$ be communication channels, and $\varset \ni x$ a set of variables.

	The syntax of lqCCS processes is defined by the productions
	\begin{align*}
		P      &
		\Coloneqq
		\nil_{\tilde{e}}
		\mid \alpha. P
		\mid \ite{e}{P}{P}
		\mid P + P
		\mid P \parallel P
		\mid P \setminus c       \\
		\alpha &
		\Coloneqq
		\tau
		\mid c?x
		\mid c!e
		\mid \E(\tilde{e})
		\mid \meas{\tilde{e}}{x} \\
		e      &
		\Coloneqq x
		\mid b
		\mid n
		\mid q
		\mid \neg e
		\mid e \lor e
		\mid e \leq e
		\mid e = e
	\end{align*}
	where $\tilde{e}$ denotes a (possibly empty) tuple $e_1, \ldots, e_n$ of expressions.
\end{definition}

Processes can perform atomic operations, written as prefixes and followed by processes representing their continuation.
As it is common in CCS, we assume that a process can
perform an invisible internal action with $\tau$;
receive over a channel $c$ a value to bind to a variable $x$ with $c?x$;
send a value obtained by evaluating an expression $e$ over a channel $c$ with $c!e$.
Note that channels abstractly model concrete communication means for classical and quantum values: the latter allow propagating qubit values (in some physical representation like light polarization) and are guaranteed to preserve their quantum features like superposition.
In addition, a process can exploit its quantum capabilities to manipulate the state of some qubits via trace-preserving superoperators and
measurement.
We allow processes to modify the state of qubits denoted by $\tilde{e}$ (where each expression is either a qubit $q$ or a variable $x$) by applying any trace-preserving superoperator $\E$ with $\E(\tilde{e})$.
For example, the process can update the state of qubits by performing unitary transformations like $\cnot(q_0,q_1)$, the controlled not over $q_0$, and $H(q_0)$, the Hadamard transformation over $q_0$.
We write $\meas{\tilde{e}}{x}$ for the operation that performs the measurement $\M$ to the qubits of $\tilde{e}$: the classical outcome is stored in the number variable $x$, and the state of the qubits is updated accordingly.
For example, the process can apply $\mstd, \mhad$ or $\mhadi$, i.e. the projective measurements in the bases $\{ \kz, \ko \}, \{ \kpl, \km \}$ and $\{ \kip, \kim \}$.
Finally, the process $\nil_{\tilde{e}}$ is an idle deadlock process that maintains ownership of the qubits in ${\tilde{e}}$.
When $\tilde{e}$ is the empty sequence, we write $\nil$ to stress the equivalence with the nil process of standard CCS\@.

Processes can be composed with parallel, non-deterministic, and boolean operators.
With the parallel composition operator ($P \parallel Q$) each component can act independently or synchronize through communication over the same channel.
The non-deterministic choice operator ($P + Q$) allows the process to select a specific component behaviour according to a yet unspecified policy.
When a choice is performed according to a boolean expression the processes are composed with the conditional operator ($\ite{e}{P}{Q}$).
Finally, the communication between two or more channels can be isolated from other processes with the restriction operator ($P \setminus c$), meaning that the channel $c$ is private and can be used only inside $P$.
For brevity, we sometimes write $\ite{e}{\alpha_1.R}{\alpha_2.R}$ as $(\ite{e}{\alpha_1}{\alpha_2}).R$, and we often omit the trailing $\nil$, writing, e.g. $a!1$ for $a!1.\nil$.

To comply with the no-cloning theorem, we enforce a single ownership principle, forbidding parallel processes from sharing qubit names.
Because of that, $\nil_{\tilde{e}}$ make the qubits in ${\tilde{e}}$ inaccessible to other processes: we say that it \emph{discards} them.
Explicitly writing the discarded qubits allows distinguishing those qubits that are hidden to the environment from the ones that can be accessed by external observers, solving an ambiguity in related works, see~\autoref{sec:rel}.

\begin{figure}
	\begin{tabular}{c c}
		$\begin{aligned}
				 \proc{Alice} \Coloneqq \    & \measrand{s}.                                                                                                \\[-.1cm]
				                             & \ite{s = 0}
				 {\un{H}(q).\measstd{q}{w}.                                                                                                                 \\[-.1cm]&\quad}
				 {\un{I}(q).\meashad{q}{w}.}                                                                                                                \\[-.1cm]
				                             & \text{AtoB}!q.                                                                                               \\[-.1cm]
				                             & \text{guess}?g.                                                                                              \\[-.1cm]
				                             & \text{witness}!w.                                                                                            \\[-.1cm]
				                             & \text{secret}!s.                                                                                             \\[-.1cm]
				                             & a!(g=s))                                                                                                     \\[.1cm]
				 \proc{QCF}      \Coloneqq\  & (\proc{Alice} \parallel \proc{Bob}) \setminus \{ \text{AtoB}, \text{guess}, \text{secret}, \text{witness} \}
			 \end{aligned}$ &
		$\begin{aligned}
				 \proc{Bob} \Coloneqq \  & \text{AtoB}?x.                                                              \\[-.1cm]
				                         & \measrand{g}.                                                               \\[-.1cm]
				                         & \ite{g=0}
				 {\meashad{x}{\mathit{p}}.                                                                             \\[-.1cm]&\quad}
				 {\measstd{x}{\mathit{p}}.}                                                                            \\[-.1cm]
				                         & \text{guess}!g.                                                             \\[-.1cm]
				                         & \text{witness}?w.                                                           \\[-.1cm]
				                         & \text{secret}?s.                                                            \\[-.1cm]
				                         & (b!(g=s) \parallel \text{cheat}!(g \neq s \land p \neq w) \parallel \nil_x) \\[.1cm]
				 \phantom{\proc{QCF}}
			 \end{aligned}$
	\end{tabular}
	\caption{The Quantum Coin Flip protocol $\proc{QCF}$, modelled using lqCCS.}
	\label{fig:QCF}
\end{figure}

\subsection{Modelling and Analysing the Quantum Coin Flip Protocol}

To exemplify our approach, we use a non-trivial concrete case, the Quantum Coin Flip protocol~\cite{bb84}, which is also a basis for implementing the more involved Quantum Byzantine Agreement~\cite{byzantine2005}.
The objective of the protocol is to randomly select a winner between two participants, named Alice and Bob, which do not trust each other.
In a nutshell, Alice chooses a secret value between $0$ and $1$, and Bob tries to guess it: if he is right, he wins the game, otherwise he loses.
As we will detail later, the
underlying idea is to exploit the quantum features to guarantee
different properties, e.g.\ that the game is fair (if both players follow the rules then each of them has equal winning probability), and that the players cannot cheat without risking getting caught.

Notably, lqCCS is expressive enough to model such kind of real-world protocols.
\begin{example}
	The formalization of the Quantum Coin Flip protocol in lqCCS is in~\autoref{fig:QCF}, where we assume the qubit $q$ is prepared as $\kz$.
	Alice starts with a simple probabilistic coin toss,
	formally represented as a measurement on no qubits $\measrand{s}$: the \emph{secret basis} $s$ will be instantiated as $0$ or $1$ with equal probability.
	If $s$ is $0$ then Alice will use the standard basis, otherwise, the Hadamard basis is used.
	The unitary transformation before measurement guarantees that each measurement result has the same probability of being observed.
	Depending on $s$, the qubit $q$ is thus in either $\{\kz, \ko\}$ or $\{\kpl, \km\}$, each with equal probability, and the classical \emph{witness} $w$ is a bit that corresponds to the quantum state (by convention, the first element of the basis corresponds to $0$, while the second to $1$).
	Alice, next, sends the prepared qubit to Bob, and keeps both $s$ and $w$ private.
	Bob, who is waiting on the main communication channel, receives the qubit and randomly selects a basis $g$: he will win if the guessed basis $g$ coincides with Alice's secret one $s$.
	If Bob selects the standard basis, he measures the received qubit on the Hadamard one, and vice versa.
	The classical result of the measurement is stored in the variable $p$, serving as a \emph{fingerprint} of the original bit randomly generated by Alice, which will be used for checking if she is cheating.
	The guess of Bob is communicated to Alice via the $\text{guess}$ channel.
	At this point, Alice reveals the secret $s$ and the witness $w$, and she computes the winner: she sends $1$ over her channel $a$ if Bob is right, $0$ otherwise.
	Bob receives the secret, he also computes and advertises the winner on his channel $b$, with the same encoding as Alice.
	Finally, he checks the received witness $w$: if the secret basis is not $g$, then it must be the one used by Bob for measuring the received qubit, and $w$ must coincide with $p$.
	If this is not the case, a $1$ bit is sent over the $\text{cheat}$ channel, signalling that a misbehaviour has been detected. \exqed
	%

\end{example}

A distinctive feature of process calculi is that they can be used to model both the implementation of a protocol and its specification: the intended behaviour of the system as seen from an external observer, in a black-box manner.
%
In this setting, verification coincides with proving that no observer can tell the difference between implementation and specification, which is indeed the primary goal of our proposal.
%
To this end, we introduce a notion of contextual equivalence known as \emph{saturated bisimilarity}: namely, two processes are deemed bisimilar if no process in parallel can distinguish one from the other.
As we will see, the standard probabilistic bisimilarity that is often employed in the literature does not coincide with the distinctive observational capability permitted by the inherent limitations of quantum theory.
To address this problem, we propose a new relation $\simds$, based on physically admissible observers, characterized by a constrained probabilistic composition of non-deterministic choices, and we prove its adequacy to the prescriptions of quantum theory.
\begin{example}
	The lqCCS modelling of the specification for the protocol is below, and corresponds to the behaviour of a fair coin that selects Alice or Bob as the winner with equal probability.
	In addition, the bit $0$ is the only one that is sent over the $\text{cheat}$ channel, because Alice in~\autoref{fig:QCF} is not cheating.
	\begin{align*}
		\proc{FairCoin} \Coloneqq \measrand{x}.\tau^{8}.(a!x \parallel b!x \parallel \text{cheat}!0)
	\end{align*}
	In~\autoref{sec:real}, we will show that $\proc{QCF}$ and $\proc{FairCoin}$ are observationally equivalent: the protocol works as desired. \exqed
\end{example}

For some protocols, provably behaving as expected, when all participants act as desired, is not enough.
In security contexts, like in coin flip and key exchange protocols, there are other important properties which can be investigated by resorting to our proposed bisimilarity.
Ideally, we want to guarantee that in $\proc{QCF}$ cheating is either impossible, or at least that malicious attempts have a reasonable probability of being detected.
Checking that Bob cannot cheat to increase his winning chances amounts to verifying confidentiality of Alice's secret value $s$.
We can prove confidentiality by considering two versions of the same process that differ only in their private variables:
confidentiality is violated when the two are not bisimilar, i.e. if and only if a distinguishing context is found.
%
\begin{example}
	To guarantee that the protocol does not favour Bob, we show that he is not able to infer the Alice's secret bit $s$ from the received qubit $q$, even though it corresponds to the state of her chosen basis.
	To prove this it is sufficient to take the two versions of $\proc{Alice}$ when $s$ is $0$ or is $1$, and to consider the part of the process that happens before the secret is revealed.
	We name them $\proc{Alice}_0$ and $\proc{Alice}_1$, defined below:
	\begin{align*}
		\proc{Alice}_0 = \un{H}(q).\measstd{q}{w}.\text{AtoB}!q
		\qquad
		\proc{Alice}_1 = \un{I}(q).\meashad{q}{w}.\text{AtoB}!q
	\end{align*}
	We will show that, when starting with $q$ in the $\kz$ state, the two processes above are bisimilar, thus there exists no observer process
	that distinguishes the two versions so as, e.g., to execute $\text{guess}!0$ when Alice is observed to behave as $\proc{Alice}_0$ and $\text{guess}!1$ when she behaves as $\proc{Alice}_1$.
	Indeed, there is no clever way of inferring the choice of Alice by inspecting the received qubits: Bob cannot do better than choosing his basis at random. \exqed

\end{example}

The confidentiality result heavily depends on the constraints that we impose on non-determinism in our behavioural equivalence.
In particular, quantum theory prescribes indistinguishability classes of quantum states, thus a reasonable behavioural equivalence must reflect these classes.
Previous results~\cite{ceragioliQuantumBisimilarityBarbs2024} have shown that this requires constraining non-determinism, which is also the solution that we apply in this work.
Indeed, the traditional probabilistic bisimilarity incorrectly distinguishes the two versions of Alice in the previous example.
As a result of this limitation, one would mistakenly conclude that $\proc{QCF}$ favours Bob over Alice, which is provably not the case.

The security of a protocol can also be investigated by substituting an honest participant with a specific attacker, as we show below by substituting Alice with two malicious versions of her: Alison and Alix, both modelled in \autoref{fig:Alices}.
\begin{example}
	Assume that Alice wants to cheat.
	The most immediate strategy is to consider the guess of Bob and simply reveal the opposite basis, independently of the actual secret $s$.
	This variant is named $\proc{Alison}$, and is modelled by the left process in~\autoref{fig:Alices}.
	\begin{align*}
		\proc{QCF}' \Coloneqq\  & (\proc{Alison} \parallel \proc{Bob}) \setminus \{ \text{AtoB}, \text{guess}, \text{secret}, \text{witness} \}
	\end{align*}
	The attack of Alison is not dangerous: she is the legitimate winner half of the time, whereas in the other cases Bob is capable of detecting the attack with half probability due to the witness (the overall probability of seeing $1$ over the $\text{cheat}$ channel is thus $1/4$).
	Indeed, if the basis chosen by Bob does not coincide with the secret one of Alison, then the witness $w$ is not correlated to the fingerprint $p$ computed by Bob, and the two coincide only by chance if Alison is lucky.
	To show the success probability of Alison's attack, we will prove that $\proc{QCF}'$ is bisimilar to:
	\begin{align*}
		\meas[\M_{\frac{3}{4}}]{}{x}.\tau^{8}.(a!0 \parallel b!0 \parallel \text{cheat}!x)
	\end{align*}
	where $\meas[\M_{\frac{3}{4}}]{}{x}$ is the measurement on no qubits that instantiates $x$ to $0$ with probability $3/4$ and to $1$ with probability $1/4$.
	While we stick with the simplest case favouring simplicity, it is easy to show that increasing the number of qubits used for the witnesses leads to an exponential increase in the probability of Bob to detect the cheating attempt, since the witnesses she sends must exactly match \emph{each} of Bob's measurement results.

	However, Alice can adopt a better strategy, which always guarantees winning the toss with no risk of being discovered.
	For this attack, Alice must have access to an additional qubit $q'$.
	We name this variant $\proc{Alix}$, modelled by the rights process in~\autoref{fig:Alices}.
	\begin{align*}
		\proc{QCF}'' \Coloneqq\  & (\proc{Alix} \parallel \proc{Bob}) \setminus \{ \text{AtoB}, \text{guess}, \text{secret}, \text{witness} \}
	\end{align*}
	The idea behind this attack is that of preparing a state that collapses in two equivalent states when Bob performs his measurement (for any basis that Bob can choose).
	This allows Alix to forge a valid witness for the basis chosen by Bob, even if it is unknown.
	To show that this attack always succeeds, we will prove that $\proc{QCF}''$ is bisimilar to:
	\begin{align*}
		\tau^{9}.(a!0 \parallel b!0 \parallel \text{cheat}!0)
	\end{align*}
	Intuitively, Alix exploits quantum entanglement to correlate two qubits such that when measured they result in the same state. \exqed
\end{example}

\begin{figure}
	\begin{tabular}{c c}
		$\begin{aligned}
				 \proc{Alison} \Coloneqq \  & \measrand{s}.                       \\[-.1cm]
				                            & \ite{s = 0}
				 {\un{H}(q).\meas[\mstd]{q}{w}.                                   \\[-.1cm]&\quad}
				 {\un{I}(q).\meas[\mhad]{q}{w}.}                                  \\[-.1cm]
				                            & \text{AtoB}!q.                      \\[-.1cm]
				                            & \text{guess}?g.                     \\[-.1cm]
				                            & \text{witness}!w.                   \\[-.1cm]
				                            & (\text{secret}!(1-g) \parallel a!0)
			 \end{aligned}$ &
		$\begin{aligned}
				 \proc{Alix} \Coloneqq \  & \un{H}(q).                          \\[-.1cm]
				                          & \un{CNOT}(q,q').                    \\[-.1cm]
				                          & \text{AtoB}!q.                      \\[-.1cm]
				                          & \text{guess}?g.                     \\[-.1cm]
				                          & \ite{g = 0}
				 {\meas[\mhad]{q'}{w}.                                          \\[-.1cm]&\quad}
				 {\meas[\mstd]{q'}{w}.}                                         \\[-.1cm]
				                          & \text{witness}!w.                   \\[-.1cm]
				                          & (\text{secret}!(1-g) \parallel a!0)
			 \end{aligned}$
	\end{tabular}
	\caption{Two malicious versions of Alice, modelled using lqCCS.}
	\label{fig:Alices}
\end{figure}

Note the critical difference of the result of our analysis with respect to one resorting to standard probabilistic bisimilarity.
Whereas the standard approach deems that Bob can distinguish the state of the qubit sent by Alice, thus implicitly inferring the secret and winning the game,
our proposal correctly predicts that Alice (and only Alice) has a winning strategy.

\section{Background}\label{sec:back}
We recall some background on probability distributions
and quantum computing, referring to~\cite{nielsenquantum2010} for further reading.

\subsection{Probability Distributions}\label{probabilisticBackground}
Given a function $\Delta : S \to [0,1]$ from a set $S$ to the interval of real numbers $[0,1]$, we let its \emph{support} be the set $\supp{\Delta} = \{ s \in S \mid \Delta(s) > 0 \}$.
When the support is finite, we define the \emph{mass} of $\Delta$ as $\mass{\Delta} = \sum_{s \in S}\Delta(s)$.
A \emph{probability distribution} $\Delta \in \sdist{S}$ over a set $S$ is a function $\Delta : S \to [0,1]$ with finite support such that $\mass{\Delta} \leq 1$.
A \emph{full probability distribution} $\Delta \in \dist{S} \subsetneq \sdist{S}$ is a probability distribution with mass equal to $1$.

We let $\epsilon$ be the empty distribution with mass equal zero, i.e. such that $\epsilon(s) = 0$ for all $s \in S$.
For each $s \in S$, we let $\delem{p}{s}$ be the \emph{point distribution} such that $(\delem{p}{s})(s) = p$.
Finally, we define the following two operations over probability distributions:
the scaling of $\Delta$ with $p \in [0,1]$ is defined as the distribution $p \Delta$ with $(p \Delta) (s) = p (\Delta (s))$;
the sum of $\Delta$ and $\Theta$ is a partial operation defined when $\mass{\Delta} + \mass{\Theta} \leq 1$ as the distribution $\Delta \oplus \Theta$ with $(\Delta \oplus \Theta) (s) = \Delta(s) + \Theta(s)$.
In general, we write $\osum_{i \in I} \delem{p_i}{\Delta_i}$ for $p_1 \Delta_1 \oplus p_2 \Delta_2 \dots \oplus p_n \Delta_n$ when $i = \{1, 2 \dots, n\}$.
When writing $\Delta \oplus \Theta$, it is assumed that the operation is defined on the operands.

A relation $\mathord{\rel} \subseteq \sdist{S} \times \sdist{S}$ is \emph{left-linear} if
$\Delta_i \rel \Theta_i$ for $i \in I$ and
$\Delta = \osum_{i \in I} \delem{p_i}{\Delta_i}$ implies $\Theta = \osum_{i \in I} \delem{p_i}{\Theta_i}$ is defined and $\Delta \rel \Theta$, 
for any finitely indexed set of probabilities $p_i \in [0,1]$ and distributions $\Delta_i$ and $\Theta_i$.
A relation $\mathord{\rel} \subseteq \sdist{S} \times \sdist{S}$ \emph{right-linear} if
$\Delta_i \rel \Theta_i$ for $i \in I$ and
$\Theta = \osum_{i \in I} \delem{p_i}{\Theta_i}$ implies
$\Delta = \osum_{i \in I} \delem{p_i}{\Delta_i}$ is defined
and $\Delta \rel \Theta$,
for any finitely indexed set of probabilities $p_i \in [0,1]$ and distributions $\Delta_i$ and $\Theta_i$.
A relation is \emph{linear} if it is both left- and right-linear.
Notice that linearity implies convexity in the sense of~\cite{bonchipower2017},
because both $\Delta$ and $\Theta$ are always defined if $\sum_{i \in I} p_i = 1$.
A relation $\rel$ is \emph{left-decomposable} if for any 
$p$, $\Delta_1$, $\Delta_2$ and $\Theta$,
$(\delem{p}{\Delta_1} \oplus \delem{(1-p)}{\Delta_2}) \rel \Theta$ implies $\Theta = \delem{p}{\Theta_1} \oplus \delem{(1-p)}{\Theta_2}$ for some $\Theta_1, \Theta_2$ with $\Delta_i \rel \Theta_i$ for $i = 1,2$.
Right-decomposability is defined symmetrically, and a relation is \emph{decomposable} when it is both left- and right-decomposable.

Given $\mathord{\rel} \subseteq A \times \sdist{B}$, its \emph{left-linear lifting} $\lift(\rel) \subseteq \sdist{A} \times \sdist{B}$ is the smallest left-linear relation such that $\delem{1}{s} \mathrel{\lift(\rel)} \Theta$ if $s \rel \Theta$.
Regarding this lifting, we can prove the following lemma

\begin{restatable}{lemma}{lliftDecomp}\label{lem:llift_decomp}
	For any relation $\mathord{\rel} \subseteq S \times \dist{S}$, $\lift(\rel)$ is left-decomposable.
\end{restatable}

\subsection{Quantum Computing}

An isolated physical system is associated with a \emph{Hilbert space} $\hilb$, i.e. a complex vector space equipped with an inner product $\braket\blank$.
We indicate column vectors as $\kp$ and their conjugate transpose as $\bra\psi = \kp^\dagger$.
The states of a system are \emph{unit vectors} in $\hilb$, i.e. vectors $\kp$ such that $\braket\psi = 1$.
A two-dimensional physical system is known as a \emph{qubit}, and we denote its Hilbert space as $\qubit = \mathbb{C}^2$.
The vectors $\kz = (1,0)^T$ and $\ko = (0,1)^T$ form the \emph{computational basis} of $\qubit$.
Other important states are $\kpl = \frac{1}{\sqrt{2}}(\kz + \ko)$ and $\km = \frac{1}{\sqrt{2}}(\kz - \ko)$, which form the \emph{Hadamard basis}.
In quantum terminology, the states in the Hadamard basis are \emph{superpositions} with respect to the computational basis, as they are a linear combination of $\kz$ and $\ko$.
A third basis of $\qubit$ contains $\ket{i} = \frac{1}{\sqrt{2}}(\kz + i\ko)$ and $\ket{-i} = \frac{1}{\sqrt{2}}(\kz -i\ko)$.
All $\kpl, \km, \kip$ and $\kim$ are superpositions of $\kz, \ko$.

We represent the state space of a composite physical system as the \emph{tensor product} of the state spaces of its components.
Consider the Hilbert spaces $\hilb_A$ with $\{\ket{\psi_i}\}_{i \in I}$ one of its bases, and $\hilb_B$ with $\{\ket{\phi_j}\}_{j \in I}$ one of its bases.
We let their  tensor product $\hilb_A \otimes \hilb_B$ be the Hilbert space with bases $\{\ket{\psi_i} \otimes \ket{\phi_j}\}_{(i,j) \in I \times J}$, where $\kp \otimes \kf$ is the Kronecker product.
We often omit the tensor product and write $\ket{\psi\phi}$ for $\ket{\psi}\otimes\ket{\phi}$.
We write $\qubits{n}$ for the $2^n$-dimensional Hilbert space defined as the tensor product of $n$ copies of $\qubit$ (i.e.\ the possible states of $n$ qubits).
A quantum state in $\hilb_A \otimes \hilb_B$ is  \emph{separable} when it can be expressed as the Kronecker product of two vectors of $\hilb_A$ and $\hilb_B$.
Otherwise, it is \emph{entangled}, like the Bell state $\kp = \frac{1}{\sqrt{2}}(\ket{00} + \ket{11})$.

In quantum physics, the evolution of an isolated system is described by a unitary transformation.
For each linear operator $A$ on $\hilb$, its \emph{adjoint} $A^\dag$ is the unique linear operator such that $\mel{\psi}{A}{\phi} = \braket{A^\dag\psi}{\phi}$.
A linear operator $U$ is \emph{unitary} when $UU^\dag = U^\dag U = \I$, with $\I$ the identity matrix.
Quantum computers allow the programmer to manipulate registers via unitaries like
$H$, $X$, $Z$ and $\cnot$, 
which are defined by the following equalities and by linearity in 
all the other cases: 
\begin{gather*}
H \kz = \kpl \quad H \ko = \km \quad X \kz = \ko \quad X \ko = \kz \quad Z \kpl = \km \quad Z \km = \kpl\\
\cnot \ket{10} = \ket{11}\quad \cnot \ket{11} = \ket{10}\quad \cnot\ket{0\psi} = \ket{0\psi}
\end{gather*}

\subsection{Density operator formalism}

The density operator formalism puts together quantum systems and probability distributions by considering mixed states, i.e.\ \emph{probability distributions of quantum states}.
A point distribution $\point{\kp}$ (called a pure state) is represented by the matrix $\kbp$.
In general, a mixed state $\Delta \in \sdist{\qubits{n}}$ for $n$ qubits is represented by the matrix $\rho_\Delta \in \cfield^{2^n\times 2^n}$, known as its \emph{density operator}, with $\rho_\Delta = \sum_{\kp} \Delta(\kp) \kbp$.
We write $\sdm{\hilb}$ for the set of density operators of $\hilb$ and $\dm{\hilb}$ for the set $\{\rho \in \sdm{\hilb} \mid \tr(\rho) = 1\}$, corresponding to full probability distributions of quantum states.
For example, the mixed state $\delem{\frac{1}{3}}{\kz} \oplus \delem{\frac{2}{3}}{\kpl}$ is represented by the density operator $\frac{1}{3} \kbz + \frac{2}{3} \kbpl$.
The empty distribution $\epsilon$ corresponds instead to the all-zero matrix $\Z$.

Note that the encoding of probabilistic mixtures of quantum states as density operators is not injective.
For example, $\frac{1}{2}\I$ is called the \emph{maximally mixed state} and represents both distributions $\Delta_{C} = \delem{\frac{1}{2}}{\kz} \oplus \delem{\frac{1}{2}}{\ko}$ and $\Delta_{H} = \delem{\frac{1}{2}}{\kpl} \oplus \delem{\frac{1}{2}}{\km}$.
This is a desired feature, as the laws of quantum mechanics consider indistinguishable all the distributions that result in the same density operator.
\begin{fact}\label{thm:qind}
	Two distributions of pure quantum states $\Delta, \Theta \in \dist{\hilb}$
	are indistinguishable for any physical observer whenever
	\[
		\rho_\Delta = \sum_{\kp \in \supp{\Delta}} \Delta(\kp) \kbp = \sum_{\kp \in \supp{\Theta}} \Theta(\kp) \kbp = \rho_\Theta.
	\]
\end{fact}
Since density operators encode distributions of pure states, the same result is easily extended to indistinguishable distributions of mixed states.

When modelling composite systems, density operators are composed with the Kronecker product as well.
Unlike pure states, they can also describe local information about subsystems.
Let $\hilb_{AB} = \hilb_A \otimes \hilb_B$ represent a composite system, with subsystems $A$ and $B$.
Given a (not necessarily separable) $\rho^{AB} \in \hilb_{AB}$, the state of the subsystem $A$ is described as the \emph{reduced density operator} $\rho^A = \tr_B(\rho^{AB})$, with $\tr_B$ the \emph{partial trace over $B$}, defined as the linear transformation such that $\tr_B(\ketbra{\psi}{\psi'} \otimes \ketbra{\phi}{\phi'}) = \ketbra{\psi}{\psi'}\tr(\ketbra{\phi}{\phi'})$.

When applied to pure separable states, the partial trace returns the actual state of the subsystem.
When applied to an entangled state, instead, it produces a mixed state, because ``forgetting'' the information on the subsystem $B$ leaves only partial
information on subsystem $A$.
For example, the partial trace over the first qubit of $\kbphip$ is the maximally mixed state.

The evolution of mixed states is given by \emph{trace non-increasing superoperators}, i.e.\ functions on density operators.
A superoperator $\E : \sdm{\hilb} \to \sdm{\hilb}$ on a $d$-dimensional Hilbert space $\hilb$ is a function defined by its \emph{Kraus operator sum decomposition} $\{E_i\}_{i = 1, \ldots, d^2}$, satisfying that $\E(\rho) = \sum_i E_i\rho E_i^\dag$ and $\sum_i E_i^\dag E_i \sqsubseteq \I$, where $A \sqsubseteq B$ means that $B - A$ is a positive semidefinite matrix.
We call $\soset{\hilb}$ the set of trace non-increasing superoperators on $\hilb$, and $\tsoset{\hilb} \subsetneq \soset{\hilb}$ the set of all \emph{trace-preserving} superoperators, i.e.\ such that $\sum_i E_i^\dag E_i = \I$.
Intuitively, trace-preserving superoperators are maps between distributions of quantum states that preserve the mass of the distributions, e.g.\ full-distributions to full-distributions, while trace non-increasing superoperators may map to distributions with less mass, e.g.\ full distributions to distributions.

Notably, the tensor product of superoperators is obtained by tensoring their Kraus decompositions, and any unitary transformation $U$ can be seen as a superoperator, that we still denote as $U$, with $\{ U \}$ as its Kraus decomposition.

\emph{Quantum measurements} describe how to extract information from a physical system.
Performing a measurement on a quantum state returns a probabilistic classical result and causes the quantum state to change (i.e.\ to \emph{collapse}).
A measurement with $k$ different outcomes is a set $\M = \{M_m\}_{m=0}^{k-1}$ of $k$ linear operators, satisfying the \emph{completeness} equation $\sum_{m=0}^{k-1} M_m^\dag M_m = \I$.
If the state of the system is $\rho$ before the measurement, then the probability of $m$-th outcome occurring is $p_m = \tr(M_m \rho M_m^\dagger)$. If $m$ is the outcome, then the state after the measurement will be $M_m \rho M_m^\dagger /p_m$.
Note that each operator $M_m$ defines a trace non-increasing superoperator $\M_m$, with $\M_m(\rho) = M_m \rho M_m^\dagger$, and the resulting state after the $m$-th outcome is the normalization of $\M_m(\rho)$.

The simplest measurements project a state into the elements of a basis, e.g.\ $\mstd = \{\kbz, \kbo\}$ and $\mhad = \{\kbpl, \kbm\}$ for the computational and Hadamard basis of $\qubit$.
As expected, applying $\mstd$ to $\kz$ always returns the classical outcome $0$ and the state $\kz$.
When applying the same measurement on $\kpl$, instead, the outcome is $0$ and the state is $\kz$, or $1$ and $\ko$ with equal probability.
Symmetrically, measuring $\kz$ with $\mhad$ leads to either $0$ and $\kpl$, or $1$ and $\km$, with equal probability.
Finally, applying the measurement $\mhadi = \{ \kbip, \kbim \}$ to each of $\kz$, $\ko$, $\kpl$ and $\km$ returns $0$ and $\kip$, or $1$ and $\kim$ with equal probability.



\section{Quantum Process Calculus and Probabilistic Semantics}\label{sec:calc}
In this section, we describe the internal syntax of lqCCS processes, i.e. one in which available actions are individually tagged.
We also define a type system for enforcing single ownership of qubits, thus guaranteeing adherence to the no-cloning theorem.
Our process calculus is enriched with a linear type system, reflecting the \emph{no-cloning theorem} of quantum mechanics, which forbids quantum values to be copied or broadcasted.
We define a formal semantics where transitions are decorated with schedulers.
Tags in the syntax name all the possible non-deterministic choices, while schedulers in the semantics allow choosing among the available tags.

\subsection{Internal Syntax and Type System}\label{Syntax}

\begin{figure}[!t]
	\begin{mathpar}
		\inferrule[Nil]{\tilde{e} \in \tilde{\Sigma}}{\Sigma \vdash \nil_{\tilde{e}}}
		\and
		\inferrule[Tau]{\Sigma \vdash P}{\Sigma \vdash t \tags \tau . P}
		\and
		\inferrule[TauPair]{\Sigma \vdash P}{\Sigma \vdash (t, t') \tags \tau . P}
		\and
		\inferrule[QOp]{\E : \opset(n) \\ |E| = n \\ \tilde{e} \in \tilde{E} \\ E \subseteq \Sigma \\ \Sigma \vdash P}{\Sigma \vdash t \tags \E(\tilde{e}) . P}
		\and
		\inferrule[QMeas]{\M : \measset(n) \\ |E| = n \\ \tilde{e} \in \tilde{E} \\ E \subseteq \Sigma \\ x : \ntype \\ \Sigma \vdash P}{\Sigma \vdash t \tags \meas{\tilde{e}}{x} . P}
		\and
		\inferrule[CRecv]{c : \chtype{T} \\ x : T \in \{\btype,\ntype\} \\ \Sigma \vdash P}{\Sigma \vdash t \tags c?x . P}
		\and
		\inferrule[QRecv]{c : \chtype{\qtype} \\ x : \qtype \\ \Sigma \cup \{x\} \vdash P}{\Sigma \vdash t \tags c?x . P}
		\and
		\inferrule[CSend]{c : \chtype{T} \\ e : T \in \{\btype,\ntype\} \\ \Sigma \vdash P}{\Sigma \vdash t \tags c!e . P}
		\and
		\inferrule[QSend]{c : \chtype{\qtype} \\ e \in \Sigma \\ \Sigma \setminus \{e\} \vdash P}{\Sigma \vdash t \tags c!e . P}
		\and
		\inferrule[Ite]{e : \btype \\ \Sigma \vdash P \\ \Sigma \vdash Q}{\Sigma \vdash \ite{e}{P}{Q}}
		\and
		\inferrule[Sum]{\Sigma \vdash P \\ \Sigma \vdash Q}{\Sigma \vdash P + Q}
		\and
		\inferrule[Par]{\Sigma_1 \cap \Sigma_2 = \emptyset \\ \Sigma_1 \vdash P \\ \Sigma_2 \vdash Q}{\Sigma_1 \cup \Sigma_2 \vdash P \parallel Q}
		\and
		\inferrule[Restr]{\Sigma \vdash P}{\Sigma \vdash P \setminus c}
	\end{mathpar}
	\caption{Typing rules for lqCCS}
	\label{fig:typingrules}
\end{figure}

We assume a denumerable set $\tagset$ ranged over by $t$, $t'$, $t_0$, $t_1$, etc.
The syntax of tagged lqCCS processes is obtained by augmenting the actions in \autoref{def:lqccs_syntax} with syntactic tags.
\begin{align*}
	\alpha &
	\Coloneqq
	t \tags \tau
	\mid (t, t') \tags \tau
	\mid t \tags c?x
	\mid t \tags c!e
	\mid t \tags \E(\tilde{e})
	\mid t \tags \meas{\tilde{e}}{x}
\end{align*}
A symbol $\E$ denotes a trace-preserving superoperator on
$\qubits{n}$ for some $n > 0$, and we write $\E : \opset(n)$ to
indicate that $\E$ is a superoperator of arity $n$.
A symbol $\M$ denotes a measurement $\{M_0,\ldots,M_{k-1}\}$ with $k$ different outcomes;
we write $\M : \measset(n)$ to indicate that each $M_i$ acts on $n$ qubits, and denote $|\M|$ the cardinality $k$ of $\M$.
Concerning the tags, for the silent action $\tau$ we consider two possibilities.
This distinction is useful for using $\tau$ actions to write an abstract specification of a concrete protocol: $t \tags \tau$ models an action with tag $t$, $(t, t')\tags\tau$ models a synchronization.
Tags are intended to explicitly name the non-deterministic choices that a scheduler can perform when running a process.
In the following, we assume that processes come with their associated tags.

The visibility of qubits is enforced explicitly through the linear type system in~\autoref{fig:typingrules}.
The typing judgment $\Sigma \vdash P$ indicates that the process $P$ is well-typed under the usage of the set of qubits $\Sigma \subseteq \qtype$.
Given a set of terms $E$, we write
$\tilde{E}$ for the set of all tuples $\tilde{e}$ composed of all and only the elements in $E$ 
(in any possible order without repetitions).
Atomic values are typed by their domain set:
$\ntype$ for natural numbers, $\btype$ for boolean values, and $\qtype$ for qubit names.
Similarly, we assume that variables come associated with the type of their possible values.
Channels $c \in \chanset$ are typed as $\chtype{\ntype}$, $\chtype{\btype}$ and $\chtype{\qtype}$, according to the values that they can propagate.
Note that processes composed in parallel must use two disjoint sets of qubits, thus enforcing single ownership.

Since the type of a process is unique, we will call $\Sigma_P$ the only context which types $P$ (when such a type exists).
\begin{restatable}{theorem}{uniquetype}
	Whenever $\Sigma \vdash P$ and $\Sigma' \vdash P$ then $\Sigma = \Sigma'$
\end{restatable}

\begin{example}\label{ex:quantumlottery}
	Consider a quantum lottery $\proc{QL}  = \proc{Pr} \parallel \proc{An} $ formed by processes $\proc{Pr}$, which prepares a qubit used as a source of randomness, and $\proc{An}$, which receives it, measures it, and  announces the winner between Alice and Bob.
	\begin{align*}
		\proc{Pr} & = (t_1\tags X(q).t_3\tags \text{c}!q.\nil) + (t_2\tags H(q).t_3\tags \text{c}!q.\nil)                                  \\
		\proc{An} & = t_4\tags \text{c}?x.t_4\tags \measstd{x}{y}.\ite{y = 0}{t_5\tags \text{a}!1.\nil_{x}}{ t_6\tags \text{b}!1.\nil_{x}}
	\end{align*}
	Intuitively, $\proc{Pr}$ can prepare and send a qubit by applying either $X$ or $H$ to its qubit $q$, and $\proc{An}$ announces that Alice wins with $a!1$ if the received qubit is found in state $\kz$, or that Bob wins with $b!1$ if it is found in state $\ko$.
	The unique typing of $\proc{QL}$ is $\{q\} \vdash
		\proc{QL}$, with $a, b : \chtype{\ntype}$, $c : \chtype{\qtype}$, $y : \ntype$,
	and $q, x : \qtype$. \exqed
\end{example}

\begin{figure}[t]
	\begin{mathpar}
		\inferrule[Tau]{}{\conf{\rho, t \tags \tau . P} \moveto{t} \sconf{\rho, P}}
		\and
		\inferrule[TauSynch]{}{\conf{\rho, (t, t') \tags \tau . P} \moveto{(t, t')} \sconf{\rho, P}}
		\and
		\inferrule[Sop]{}{\conf{\rho, t \tags \E(\tilde{q}) . P} \moveto{t} \sconf{\E^{\tilde{q}}(\rho), P}}
		\and
		\inferrule[Meas]
		{}
		{\conf{\rho, t \tags \meas{\tilde{q}}{x}.P}
			\moveto{t}
			\osum_{m = 0}^{|\mathbb{M}|-1}\delem{\tr(M_m^{\tilde{q}}(\rho))}
			{\conf{\frac{M_m^{\tilde{q}}(\rho)}{\tr(M_m^{\tilde{q}}(\rho))}, P[\sfrac{m}{x}]}{}}
		}
		\and
		\inferrule[Send]{e \Downarrow v}{\conf{\rho, t \tags c!e . P} \moveto[c!v]{t} \sconf{\rho, P}}
		\and
		\inferrule[Recv]{c : \chtype{\qtype} \Rightarrow v \in \Sigma_\rho \setminus \Sigma_P}{\conf{\rho, t \tags c?x.P} \moveto[c?v]{t} \sconf{\rho, P[\sfrac{v}{x}]}}
		\and
		\inferrule[IteT]
		{  e \Downarrow \true \\
			\conf{\rho, P} \moveto[\mu]{s} \Delta}
		{\conf{\rho, \ite{e}{P}{Q}} \moveto[\mu]{s} \Delta}
		\and
		\inferrule[IteF]{e \Downarrow \false \\ \conf{\rho, Q} \moveto[\mu]{s} \Delta}{\conf{\rho, \ite{e}{P}{Q}} \moveto[\mu]{s} \Delta}
		\and
		\inferrule[SumL]{\conf{\rho, P} \moveto[\mu]{s} \Delta}{\conf{\rho, P + Q} \moveto[\mu]{s} \Delta}
		\and
		\inferrule[SumR]{\conf{\rho, Q} \moveto[\mu]{s} \Delta}{\conf{\rho, P + Q} \moveto[\mu]{s} \Delta}
		\and
		\inferrule[ParL]{\conf{\rho, P} \moveto[\mu]{s} \Delta \\ \mu \not \in \{ c?v \mid  v \in \Sigma_Q \}}{\conf{\rho, P \parallel Q} \moveto[\mu]{s} \Delta \parallel Q}
		\and
		\inferrule[ParR]{\conf{\rho, Q} \moveto[\mu]{s} \Delta \\ \mu \not \in \{ c?v \mid  v \in \Sigma_P \}}{\conf{\rho, P \parallel Q} \moveto[\mu]{s} P \parallel \Delta}
		\and
		\inferrule[SynchL]{\conf{\rho, P} \moveto[c!v]{t} \sconf{\rho, P'} \\ \conf{\rho, Q} \moveto[c?v]{t'} \sconf{\rho, Q'}}{\conf{\rho, P \parallel Q} \moveto{(t,t')} \sconf{\rho, P' \parallel Q'}}
		\and
		\inferrule[SynchR]{\conf{\rho, P} \moveto[c?v]{t} \sconf{\rho, P'} \\ \conf{\rho, Q} \moveto[c!v]{t'} \sconf{\rho, Q'}}{\conf{\rho, P \parallel Q} \moveto{(t,t')} \sconf{\rho, P' \parallel Q'}}
		\and
		\inferrule[Restr]{\conf{\rho, P} \moveto[\mu]{s} \Delta \\ \mu \not \in \{c!v, c?v \mid v\}}{\conf{\rho, P \setminus c} \moveto[\mu]{s} \Delta}
	\end{mathpar}
	\caption{Rules of lqCCS semantics.}
	\label{semantics}
\end{figure}

\subsection{Operational Semantics}\label{sec:opsemantics}

We describe a labelled semantics for lqCCS in terms of \emph{configurations} $\conf{\rho, P} \in \confset$, each composed by a global quantum state and a tagged lqCCS process.
Given a set $\Sigma = \{q_1,\ldots,q_n\} \subseteq \qtype$ and its associated Hilbert space $\hilb_\Sigma = \qubits{n}$, a global quantum state $\rho$ is a density operator in $\dm{\hilb_\Sigma}$.
The type system is extended to configurations and distributions by considering the qubits of the underlying quantum state.
Let $\Sigma_\rho$ be the set of qubits appearing in $\rho$.
\begin{definition}[Configuration Typing]
	Let $\conf{\rho, P} \in \confset$, with $P$ typed and $\Sigma_P \subseteq \Sigma_\rho$, then
	we let $(\Sigma_\rho, \Sigma_P) \vdash \conf{\rho, P}$.

	Let $\Delta \in \sdist{\confset}$, then we let $(\Sigma, \Sigma') \vdash \Delta$ if $(\Sigma, \Sigma') \vdash
		\mathcal{C}$ for any $\mathcal{C} \in \supp{\Delta}$.
\end{definition}
Hereafter, we restrict ourselves to well-typed distributions, and denote with $\Sigma_{\overline{P}}$ the set $\Sigma_\rho \setminus \Sigma_P$, i.e.\ the qubits in the quantum state that are not owned by the process and thus freely accessible by the environment.

Our semantics is decorated with \emph{syntactic schedulers}, which resolve non-determinism by choosing a tag.
The syntax of schedulers $s \in \choiceset$ for tagged lqCCS is defined as follows
\[
	s \Coloneqq t \mid (t, t)
\]
A scheduler can select an action with a tag $t$, or a synchronization with a pair $(t_1, t_2)$.
Intuitively, tags represent available visible options, upon which the scheduler can choose (possibly using a pair of them for synchronization).
We call $\tagset(P)$ and $\tagset(s)$ the sets of tags that syntactically appear in the process $P$ and in the scheduler $s$, respectively;
we let $\tagset(\Delta) = \bigcup_{\conf{\rho, P} \in \supp{\Delta}}\tagset(P)$.

We assume a set $\actset$ of actions, containing $\tau$, $c!v$ and $c?v$ for any channel $c$ and value $v$.
The semantics of lqCCS is a probabilistic LTS (pLTS), i.e. a triple $(\confset, Act, \longrightarrow)$, with $\mathord{\longrightarrow} \subseteq \confset \times \actset \times \choiceset \times \dist{\confset}$.
A transition $(\conf{\rho, P}, \mu, s, \Delta) \in \mathord{\longrightarrow}$ is also denoted as $\conf{\rho, P} \moveto[\mu]{s} \Delta$.

The transition relation $\longrightarrow$ is the smallest relation over configurations with closed processes that satisfies the rules in~\autoref{semantics}.
As expected, tagged actions require a matching tag on the transition.
Expressions $e$ are evaluated through a big step semantics $e \Downarrow v$ with $v$ a value, i.e.\ either $n \in \ntype$, $b \in \btype$, or  $q \in \qtype$.
We restrict to arithmetic and logical operations, and therefore omit the rules and assume that free variables are not evaluated.
Note that the values of the qubits can only be observed through measurements, which alter such values and have a probabilistic outcome.
In particular, the expression $q = q'$ just compares two qubit names, and not their values.
In rule $\textsc{Sop}$, the superoperator $\E$ is applied to the qubits in $\tilde{q}$. Since $\tilde{q}$ can be smaller than the whole $\Sigma_\rho$, we define $\E^{\tilde{q}}$ as the superoperator that acts on the whole $\rho$ but ``ignores'' the qubits outside $\tilde{q}$.
More precisely, $\E^{\tilde{q}}$ is obtained by composing:
$(i)$ a suitable set of SWAP unitaries to bring the qubits $\tilde{q}$ in the first positions;
$(ii)$ the tensor product of the superoperator $\E$ with the identity on
untouched qubits on the right; and
$(iii)$ the inverse of the SWAP operators of point $(i)$ to recover the original order of qubits~\cite{lalireRelationsQuantumProcesses2006}.
In rule $\textsc{Meas}$, given a measurement $\M = \{M_m\}$, for each $m$, $\M_m$ stands for the trace non-increasing superoperator such that $\M_m(\sigma) = M_m \sigma M_m^\dagger$, and $\M_m^{\tilde{q}}$ is defined as before.
In the rules $\textsc{ParL}$ and $\textsc{ParR}$ the parallel composition of a distribution $\Delta = \osum_{i \in I} \delem{p_i}{\conf{\rho_i, P_i}}$ and a process $Q$ is defined as $\osum_{i \in I} \delem{p_i}{\conf{\rho_i, P_i \parallel Q}}$, and similarly for the restriction $\Delta \setminus c$.
Notice that in $\textsc{ParL}$ we require that $P$ does not receive a qubit that is already owned by $Q$, and symmetrically in $\textsc{ParR}$.

\begin{example}\label{ex:quantumlottery2}
	The pLTS semantics of the quantum lottery $\proc{QL}$ from~\autoref{ex:quantumlottery} on quantum state $\kbz$ is in~\autoref{fig:ex-ql}, where
	$\proc{An}'$ stands for $t_4\tags\measstd{x}{y}.\proc{An}''$, and $\proc{An}''$ is $\ite{y = 0}{t_5\tags a!1.\nil_x}{t_6\tags b!1.\nil_x}$.
	To simplify the presentation, we draw $\xlongrightarrow{s\tags\mu}$ instead of $\moveto[\mu]{s}$.
	A scheduler for $\proc{QL}$ must decide $(i)$ which unitary
	$\proc{Pr}$ applies to the qubit $q$, either $X$ (i.e. $s = t_1$) or $H$ (i.e. $s = t_2$);
	and $(ii)$ whether the qubit is sent to $\proc{An}$ or to the external environment since the channel is not restricted, i.e. the available moves are with label $\tau$ and scheduler $(t_3, t_4)$, or with label $c!q$ and $s = t_3$.
	Intuitively, if $\proc{Pr}$ chooses $X$, then Bob will win, otherwise either Alice or Bob will win with the same probability.
	Note that at the beginning, in configuration $\conf{\kbz, \proc{QL}}$, the choice $t_4$ cannot be taken because the quantum state has no qubit apart from those in $\Sigma_{\proc{QL}}$. \exqed
\end{example}

\begin{figure}[t]
	\centering
\begin{tikzpicture}[node distance=5mm]
  \node [conf] (c1) at (0, 0) {$\conf{\kbz, \proc{QL}}$};
  \node [psplit, below left=4mm and 8mm of c1] (p12) {};
  \node [psplit, below right=4mm and 8mm of c1] (p13) {};
  \node [conf, below= of p12] (c2) {$\conf{\kbo, t_3\tags c!q.\nil \parallel\proc{An}}$};
  \node [conf, below= of p13] (c3) {$\conf{\kbpl, t_3\tags c!q.\nil \parallel\proc{An}}$};
  \node [psplit, below=7mm of c2] (p2) { };l
  \node [psplit, below left=0mm and 12mm of c2] (p2') { };
  \node [psplit, below=7mm of c3] (p3) { };
  \node [psplit, below right=0mm and 12mm of c3] (p3') { };
  \node [conf, below left=3mm and -8mm of p2'] (c4) {$\conf{\kbo, \nil \parallel\proc{An}}$};
  \node [psplit, below=11mm of p2'] (p4') { };
  \node [conf, below=of p2] (c5) {$\conf{\kbo, \nil \parallel\proc{An}'}$};
  \node [conf, below=of p3] (c6) {$\conf{\kbpl, \nil \parallel\proc{An}'}$};
  \node [conf, below right=3mm and -8mm of p3'] (c7) {$\conf{\kbpl, \nil \parallel\proc{An}}$};
  \node [psplit, below=11mm of p3'] (p7') { };
  \node [psplit, below=of c5] (p5) { };
  \node [psplit, below=of c6] (p6) { };
  \node [conf, below left=6mm and -3mm of p5] (c8) {$\conf{\kbo, \nil \parallel\proc{An}''[\sfrac{1}{y}]}$};
  \node [conf, below right=6mm and -3mm of p6] (c9) {$\conf{\kbz, \nil \parallel\proc{An}''[\sfrac{0}{y}]}$};
  \node [psplit, below=of c8] (p8) { };
  \node [psplit, below=of c9] (p9) { };
  \node [conf, below=of p8] (c8f) {$\conf{\kbo, \nil \parallel \nil_q}$};
  \node [conf, below=of p9] (c9f) {$\conf{\kbz, \nil \parallel \nil_q}$};

  \draw [arrow] (c1) -- (p12) node[weight,above left] {$t_1\tags\tau$};
  \draw [arrow] (c1) -- (p13) node[weight,above right] {$t_2\tags\tau$};
  \draw [parrow] (p12) -- (c2) node[weight,above left=.3mm and 1mm] { };
  \draw [parrow] (p13) -- (c3) node[weight,above left=.3mm and 1mm] { };
  \draw [arrow] (c2) -- (p2) node[weight,left] {$(t_3, t_4)\tags\tau$};
  \draw [arrow] (c3) -- (p3) node[weight,right] {$(t_3, t_4)\tags\tau$};
  \draw [arrow] (c2.west) -- (p2') node[weight,above left] {$t_3\tags c!q$};
  \draw [arrow] (c3.east) -- (p3') node[weight,above right] {$t_3 \tags c!q$};
  \draw [parrow] (p2') -- (c4) node[weight,left] { };
  \draw [parrow] (p2) -- (c5) node[weight,above left] { };
  \draw [parrow] (p3) -- (c6) node[weight,right] { };
  \draw [parrow] (p3') -- (c7) node[weight,above right] { };
  \draw [arrow] (c5) -- (p5) node[weight,left] {$t_4\tags\tau$};
  \draw [arrow] (c6) -- (p6) node[weight,right] {$t_4\tags\tau$};
  \draw [parrow] (p5) -- (c8) node[weight,right] { };
  \draw [parrow] (p6) -- (c8) node[weight,above right] {$1/2$};
  \draw [parrow] (p6) -- (c9) node[weight,above right] {$1/2$};
  \draw [arrow] (c8) -- (p8) node[weight,left] {$t_6\tags b!1$};
  \draw [arrow] (c9) -- (p9) node[weight,right] {$t_5\tags a!1$};
  \draw [parrow] (p8) -- (c8f) node[weight,right] { };
  \draw [parrow] (p9) -- (c9f) node[weight,above right] { };
  \draw [arrow] (c4) -- (p4') node[weight,left] {$t_4\tags c?q$};
  \draw [arrow] (c7) -- (p7') node[weight,right] {$t_4\tags c?q$};
  \draw [parrow] (p4') -- (c5) node[weight,above right] { };
  \draw [parrow] (p7') -- (c6) node[weight,above right] { };
\end{tikzpicture}
	\caption{Semantics of the quantum lottery process.}
	\label{fig:ex-ql}
\end{figure}
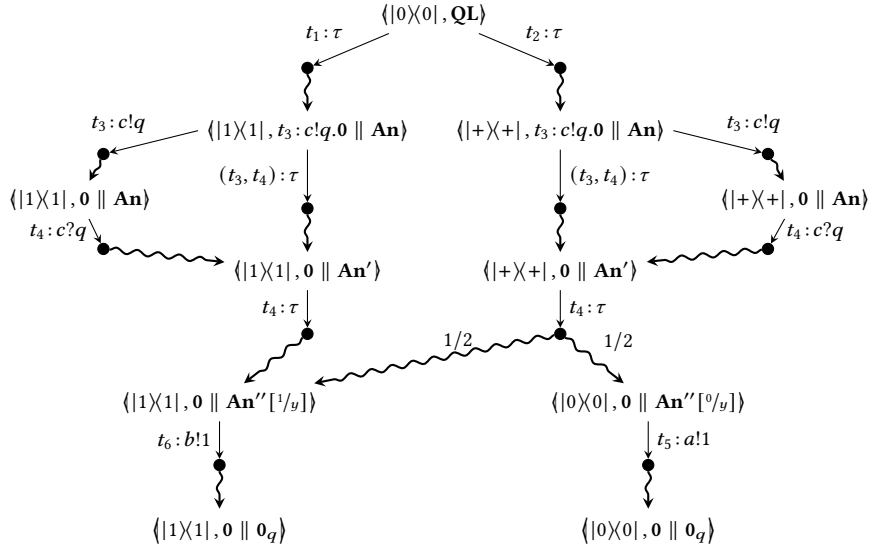

It is worth noting that the typing is preserved by $\tau$-transitions.
\begin{restatable}[Typing Preservation]{theorem}{typingpreservation}\label{thm:typepreservation}
	If $(\Sigma_\rho, \Sigma_P) \vdash \conf{\rho, P}$ and $\conf{\rho, P} \moveto{s} \Delta$ then
	$(\Sigma_\rho, \Sigma_P) \vdash \Delta$.
\end{restatable}

As hinted by the previous theorem, type preservation does not hold in general.
However, the updated typing context is uniquely determined by the transition label, and it is still a subset of the qubits available in the quantum state.
\begin{restatable}[Typing Quasi-Preservation]{theorem}{quasitypingpreservation}\label{thm:quasipreservation}
	Let $\Sigma \subseteq \qtype$ and $\mu \in \actset$. Then there exists $\Sigma'$ such that for all $(\Sigma'', \Sigma) \vdash \mathcal{C}$ and $\Delta \in \dist{\confset}$ with $\mathcal{C} \moveto[\mu]{s} \Delta$ it holds $(\Sigma'', \Sigma') \vdash \Delta$.
\end{restatable}


Following standard practice, we assume hereafter that processes are tagged in a deterministic way, meaning that the behaviour of configurations is determined by the choice of the scheduler.
For this, we take the standard approach of~\cite{chatzikokolakisMakingRandomChoices2007,chatzikokolakisBisimulationDemonicSchedulers2009}, and we extend it to deal with the value-passing features of lqCCS.
Intuitively, a deterministic process associates a different tag to all of its available non-deterministic choices, thus any given scheduler $s$ forces a unique action $\mu$ and a unique transition $\cmoveto[\mu]{s}$.
However, this is too restrictive for a labelled semantics where receptions are caused by the choices of other processes.
Our schedulers, indeed, solve only the \emph{internal} non-determinism (like $t\tags c!0.P + t'\tags c!1.Q$), and do not address the \emph{external} one, which comes from outside input (like $c?x.P$).
	Thus, only for input transitions, we allow a deterministic process $P$ to have multiple reductions for the same scheduler $s$, one for each value it can receive, like $\cmoveto[c?0]{s}$ and $\cmoveto[c?1]{s}$.

\begin{definition}[Deterministic Process]
	A process $P$ is \emph{deterministic} if for all $\rho \in \hilb_{\Sigma}$ with $\Sigma \supseteq \Sigma_P$, $s \in \tagset(P)$, $\mu, \mu' \in \actset$ and $\Delta, \Delta' \in \sdist{\confset}$, if $\conf{\rho, P} \cmoveto[\mu]{s} \Delta$ and $\conf{\rho, P} \cmoveto[\mu']{s} \Delta'$, then either $\mu = \mu'$ and $\Delta = \Delta'$, or $\mu \neq \mu'$ and there exists $c$, $v$ and $v'$ such that $\mu = c?v$ and $\mu' = c?v'$.
\end{definition}

\begin{assumption}\label{ass:det}
	All processes considered in the following are deterministic.
\end{assumption}

In quantum process calculi, it is useful to consider a pLTS as a simple LTS between distributions~\cite{fengtoward2015-1,dengbisimulations2018,ceragioliQuantumBisimilarityBarbs2024}.
	We will detail two different way to lift our scheduled pLTSs semantics: one preserves the scheduler information and the other one ignores it.
	We will explore the difference of these two systems in the next section, showing how consistent scheduling policies are crucial in the quantum setting.

First, we define a semantics in which $\Delta \dmoveto[\mu]{s} \Delta'$ means that a system in state $\Delta$ can transition to $\Delta'$ if all the elements of the support of $\Delta$ perform the move chosen by the scheduler $s$ with action $\mu$.
As an initial step, we add a ``deadlock'' rule, allowing configurations to evolve to the empty distribution if no legal move is available for the same label and scheduler:
Formally, we let $\operatorname{bot}(\cmoveto[]{}) \in \confset \times \actset \times \choiceset \times \sdist{\confset}$ as the least relation such that $\operatorname{bot}(\cmoveto[]{}) \supseteq \mathord{\cmoveto[]{}}$ and $\mathcal{C} \mathrel{\operatorname{bot}(\cmoveto[\mu]{s})} \epsilon$ if there is no $\Delta$ such that $\mathcal{C} \cmoveto[\mu]{s} \Delta$.
Notice that, by definition, $(\Sigma, \Sigma') \vdash \epsilon$ for any $(\Sigma, \Sigma')$, since the support of $\epsilon$ is empty.
%
%
%
Then, we lift the semantics to deal with distributions of configurations, defining $\dmoveto[]{}$ as the relation in $\sdist{\confset} \times \actset \times \choiceset \times \sdist{\confset}$ such that for all $\mu$ and $s$, it holds that
$\mathord{\dmoveto[\mu]{s}} = \mathord{\lift(\operatorname{bot}(\cmoveto[\mu]{s}))}$.
The obtained transition relation is linear and left-decomposable, by \autoref{lem:llift_decomp}.


Due to~\autoref{ass:det}, the evolution of distributions under a given action is uniquely determined by the scheduler.
\begin{restatable}{proposition}{determdistr}\label{thm:determdistr}
	For any distribution $\Delta \in \sdist{\confset}$ composed of deterministic processes, $\mu \in \actset$, $s \in \choiceset$, and $\Theta, \Theta' \in \sdist{\confset}$, if $\Delta \dmoveto[\mu]{s} \Theta$ and $\Delta \dmoveto[\mu]{s} \Theta'$ then $\Theta = \Theta'$.
\end{restatable}

Finally, to discuss the impact of using schedulers, and their importance for correctly representing feasible choices of quantum-capable agents, we define an unscheduled 
lifted semantics for comparison.
Let $\dumoveto[]$ be the relation in
$\sdist{\confset} \times \actset \times \sdist{\confset}$ such that for each $\mu$ it holds that
$\mathord{\dumoveto[\mu]} = \lift{(\bigcup_{s}\operatorname{bot}(\xlongrightarrow{\mu}_s))}$.
Intuitively, $\dumoveto[\mu]$ allows to probabilistic compose moves regardless of whether the schedulers coincide.

\begin{figure}[t]
	\centering
\begin{tikzpicture}[node distance=5mm]
  \node [conf] (c1) at (0, 0) {$\conf{\kbz, \proc{QL}}$};
  \node [conf, below left=8mm and 3mm of c1] (c2) {$\conf{\kbo, t_3\tags c!q.\nil \parallel\proc{An}}$};
  \node [conf, below right=8mm and 3mm of c1] (c3) {$\conf{\kbpl, t_3\tags c!q.\nil \parallel\proc{An}}$};
  \node [conf, below left=2mm and 0mm of c2] (c4) {$\conf{\kbo, \nil \parallel\proc{An}}$};
  \node [conf, below=8mm of c2] (c5) {$\conf{\kbo, \nil \parallel\proc{An}'}$};
  \node [conf, below=8mm of c3] (c6) {$\conf{\kbpl, \nil \parallel\proc{An}'}$};
  \node [conf, below right=2mm and 0mm of c3] (c7) {$\conf{\kbpl, \nil \parallel\proc{An}}$};
  \node [conf, below=8mm of c5] (c8) {$\conf{\kbo, \nil \parallel\proc{An}''[\sfrac{1}{y}]}$};
  \node [conf, below=8mm of c6] (c9) {$\delem{\frac{1}{2}}{\conf{\kbz, \nil \parallel\proc{An}''[\sfrac{0}{y}]}} \oplus \delem{\frac{1}{2}}{\conf{\kbo, \nil \parallel\proc{An}''[\sfrac{1}{y}]}}$};
  \node [conf, below=8mm of c8] (c8f) {$\conf{\kbo, \nil \parallel \nil_q}$};
  \node [conf, below left=8mm and -22mm of c9] (c9f) {$\delem{\frac{1}{2}}{\conf{\kbz, \nil \parallel \nil_q}}$};
  \node [conf, below right=8mm and -22mm of c9] (c9e) {$\delem{\frac{1}{2}}{\conf{\kbo, \nil \parallel \nil_q}}$};

  \draw [darrow] (c1) -- (c2) node[weight,above left] {$t_1\tags\tau$};
  \draw [darrow] (c1) -- (c3) node[weight,above right] {$t_2\tags\tau$};
  \draw [darrow] (c2) -- (c5) node[weight,left] {$(t_3, t_4)\tags\tau$};
  \draw [darrow] (c3) -- (c6) node[weight,right] {$(t_3, t_4)\tags\tau$};
  \draw [darrow] (c2.west) -- (c4) node[weight,above left] {$t_3\tags c!q$};
  \draw [darrow] (c3.east) -- (c7) node[weight,above right] {$t_3 \tags c!q$};
  \draw [darrow] (c5) -- (c8) node[weight,left] {$t_4\tags\tau$};
  \draw [darrow] (c6) -- (c9) node[weight,right] {$t_4\tags\tau$};
  \draw [darrow] (c8) -- (c8f) node[weight,left] {$t_6\tags b!1$};
  \draw [darrow] (c9) -- (c9f) node[weight,right] {$t_5\tags a!1$};
  \draw [darrow] (c9) -- (c9e) node[weight,left] {$t_6\tags b!1$};
  \draw [darrow] (c4) -- (c5) node[weight,below left] {$t_4\tags c?q$};
  \draw [darrow] (c7) -- (c6) node[weight,below right] {$t_4\tags c?q$};
\end{tikzpicture}
	\caption{Lifted semantics of the quantum lottery process.}
	\label{fig:ex-ql-lifted}
\end{figure}
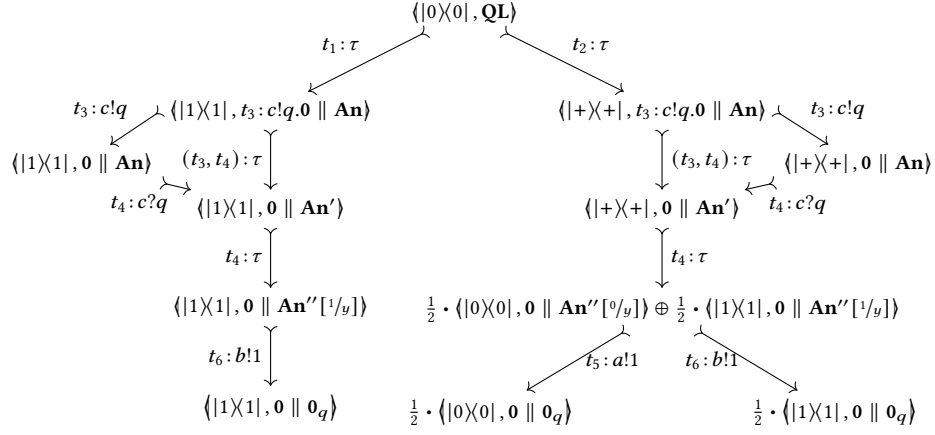

\begin{example}\label{ex:quantumlottery3}
	Recall \autoref{ex:quantumlottery2}.
	The \emph{scheduled} lifted semantics to distributions of configurations is in~\autoref{fig:ex-ql-lifted}.
	The application of the lifting operations is evident in the last step of the protocol of the right branch.
	After the measurement the semantics has reached the full distribution of configuration $\delem{\frac{1}{2}}{\conf{\kbz, \nil \parallel\proc{An}''[\sfrac{0}{y}]}} \oplus \delem{\frac{1}{2}}{\conf{\kbo, \nil \parallel\proc{An}''[\sfrac{1}{y}]}}$, where, as expected, both Alice and Bob have half of the probability of being selected.
	Each element of the distribution has a continuation with a distinct tag and action, therefore by definition of $\operatorname{bot}(\blank)$ the other pairing can only move to the empty distribution $\epsilon$, e.g.\ 
	$\delem{\frac{1}{2}}{\conf{\kbz, \nil \parallel\proc{An}''[\sfrac{0}{y}]}} \dmoveto[b!1]{t_6} \epsilon$.
	Thus, by lifting, full distribution can evolve either in the distribution $\delem{\frac{1}{2}}{\conf{\kbz, \nil \parallel \nil_{q}}}$ with tag $t_5$ and action $a!1$ or to the distribution $\delem{\frac{1}{2}}{\conf{\kbo, \nil \parallel \nil_{q}}}$ with tag $t_6$ and action $b!1$. \exqed \anote{appunto per il futuro: bisognerebbe far vedere anche la semantica unscheduled}
\end{example}

\section{Contextual Equivalences for Quantum Systems}\label{sec:sat}

We now pursue a suitable behavioural equivalence for concurrent quantum systems.
We start by defining a saturated bisimilarity \emph{\`a la}~\cite{bonchigeneral2014}, a natural notion of contextual equivalence that relates systems when they are indistinguishable for any observer.
We show that the naive definition of saturated probabilistic bisimilarity is not consistent with the prescriptions of quantum theory.
Next, we show that schedulers guarantee adherence with the prescriptions of quantum theory, as contexts cannot choose their move based on the unknown states of unmeasured qubits, and we investigate some desired properties of the scheduled saturated bisimilarity.
The main drawback of our saturated approach is that the quantification over all possible observers make the relation difficult to compute.
Thus, in the next section, we will give an equivalent characterization in terms of labelled bisimulations.

\subsection{Saturated Bisimilarity}\label{sec:equivalences}

Before presenting saturated bisimilarity, we need a couple of definitions.
We start by introducing well-typed relations, i.e. those that only associate equally typed distributions of configurations.
Recall that the empty distribution $\epsilon$ is typed by any pair of sets of qubits $(\Sigma, \Sigma')$ with $\Sigma' \subseteq \Sigma$.
\begin{definition}[Well-typed Relation]\label{def:welltyped}
	A relation $\mathord{\rel} \subseteq \sdist{\confset} \times \sdist{\confset}$ is well-typed if $\Delta \rel \Theta$ implies $(\Sigma, \Sigma') \vdash \Delta$ and $(\Sigma, \Sigma') \vdash \Theta$ for some $\Sigma, \Sigma'$.
\end{definition}

In saturated bisimilarities, we call \emph{contexts} processes with a typed hole; they play the role of process-discriminating observers.
\begin{definition}[Context]\label{def:context}
	A context $B[\blank]_{\Sigma}$ is generated by the productions
	$B[\blank]_{\Sigma} \Coloneqq [\blank]_{\Sigma}\;|\;[\blank]_{\Sigma} \parallel P$ typed
	according to the rules in \autoref{fig:typingrules},
	extended as such:
	\begin{mathpar}
		\inferrule[Id]
		{}
		{\Sigma \vdash [\blank]_{\Sigma}}

		\and

		\inferrule[Hole]
		{\Sigma' \setminus \Sigma \vdash P \\ \Sigma \subseteq \Sigma'}
		{\Sigma' \vdash [\blank]_{\Sigma} \parallel P}.
	\end{mathpar}
\end{definition}
A context is applied to a process $P$ by replacing the hole with $P$.
Intuitively, a context $\Sigma' \vdash B[\blank]_{\Sigma}$ is a function that given a process $P$ returns a process $B[P]$ obtained by replacing $P$ for $[\blank]$, where $\Sigma$ is the typing context of the valid inputs and $\Sigma'$ the one of the outputs.
Note that a context can own some qubits and each qubit cannot be referred to in both $P$ and $B[\blank]$.
We apply $\Sigma' \vdash B[\blank]_{\Sigma}$ to configurations $(\Sigma_\rho, \Sigma_P) \vdash \conf{\rho, P}$ obtaining $(\Sigma_\rho, \Sigma') \vdash \conf{\rho, B[P]}$ when $\Sigma' \subseteq \Sigma_\rho$ and $\Sigma = \Sigma_P$, i.e.\ when the qubits referred by $B[\blank]$ are defined in $\rho$ and the process $P$ is as prescribed.
We write $B[\conf{\rho, P}]$ for $\conf{\rho, B[P]}$, and $B[\Delta]$ for the distribution obtained by applying $B[\blank]$ to the support of $\Delta$.
It is trivial to show that if $\Delta$ and $\Theta$ are typed by the same typing context, then $B[\Delta]$ is defined if and only if $B[\Theta]$ is defined.
In the following we let $\tagset(B[\blank])$ be $\tagset(P)$ if $B[\blank] = [\blank] \parallel P$ and $\emptyset$ if $B[\blank] = [\blank]$.

We formalize compatible contexts that can be applied to distributions as follows.
\begin{definition}[Compatible Context]
	A context $\Sigma' \vdash B[\blank]_{\Sigma}$ is compatible with respect to a distribution $\Delta \in \sdist{\confset}$ if
	$(\Sigma'', \Sigma) \vdash \Delta$ for some $\Sigma''$ such that $\Sigma' \subseteq \Sigma''$, and if $\tagset(\Delta) \cap \tagset(B[\blank]) = \emptyset$.
\end{definition}

The requirement for tags to be distinct between distributions and contexts ensures that distributions preserve their determinism when inserted into compatible contexts.
Intuitively, this guarantees that we are not introducing non-deterministic behaviour besides the one regulated by schedulers and by the external choices.
\begin{restatable}{lemma}{determpar}\label{lem:deterPar}
	For all deterministic $P, R \in \procset$, if $\tagset(P) \cap \tagset(R) = \emptyset$
	then $P \parallel R$ is deterministic.
\end{restatable}
\begin{proofsketch}
	By case analysis on the scheduler.
	The interesting case is the synchronization between $P$ and $R$.
	By \autoref{ass:det} there cannot be both a valid \textsc{SynchL} and \textsc{SynchR} transition otherwise there would be two moves with distinct labels with the same scheduler in both $P$ and $R$.
	Furthermore, reception and sending must have the same value in synchronizations, and since by \autoref{ass:det} there can only be one specific sending value, then also the reception must be compatible.
\end{proofsketch}


We are now ready to introduce our first equivalence.
A saturated bisimulation is a relation over distributions where related pairs must have the same type and mass, and must reduce in related distributions under every possible compatible context.
We start with unscheduled saturated bisimilarity, which is the standard definition for probabilistic systems, used e.g. in~\cite{fengprobabilistic2007}.

\begin{definition}[Unscheduled Saturated Bisimilarity]
	A well-typed relation $\mathord{\rel} \subseteq \sdist{\confset} \times \sdist{\confset}$ is a (unscheduled) \emph{saturated bisimulation} if $\Delta \rel \Theta$ implies
	$\mass{\Delta} = \mass{\Theta}$ and, for any context $B[\blank]$
	compatible with both $\Delta$ and $\Theta$, it holds that
	\begin{itemize}
		\item whenever $B[\Delta] \dumoveto \Delta'$, there exists $\Theta'$
		      such that $B[\Theta] \dumoveto \Theta'$ and $\Delta' \rel \Theta'$;
		\item whenever $B[\Theta] \dumoveto \Theta'$, there exists $\Delta'$
		      such that $B[\Delta] \dumoveto \Delta'$ and $\Delta' \rel \Theta'$.
	\end{itemize}

	Let \emph{unscheduled saturated bisimilarity}, denoted $\simus$, be the largest unscheduled saturated bisimulation.

\end{definition}

\begin{figure}[t]
	\centering
	\begin{tikzpicture}[node distance=5mm]

		%

		\node [psplit] at (0,0) (p1) {};
		\node [conf, above=1mm of p1] (pd1) {$B[\Delta_{01}]$};
		\node [conf, below left=3mm and 0mm of p1] (c12) {$\conf{\kbz, B[t_0\tags c!q]}$};
		\node [conf, below right=3mm and 0mm of p1] (c13) {$\conf{\kbo, B[t_0\tags c!q]}$};
		\node [psplit, below=of c12] (p12) {};
		\node [psplit, below=of c13] (p13) {};
		\node [conf, below=of p12] (c2) {$\confel_0$};
		\node [conf, below=of p13] (c3) {$\confel_1$};
		\node [psplit, below left=5mm and 3mm of c2] (p2) {};
		\node [psplit, below=of $(c2.south)!0.5!(c3.south)$] (p3) {};
		\node [psplit, below right=5mm and 3mm of c3] (p4) {};
		\node [conf, below left=3mm and 3mm of p2] (c4) {$\confel^0_0$};
		\node [conf, below right=3mm and 6mm of p2] (c5) {$\confel^0_i$};
		\node [conf, below left=3mm and 6mm of p4] (c6) {$\confel^1_{\uminus i}$};
		\node [conf, below right=3mm and 3mm of p4] (c7) {$\confel^1_1$};
		\node [psplit, below=of c4] (p5) {};
		\node [psplit, below=of c5] (p6) {};
		\node [conf, below=of p5] (c8) {$\circ$};
		\node [conf, below=of p6] (c9) {$\circ$};


		\draw [parrow] (p1) -- (c12) node[weight,above left=-.3mm and -1mm] {$1/2$};
		\draw [parrow] (p1) -- (c13) node[weight,above right=-.3mm and -1mm] {$1/2$};
		\draw [arrow] (c12) -- (p12) node[weight,left=1mm] { $(t_0,t_1)\tags\tau$ };
		\draw [arrow] (c13) -- (p13) node[weight,right=1mm] { $(t_0,t_1)\tags\tau$ };
		\draw [parrow] (p12) -- (c2) node[weight,above left=.3mm and 1mm] { };
		\draw [parrow] (p13) -- (c3) node[weight,above right=.3mm and 1mm] { };
		\draw [arrow] (c2) -- (p2) node[weight,above left] { $t_2\tags\tau$ };
		\draw [arrow] (c2) -- (p3) node[weight,above] { $t_3\tags\tau$ };
		\draw [arrow] (c3) -- (p3) node[weight,above] { $t_3\tags\tau$ };
		\draw [arrow] (c3) -- (p4) node[weight,above right] { $t_2\tags\tau$ };
		\draw [parrow] (p2) -- (c4) node[weight,above left=-1.5mm and 1mm] { };
		\draw [parrow] (p3) -- (c5) node[weight,above left=-1.5mm and 1mm] {$1/2$};
		\draw [parrow] (p3) -- (c6) node[weight,above right=-1.5mm and 1mm] {$1/2$};
		\draw [parrow] (p4) -- (c7) node[weight,above right=-1.5mm and 1mm] { };
		\draw [arrow] (c4) -- (p5) node[weight,left] { $t_4\tags\tau$ };
		\draw [arrow] (c5) -- (p6) node[weight,right] { $t_4\tags\tau$ };
		\draw [parrow] (p5) -- (c8) node[weight,above right=-1.5mm and 1mm] { };
		\draw [parrow] (p6) -- (c9) node[weight,above right=-1.5mm and 1mm] { };
	\end{tikzpicture}
	\hspace{.6cm}
	\begin{tikzpicture}[node distance=5mm]
		\node [psplit] at (0,0) (p1) {};
		\node [conf, above=1mm of p1] (pd1) {$B[\Delta_{\pm}]$};
		\node [conf, below left=3mm and 0mm of p1] (c12) {$\conf{\kbpl, B[t_0\tags c!q]}$};
		\node [conf, below right=3mm and 0mm of p1] (c13) {$\conf{\kbm, B[t_0\tags c!q]}$};
		\node [psplit, below=of c12] (p12) {};
		\node [psplit, below=of c13] (p13) {};
		\node [conf, below=of p12] (c2) {$\confel_+$};
		\node [conf, below=of p13] (c3) {$\confel_-$};
		\node [psplit, below=of c2] (p2) {};
		\node [psplit, below=of c3] (p3) {};
		\node [conf, below left=3mm and 6mm of p2] (c4) {$\confel^0_0$};
		\node [conf, below right=3mm and 6mm of p2] (c5) {$\confel^1_1$};
		\node [conf, below left=3mm and 6mm of p3] (c6) {$\confel^0_i$};
		\node [conf, below right=3mm and 6mm of p3] (c7) {$\confel^1_{\uminus i}$};
		\node [psplit, below=of c4] (p5) {};
		\node [psplit, below=of c6] (p6) {};
		\node [conf, below=of p5] (c8) {$\circ$};
		\node [conf, below=of p6] (c9) {$\circ$};

		\draw [parrow] (p1) -- (c12) node[weight,above left=-.3mm and -1mm] {$1/2$};
		\draw [parrow] (p1) -- (c13) node[weight,above right=-.3mm and -1mm] {$1/2$};
		\draw [arrow] (c12) -- (p12) node[weight, left=1mm] { $(t_0,t_1)\tags\tau$ };
		\draw [arrow] (c13) -- (p13) node[weight,right=1mm] { $(t_0,t_1)\tags\tau$ };
		\draw [parrow] (p12) -- (c2) node[weight,above left=.3mm and 1mm] { };
		\draw [parrow] (p13) -- (c3) node[weight,above right=.3mm and 1mm] { };
		\draw [arrow] (c2) -- (p2) node[weight, left] { $t_2\tags\tau$ };
		\draw [arrow] (c2) -- (p3) node[weight,above  left=1mm and 1mm] { $t_3\tags\tau$ };
		\draw [arrow] (c3) -- (p2) node[weight,above right=1mm and 1mm] { $t_3\tags\tau$ };
		\draw [arrow] (c3) -- (p3) node[weight,right] { $t_2\tags\tau$ };
		\draw [parrow] (p2) -- (c4) node[weight,above left=-1.5mm and 1mm] {$1/2$};
		\draw [parrow] (p2) -- (c5) node[weight,above right=-1.5mm and 1mm] {$1/2$};
		\draw [parrow] (p3) -- (c6) node[weight,above left=-1.5mm and 1mm] {$1/2$};
		\draw [parrow] (p3) -- (c7) node[weight,above right=-1.5mm and 1mm] {$1/2$};
		\draw [arrow] (c4) -- (p5) node[weight, left] { $t_4\tags\tau$ };
		\draw [arrow] (c6) -- (p6) node[weight,right] { $t_4\tags\tau$ };
		\draw [parrow] (p5) -- (c8) node[weight,above right=-1.5mm and 1mm] { };
		\draw [parrow] (p6) -- (c9) node[weight,above right=-1.5mm and 1mm] { };
	\end{tikzpicture}

	\caption{Indistinguishable qubit sources $\Delta_{01}$ and $\Delta_{\pm}$ are distinguished by the context $B[\blank]$.}
	\label{fig:ex-zopm}
\end{figure}
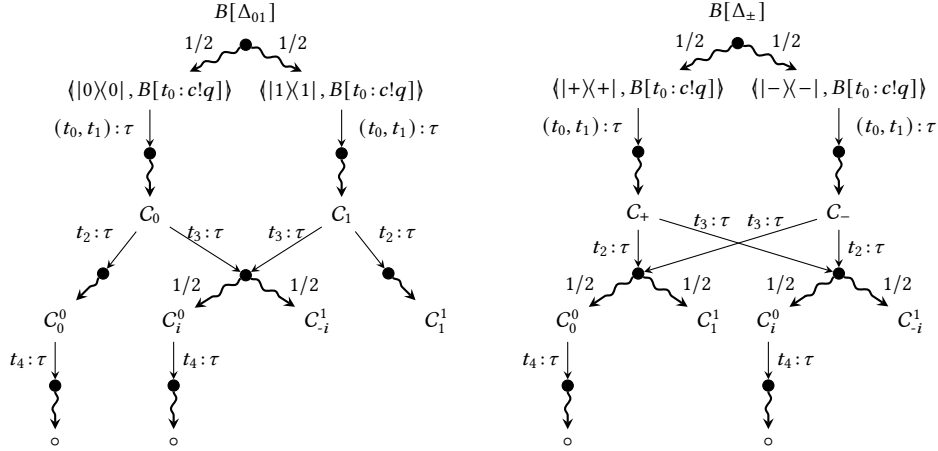

We now compare the bisimilarity above with the prescriptions of quantum theory, showing that it is not consistent with~\autoref{thm:qind}.

We consider a specific counterexample of processes that should be indistinguishable according to quantum theory:
a pair of non-biased random qubit sources, the first sending a qubit in state $\kz$ or
$\ko$, the second in state $\kpl$ or $\km$.
Quantum theory prescribes that such two sources cannot be distinguished by any observer, as in both cases the received qubit is in state $\frac{1}{2}I$, and therefore behaves the same (see~\cite{nielsenquantum2010}).
One expects the lqCCS encoding of these sources to be bisimilar, however, this is not the case for $\simus$.
\begin{example}\label{ex:broken-nondet}
	The lqCCS encodings of these sources follows:
	\begin{gather*}
		\Delta_{01}  = \delem{\frac{1}{2}}{\conf{\kbz, t_0\tags c!q}} \oplus \delem{\frac{1}{2}}{\conf{\kbo, t_0\tags c!q}}
		\qquad\text{ and }\qquad
		\Delta_{\pm} = \delem{\frac{1}{2}}{\conf{\kbpl, t_0\tags c!q}} \oplus \delem{\frac{1}{2}}{\conf{\kbm, t_0\tags c!q}}
	\end{gather*}
	The unscheduled bisimilarity erroneously discriminates the two sources.
	To see that $\Delta_{01} \nsimus \Delta_{\pm}$, take $B[\blank] = [\blank] \parallel t_1\tags c?x.(P+Q)$ where
	\begin{gather*}
		P =  t_2\tags \measstd{x}{y} .R
		\qquad\text{ and }\qquad
		Q =  t_3\tags \meashadi{x}{y} .R
		\qquad\text{ with }\qquad
		R =  (\ite{y=0}{t_4\tags \tau}{\nil}) \parallel \nil_{x}
	\end{gather*}
	Notice that $P$ and $Q$ perform different measurements and
	make the outcome observable by enabling a $\tau$-transition only when $y = 0$.
	In a nutshell, we distinguish the sources since a series of reductions
	$B[\Delta_{01}] \dumoveto \Delta_{01}' \dumoveto \Delta_{01}'' \dumoveto \Theta_{01}$ exists where $\mass{\Theta_{01}} = 3/4$, while
	all reductions
	$B[\Delta_{\pm}] \dumoveto \Delta_{\pm}' \dumoveto \Delta_{\pm}'' \dumoveto \Theta_{\pm}$ are such that
	$\mass{\Theta_{\pm}} = 1/2$.

	Formally, let $\confel_{\psi}$ be $\conf{\kbp, P+Q}[\sfrac{q}{x}]$
	and $\confel^n_{\psi}$ be $\conf{\kbp, R[\sfrac{n}{y}]}$ for $n \in \{0,1\}$, then
	\autoref{fig:ex-zopm} shows the $\tau$-transitions of the two systems: $B[\Delta_{01}]$ reduces to $\Delta_{01}' = \delem{\frac{1}{2}}{\confel_0} \oplus \delem{\frac{1}{2}}{\confel_1}$, and $B[\Delta_{\pm}]$ can only reduce to $\Delta_{\pm}' = \delem{\frac{1}{2}}{\confel_+} \oplus \delem{\frac{1}{2}}{\confel_-}$ to match this move.
	As seen in figure, each $\confel_\psi$ has available reductions both with $t_2$ and $t_3$.
	Consider the following ones:
	\begin{align*}
		{\confel_0} & \moveto{t_2} \sconf{\kbz, R[\sfrac{0}{y}]} = \delem{1}{\confel_0^0}
		,                                                                                    \\
		{\confel_1} & \moveto{t_3} \delem{\frac{1}{2}}{\conf{\kbip, R[\sfrac{0}{y}]}} \oplus
		\delem{\frac{1}{2}}{\conf{\kbim, R[\sfrac{1}{y}]}} = \delem{\frac{1}{2}}{\confel_i^0} \oplus \delem{\frac{1}{2}}{\confel_{\uminus i}^1}
	\end{align*}
	In the unscheduled semantics, $\Delta_{01}'$ can mix the two choices and reduce $\confel_0$ with $t_2$ and $\confel_1$ with $t_3$, formally,
	$\Delta_{01}' \dumoveto \delem{\frac{1}{2}}{\confel^0_0} \oplus \delem{\frac{1}{4}}{\confel^0_i} \oplus \delem{\frac{1}{4}}{\confel^1_{\uminus i}} = \Delta_{01}''$.
	This move corresponds to taking the left reductions both probabilistic branches of the pLTS of~\autoref{fig:ex-zopm}.
	However, choosing these tags requires knowing if we are in $C_0$ or $C_1$, which should be only possible by measuring the received qubit, which is not the case.
	Finally, note that $\Delta_{01}'' \dumoveto \Theta_{01}$ with $\mass{\Theta_{01}} = 3/4$, because $\confel_\psi^n$ can move (with $t_4$) only if $n = 0$.

	Consider now $\Delta_{\pm}'$.
	The following moves are all those available for $\confel_{+}$:
	\begin{align*}
		\confel_{+} & \moveto{t_2} \delem{\frac{1}{2}}{\conf{\kbz, R[\sfrac{0}{y}]}} \oplus
		\delem{\frac{1}{2}}{\conf{\kbo, R[\sfrac{1}{y}]}} = \delem{\frac{1}{2}}{\confel_0^0} \oplus \delem{\frac{1}{2}}{\confel_{1}^1} \\
		\confel_{+} & \moveto{t_3} \delem{\frac{1}{2}}{\conf{\kbip, R[\sfrac{0}{y}]}} \oplus
		\delem{\frac{1}{2}}{\conf{\kbim, R[\sfrac{1}{y}]}} = \delem{\frac{1}{2}}{\confel_i^0} \oplus \delem{\frac{1}{2}}{\confel_{\uminus i}^1},
	\end{align*}
	Noticeably, the available moves of $\confel_{-}$ coincide.

	It is easy to check that $\Delta_{\pm}'$ cannot replicate the behaviour of $\Delta_{01}'$: for any
	choice of $t_2$ and $t_3$ it will reduce to a $\Delta_{\pm}''$ where only half of the configurations in the support can move, i.e. such that  $\Delta_{\pm}'' \dumoveto \Theta_{\pm}$ only if $\mass{\Theta_{\pm}} = 1/2$. \exqed
\end{example}

\autoref{ex:broken-nondet} is paradigmatic, where different mixtures of quantum states are discriminated because the moves are chosen
according to the value of some received qubit, which in theory should be unknown.
Our schedulers do not allow processes to replicate this behaviour, forcing distributions to make a reasonable choice based on classically determined tags.
On a practical account, this feature is relevant for real-world applications, as several quantum protocols rely on the indistinguishability of these sources (or similar ones).

The counterexample highlighted by the example above is not typical of our specific formulation.
Indeed, the two sources are commonly discriminated by 
the earliest probabilistic bisimilarities
 proposed for quantum protocols, e.g., the ones of~\cite{fengprobabilistic2007,dengopen2012}.
This problem is solved by more recent proposals~\cite{fengtoward2015-1,dengbisimulations2018}, which correctly identify the two processes.
However, similar cases remain unsolved in literature, calling for a principled approach like the one we propose.
As an example, we present a simple protocol that does not feature communication, yet it permits
unscheduled bisimilarity to instantly distinguish which operation has been locally performed, achieving physically implausible faster-than-light communication.
Furthermore, this example proves that schedulers are needed also when no quantum communication is involved.

\begin{example}\label{ex:broken-ftl}
	The protocol is set between two actors, Alice and Bob, who share an entangled pair of qubits ($q_1,q_2$) in the entangled state $\kphip$.
	Before moving apart, they synchronize their clocks and agree that at some given time Alice will measure her qubit in one of two agreed bases, and Bob will immediately try to guess which measurement Alice performed.
	If the two succeed in finding a physically realizable procedure for Bob to correctly guess the measurement, then Alice can instantly communicate a bit to Bob, regardless of the distance separating them.
	Such a procedure should be impossible according to the commonly understood laws of physics.

	Assume Alice and Bob agree that Alice will measure her qubit in either the computational or diagonal basis, thus implicitly sending respectively $0$ or $1$.
	We will model the actors in lqCCS as follows
	\begin{align*}
		\proc{A}_{01}  & \coloneqq t_0 \tags \measstd{q_1}{x}.\nil_{q_1}                                                                          \\
		\proc{A}_{\pm} & \coloneqq t_0 \tags \meashad{q_1}{x}.\nil_{q_1}                                                                          \\
		\proc{B}       & \coloneqq (t_1 \tags \measstd{q_2}{y}.t_1 \tags c!y.\nil_{q_2}) + (t_2 \tags \meashadi{q_2}{y}.t_2 \tags c!y.\nil_{q_2})
	\end{align*}

	Superluminal communication is possible if Bob is capable of distinguishing the two Alice processes, that is, if we have that $\conf{\kbphip, \proc{A}_{01} \parallel \proc{B}} \nsimus \conf{\kbphip,\proc{A}_{\pm} \parallel \proc{B}}$.

	Assume Alice performs her measure, as expected
	\begin{align*}
		\sconf{\kbphip, \proc{A}_{01} \parallel \proc{B}} & \dumoveto \Delta = \delem{\frac{1}{2}}{\conf{\ketbra{00}, \nil_{q_1} \parallel \proc{B}}} \oplus \delem{\frac{1}{2}}{\conf{\ketbra{11}, \nil_{q_1} \parallel \proc{B}}} \\
		\sconf{\kbphip,\proc{A}_{\pm} \parallel \proc{B}} & \dumoveto \Theta = \delem{\frac{1}{2}}{\conf{\ketbra{++}, \nil_{q_1} \parallel \proc{B}}} \oplus \delem{\frac{1}{2}}{\conf{\ketbra{--}, \nil_{q_1} \parallel \proc{B}}}
	\end{align*}

	In the unscheduled semantics, the distribution $\Delta$ can evolve by selecting distinct tags for each element of its distribution, for example by performing the measurement on the standard basis for the distribution with state $\ketbra{00}$ and the measurement in the imaginary basis for the distribution with state $\ketbra{11}$.
	\begin{align*}
		\delem{\frac{1}{2}}{\conf{\ketbra{00}, \nil_{q_1} \parallel \proc{B}}} \oplus \delem{\frac{1}{2}}{\conf{\ketbra{11}, \nil_{q_1} \parallel \proc{B}}} & \dumoveto \Delta' = \delem{\frac{1}{2}}{C^0_{\ket{00}}} \oplus \delem{\frac{1}{4}}{C^0_{\ketbra{1i}}} \oplus \delem{\frac{1}{4}}{C^1_{\ketbra{11}}}
	\end{align*}
	where $C^i_{\kp}$ is the configuration $\conf{\kbp, \nil_{q_1} \parallel c!i.\nil{q_2}}$.
	Take the context $B[\blank] = [\blank] \parallel c?x.\ite{x = 0}{\tau.\nil}{\nil}$, that immediately goes in deadlock only if the output on the channel $c$ is $1$.
	After preforming the synchronization, only the distribution starting from $C^0$ can perform another move, therefore the total mass of the distribution is $3/4$.

	This particular execution cannot be replicated by $\Theta$ for any possible combination of tags.
	\begin{align*}
		\delem{\frac{1}{2}}{\conf{\ketbra{++}, \nil_{q_1} \parallel \proc{B}}} \oplus \delem{\frac{1}{2}}{\conf{\ketbra{--}, \nil_{q_1} \parallel \proc{B}}} & \dumoveto \delem{\frac{1}{4}}{C^0_{\ket{+0}}} \oplus \delem{\frac{1}{4}}{C^1_{\ket{+1}}} \oplus \delem{\frac{1}{4}}{C^0_{\ketbra{-0}}} \oplus \delem{\frac{1}{4}}{C^1_{\ketbra{-1}}}                        \\
		\delem{\frac{1}{2}}{\conf{\ketbra{++}, \nil_{q_1} \parallel \proc{B}}} \oplus \delem{\frac{1}{2}}{\conf{\ketbra{--}, \nil_{q_1} \parallel \proc{B}}} & \dumoveto \delem{\frac{1}{4}}{C^0_{\ket{+0}}} \oplus \delem{\frac{1}{4}}{C^1_{\ket{+1}}} \oplus \delem{\frac{1}{4}}{C^0_{\ketbra{\uminus i}}} \oplus \delem{\frac{1}{4}}{C^1_{\ketbra{-\uminus i}}}         \\
		\delem{\frac{1}{2}}{\conf{\ketbra{++}, \nil_{q_1} \parallel \proc{B}}} \oplus \delem{\frac{1}{2}}{\conf{\ketbra{--}, \nil_{q_1} \parallel \proc{B}}} & \dumoveto \delem{\frac{1}{4}}{C^0_{\ket{+i}}} \oplus \delem{\frac{1}{4}}{C^1_{\ket{+\uminus i}}} \oplus \delem{\frac{1}{4}}{C^0_{\ketbra{-0}}} \oplus \delem{\frac{1}{4}}{C^1_{\ketbra{-1}}}                \\
		\delem{\frac{1}{2}}{\conf{\ketbra{++}, \nil_{q_1} \parallel \proc{B}}} \oplus \delem{\frac{1}{2}}{\conf{\ketbra{--}, \nil_{q_1} \parallel \proc{B}}} & \dumoveto \delem{\frac{1}{4}}{C^0_{\ket{+i}}} \oplus \delem{\frac{1}{4}}{C^1_{\ket{+\uminus i}}} \oplus \delem{\frac{1}{4}}{C^0_{\ketbra{\uminus i}}} \oplus \delem{\frac{1}{4}}{C^1_{\ketbra{-\uminus i}}}
	\end{align*}
	It is evident that applying the same context $B[\blank]$ as for $\Delta$, after the synchronization the mass of the resulting distribution is always $1/2$.
	Therefore, we showed that $\conf{\kbphip, \proc{A}_{01} \parallel \proc{B}} \nsimus \conf{\kbphip, \proc{A}_{\pm} \parallel \proc{B}}$, and proved that an unscheduled semantics allows superluminal communication. \exqed
\end{example}

\subsection{Scheduled Bisimilarity}\label{sec:dem}
The two examples on qubit sources show that classical probabilistic approaches fail to correctly model quantum observability, due to processes capable to act differently in each element of a distribution.
In the following we define a scheduled bisimilarity, where distributions are required to evolve according to the same scheduler.

\begin{definition}[Scheduled Saturated Bisimilarity]
	A well-typed relation $\mathord{\rel} \subseteq \sdist{\confset} \times \sdist{\confset}$ is a \emph{scheduled saturated bisimulation} if $\Delta \rel \Theta$ implies $\mass{\Delta} = \mass{\Theta}$ and for any context $B[\blank]$
	compatible with both $\Delta$ and $\Theta$, and for any scheduler $s$, it holds that
	\begin{itemize}
		\item whenever $B[\Delta] \dmoveto{s} \Delta'$, there exists $\Theta'$
		      such that $B[\Theta] \dmoveto{s} \Theta'$ and $\Delta' \rel \Theta'$;
		\item whenever $B[\Theta] \dmoveto{s} \Theta'$, there exists $\Delta'$
		      such that $B[\Delta] \dmoveto{s} \Delta'$ and $\Delta' \rel \Theta'$.
	\end{itemize}

	Let \emph{scheduled saturated bisimilarity}, denoted $\simds$, be the largest scheduled saturated bisimulation.

\end{definition}

Schedulers suffice for guaranteeing indistinguishability of qubit sources sending $\kz$ and $\ko$ or $\kpl$ and $\km$ with equal probability.
\begin{example}\label{ex:fixed-nondet}
	Consider a pair of non-biased random qubit sources as those of~\autoref{ex:broken-nondet}.
	Fittingly, the lqCCS encodings of these sources are scheduled saturated bisimilar: $\Delta_{01} \simds \Delta_{\pm}$.
	Intuitively, this holds because scheduler consistency in the definition of $\dmoveto{}$ forbids to combine moves labelled by different choices:
	$\Delta_{01}'$ cannot mix the two choices to obtain an outcome that $\Delta_{\pm}'$ cannot replicate as it does in~\autoref{ex:broken-nondet}.
	In a nutshell, this allows to prove that both $\Delta_{01}$ and $\Delta_{\pm}$ are bisimilar to the point distribution with quantum state $\frac{1}{2}\I$, as one may expect from~\autoref{thm:qind}.
	\begin{align*}
		\Delta_{01}  & = \delem{\frac{1}{2}}{\conf{\kbz, t_0\tags c!q}} \oplus \delem{\frac{1}{2}}{\conf{\kbo, t_0\tags c!q}} \simds
		\sconf{\frac{1}{2}\I, t_0\tags c!q}                                                                                           \\
		\Delta_{\pm} & = \delem{\frac{1}{2}}{\conf{\kbpl, t_0\tags c!q}} \oplus \delem{\frac{1}{2}}{\conf{\kbm, t_0\tags c!q}} \simds
		\sconf{\frac{1}{2}\I, t_0\tags c!q} \exqed
	\end{align*}
\end{example}

The following theorem generalizes the equality of the example above, and lifts to lqCCS the indistinguishability relations between quantum states of~\autoref{thm:qind}.
This property states that the same process, acting on two different but indistinguishable mixed quantum states, exhibits behaviour that cannot be distinguished by any observer, therefore yielding bisimilar distributions.

\begin{restatable}{theorem}{propertyA}\label{thm:propertyA}
	If $\sum_i p_i \cdot \rho_i = \sum_j q_j \cdot \rho_j$ then $\osum_{i} \delem{p_i}{\conf{\rho_i, P}} \simds \osum_{j} \delem{q_j}{\conf{\sigma_j, P}}$.
\end{restatable}
\begin{proofsketch}
	We prove that $\Delta = \delem{p}{\conf{\rho , P}} \oplus \delem{(1-p)}{\conf{\sigma, P}} \simds \delem{1}{\conf{p \cdot \rho + (1 - p) \cdot \sigma, P}} =\Theta$, and the theorem follows by transitivity.
	First, we show that $\Delta$ replicates the moves of $\Theta$ by the linearity of superoperators, i.e. $p \cdot \E(\rho) + (1-p) \cdot \E(\sigma) = \E(p \cdot \rho + (1-p) \cdot \sigma)$, and similarly for measurements.
	Then, we show that $\Theta$ simulates $\Delta$.
	Here, the presence of schedulers is key: it forbids
	$\Delta$ from combining different non-deterministic choices to perform a move that would not be available to $\Theta$ (see~\autoref{ex:broken-nondet}).
\end{proofsketch}

This result directly derives from the use of schedulers, and it is not common in quantum versions of CCS: it holds only for specific distributions in~\cite{ceragioliQuantumBisimilarityBarbs2024} and in~\cite{dengbisimulations2018}.

Finally, we notice that $\simds$ is a linear relation, meaning that linear combinations of bisimilar distributions yield bisimilar distributions.
\begin{restatable}{theorem}{linearityCongruence}\label{thm:linearity}
	The bisimilarity $\simds$ is linear.
\end{restatable}
\begin{proof}
	Ultimately, the result follows from decomposability and linearity of $\dmoveto[\mu]{s}$.
	In the appendix, we prove a stronger result, namely that taking the \emph{linear closure} of a relation $\rel$ is a valid \emph{up-to technique}, in the sense of ~\cite{sangiorgienhancements2011}.
\end{proof}

\begin{figure}[t]
	\centering
	\begin{subfigure}{\textwidth}
		\centering
	\begin{tikzpicture}
		\node [conf] at (0,0) (c0) {$\conf{\kbphip,\proc{A}_{01} \parallel \proc{B}}$};
		\node [conf, below=12mm of c0] (c1) {$\delem{\frac{1}{2}}{\conf{\ketbra{00},\nil_{q_1} \parallel \proc{B}}} \oplus \delem{\frac{1}{2}}{\conf{\ketbra{11},\nil_{q_1} \parallel \proc{B}}}$};
		\node [conf, below left=6mm and 2mm of c0] (c2) {$\delem{\frac{1}{2}}{\conf{\ketbra{00},A \parallel C^0}} \oplus \delem{\frac{1}{2}}{\conf{\ketbra{11},A \parallel C^1}}$};
		\node [conf, below right=6mm and 2mm of c0] (c3) {$\delem{\frac{1}{2}}{\conf{\ketbra{ii},A \parallel C^0}} \oplus \delem{\frac{1}{2}}{\conf{\ketbra{\uminus i\uminus i},A \parallel C^1}}$};
		\node [conf, below=12mm of c2] (c4) {$\delem{\frac{1}{2}}{C^0_{\ket{00}}} \oplus \delem{\frac{1}{2}}{C^1_{\ket{11}}}$};
		\node [conf, below=12mm of c3] (c5) {$\delem{\frac{1}{4}}{C^0_{\ket{0i}}} \oplus \delem{\frac{1}{4}}{C^1_{\ket{1i}}} \oplus \delem{\frac{1}{4}}{C^0_{\ket{0\uminus i}}} \oplus \delem{\frac{1}{4}}{C^1_{\ket{1\uminus i}}}$};

		\draw [darrow] (c0) -- (c1) node[weight, left] {$t_0\tags\tau$};
		\draw [darrow] (c0) -- (c2) node[weight, above left] {$t_1\tags\tau$};
		\draw [darrow] (c0) -- (c3) node[weight, above right] {$t_2\tags\tau$};
		\draw [darrow] (c2) -- (c4) node[weight, left] {$t_0\tags\tau$};
		\draw [darrow] (c1) -- (c4) node[weight, below right] {$t_1\tags\tau$};
		\draw [darrow] (c3) -- (c5) node[weight, right] {$t_0\tags\tau$};
		\draw [darrow] (c1) -- (c5) node[weight, below left] {$t_2\tags\tau$};
	\end{tikzpicture}
	\end{subfigure}
	\vspace{.5cm}

	\begin{subfigure}{\textwidth}
		\centering
	\begin{tikzpicture}
		\node [conf] at (0,0) (c0) {$\conf{\kbphip,\proc{A}_{\pm} \parallel \proc{B}}$};
		\node [conf, below=12mm of c0] (c1) {$\delem{\frac{1}{2}}{\conf{\ketbra{++},\nil_{q_1} \parallel \proc{B}}} \oplus \delem{\frac{1}{2}}{\conf{\ketbra{--},\nil_{q_1} \parallel \proc{B}}}$};
		\node [conf, below left=6mm and 2mm of c0] (c2) {$\delem{\frac{1}{2}}{\conf{\ketbra{00},A \parallel C^0}} \oplus \delem{\frac{1}{2}}{\conf{\ketbra{11},A \parallel C^1}}$};
		\node [conf, below right=6mm and 2mm of c0] (c3) {$\delem{\frac{1}{2}}{\conf{\ketbra{ii},A \parallel C^0}} \oplus \delem{\frac{1}{2}}{\conf{\ketbra{\uminus i\uminus i},A \parallel C^1}}$};
		\node [conf, below=12mm of c2] (c4) {$\delem{\frac{1}{4}}{C^0_{\ket{+0}}} \oplus \delem{\frac{1}{4}}{C^1_{\ket{-0}}} \oplus \delem{\frac{1}{4}}{C^0_{\ket{+1}}} \oplus \delem{\frac{1}{4}}{C^1_{\ket{-1}}}$};
		\node [conf, below=12mm of c3] (c5) {$\delem{\frac{1}{4}}{C^0_{\ket{+i}}} \oplus \delem{\frac{1}{4}}{C^1_{\ket{-i}}} \oplus \delem{\frac{1}{4}}{C^0_{\ket{+\uminus i}}} \oplus \delem{\frac{1}{4}}{C^1_{\ket{-\uminus i}}}$};

		\draw [darrow] (c0) -- (c1) node[weight, left] {$t_0\tags\tau$};
		\draw [darrow] (c0) -- (c2) node[weight, above left] {$t_1\tags\tau$};
		\draw [darrow] (c0) -- (c3) node[weight, above right] {$t_2\tags\tau$};
		\draw [darrow] (c2) -- (c4) node[weight, left] {$t_0\tags\tau$};
		\draw [darrow] (c1) -- (c4) node[weight, below right] {$t_1\tags\tau$};
		\draw [darrow] (c3) -- (c5) node[weight, right] {$t_0\tags\tau$};
		\draw [darrow] (c1) -- (c5) node[weight, below left] {$t_2\tags\tau$};

	\end{tikzpicture}
	\end{subfigure}
	\caption{Scheduled semantics of the protocol described in \autoref{ex:broken-ftl}}
	\label{fig:FTL}
\end{figure}
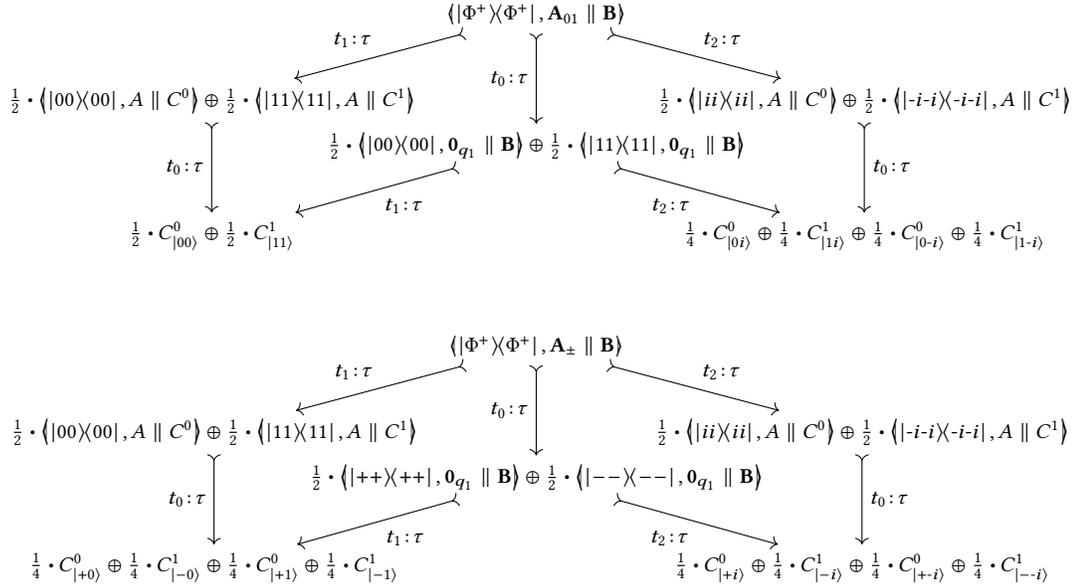

\begin{example}
	Recall the of the scheduled semantics of the processes $\proc{A}_{01} \parallel \proc{B}$ and $\proc{A}_{\pm} \parallel \proc{B}$ are shown in~\autoref{fig:FTL},
	where $\confel^i$ denotes the process $\nil_{q_1} \parallel c!i.\nil_{q_2}$, and $C^i_{\kp}$ is the configuration $\conf{\kbp, \nil_{q_1} \parallel c!i.\nil_{q_2}}$.
	From the figures it is evident that the branching structure of the two semantics is identical, therefore we will focus on the mass of the distribution.
	As for the unscheduled semantics, the only information that context can use to distinguish the two processes is the result communicated on the channel $c$.
	However, all final distributions of configurations have the same uniform distribution of $0$ and $1$ outcomes, therefore cannot be discriminated. \exqed
\end{example}

\section{Deciding Behavioural Equivalence Through Labelled Bisimilarities}\label{sec:lab}
\begin{figure}[t]
	\begin{mathpar}
		\inferrule[Tau]{}{\delem{\rho}{t \tags \tau . P} \qcmoveto{t} \delem{\rho}{P}}
		\and
		\inferrule[TauSynch]{}{\delem{\rho}{(t, t') \tags \tau . P} \qcmoveto{(t, t')} \delem{\rho}{P}}
		\and
		\inferrule[Sop]{}{\delem{\rho}{t \tags \E(\tilde{q}) . P} \qcmoveto{t} \delem{\E^{\tilde{q}}(\rho)}{P}}
		\and
		\inferrule[Meas]
		{}
		{\delem{\rho}{t \tags \meas{\tilde{q}}{x}.P}
			\qcmoveto{t}
			\osum_{m = 0}^{|\mathcal{M}|-1}\delem{\mathcal{M}_m^{\tilde{q}}(\rho)}
			{P[\sfrac{m}{x}]}{}
		}
		\and
		\inferrule[Send]{e \Downarrow v}{\delem{\rho}{t \tags c!e . P} \qcmoveto[c!v]{t} \delem{\rho}{P}}
		\and
		\inferrule[Recv]{c : \chtype{\qtype} \Rightarrow v \in \Sigma_\rho}{\delem{\rho}{t \tags c?x.P} \qcmoveto[c?v]{t} \delem{\rho}{P[\sfrac{v}{x}]}}
		\and
		\inferrule[IteT]
		{	e \Downarrow \true \\
			\delem{\rho}{P} \qcmoveto[\mu]{s} \qd}
		{\delem{\rho}{\ite{e}{P}{Q}} \qcmoveto[\mu]{s} \qd}
		\and
		\inferrule[IteF]{e \Downarrow \false \\ \delem{\rho}{Q} \qcmoveto[\mu]{s} \qd}{\delem{\rho}{\ite{e}{P}{Q}} \qcmoveto[\mu]{s} \qd}
		\and
		\inferrule[SumL]{\delem{\rho}{P} \qcmoveto[\mu]{s} \qd}{\delem{\rho}{P + Q} \qcmoveto[\mu]{s} \qd}
		\and
		\inferrule[SumR]{\delem{\rho}{Q} \qcmoveto[\mu]{s} \qd}{\delem{\rho}{P + Q} \qcmoveto[\mu]{s} \qd}
		\and
		\inferrule[ParL]{\delem{\rho}{P} \qcmoveto[\mu]{s} \qd \\ \mu \not \in \{ c?v \mid  v \in \Sigma_Q \}}{\delem{\rho}{P \parallel Q} \qcmoveto[\mu]{s} \qd \parallel Q}
		\and
		\inferrule[ParR]{\delem{\rho}{Q} \qcmoveto[\mu]{s} \qd \\ \mu \not \in \{ c?v \mid  v \in \Sigma_P \}}{\delem{\rho}{P \parallel Q} \qcmoveto[\mu]{s} P \parallel \qd}
		\and
		\inferrule[SynchL]{\delem{\rho}{P} \qcmoveto[c!v]{t} \delem{\rho}{P'} \\ \delem{\rho}{Q} \qcmoveto[c?v]{t'} \delem{\rho}{Q'}}{\delem{\rho}{P \parallel Q} \qcmoveto{(t,t')} \delem{\rho}{P' \parallel Q'}}
		\and
		\inferrule[SynchR]{\delem{\rho}{P} \qcmoveto[c?v]{t} \delem{\rho}{P'} \\ \delem{\rho}{Q} \qcmoveto[c!v]{t'} \delem{\rho}{Q'}}{\delem{\rho}{P \parallel Q} \qcmoveto{(t,t')} \delem{\rho}{P' \parallel Q'}}
		\and
		\inferrule[Restr]{\delem{\rho}{P} \qcmoveto[\mu]{s} \qd \\ \mu \not \in \{c!v, c?v \mid v\}}{\delem{\rho}{P \setminus c} \qcmoveto[\mu]{s} \qd}
		\and
		\inferrule[Id]
		{\qd \qcmoveto[\mu]{s} \qd'}
		{\qd \qmoveto[\mu]{s} \qd'}
		\and
		\inferrule[Not]
		{\qd \qcnmoveto[\mu]{s}}
		{\qd \qmoveto[\mu]{s} \epsilon}
		\and
		\inferrule[Conv]
		{\qd \qmoveto[\mu]{s} \qd' \\
			\qt \qmoveto[\mu]{s} \qt'}
		{\qd \oplus \qt \qmoveto[\mu]{s} \qd' \oplus \qt'}
	\end{mathpar}
	\caption{Rules of Density Operator Semantics.}
	\label{fig:qsemantics}
\end{figure}

%

Having chosen $\simds$ as our preferred behavioural equivalence, we search a computable alternative characterization.
We present a novel semantics that is based on a quantum extension of probability distributions; subsequently, we give a labelled bisimilarity, and we discuss its adherence with respect to the scheduled bisimilarity.
Finally, we prove its decidability for a rich class of lqCCS processes.

\subsection{Density Operator Semantics}

	The indistinguishability relation of~\autoref{thm:qind} suggests a new operational semantics.
Instead of dealing with multiple indistinguishable distributions of lqCCS configurations, 
we represent them with the same mathematical object, just as a single density operators represents multiple indistinguishable distributions of quantum values.
This leads to a semantics based on a quantum generalization of probability distributions, one that satisfies~\autoref{thm:propertyA} by construction.

We introduce \emph{quantum distributions}, which have density operators as weights, instead of real numbers in $[0,1]$.
More in detail, density operators exhibit a \emph{Partial Commutative Monoid} (PCM) structure: they can be summed with matrix addition, but only if the result is still in $\sdm{\hilb}$.
This generalizes the case of probabilities, i.e. the PCM of real numbers with sum in the interval $[0, 1]$.
Building upon this similarity, we can define quantum distributions as follows.

\begin{definition}[Quantum Distribution]\label{quantum_dst}
	Given a Hilbert space $\hilb_{\Sigma}$, the set of \emph{quantum distributions} over $S$ is
	\[
		\qdist[\Sigma]{S} =
		\left\{\,
		\qd \in \sdm{\hilb_{\Sigma}}^{S}
		\;\middle| \;
		\supp{\qd} \text{ is finite and }
		\sum_{s \in S} \tr(\qd(s)) \leq 1 \right\}
	\]
	where $\supp{\qd}$ is the support of $\qd$, i.e. the set $\{\,s \in S\;|\;\qd(s) \neq \Z\,\}$.
\end{definition}
We maintain the same notation:
we write $\epsilon$ for the empty distribution associating each element of $S$ with the matrix $\Z$;
we let $\delem{\rho}{s}$ be a distribution such that $(\delem{\rho}{s})(s) = \rho$ and $(\delem{\rho}{s})(s') = \Z$ if $s \neq s'$;
finally, we write
$\osum_{i \in I} \delem{p_i}{\qd_i}$ for the distribution such that $(\osum_{i \in I} \delem{p_i}{\qd_i})(s) = \sum_{i \in I}p_i\qd_i(s)$, with
$\{p_i\}_{i \in I}$ a finite set of non-negatives reals and $\sum_{i \in I} p_i = 1$.
%
Note that $\qdist[\emptyset]{S}$ is isomorphic to $\dist{S}$ since $1$-dimensional density operators are just real numbers in $[0,1]$.

Our new semantics is given directly on quantum distributions of lqCCS processes $\qdist[\Sigma]{\procset}$, where both the qubits state and the classical probability is encoded in the weights of the distribution.
The type system is extended to quantum distributions:
distributions defined on a Hilbert space $\hilb_{\Sigma}$ must contain only processes manipulating qubits in $\Sigma$.
\begin{definition}[Quantum Distribution Typing]
	We say that a quantum distribution of lqCCS processes $\qd \in \qdist[\Sigma]{\procset}$ has type $(\Sigma, \Sigma')$ (written $(\Sigma, \Sigma') \vdash \qd$) if $\Sigma' \vdash P$ for all processes $P \in \supp{\qd}$, and if $\Sigma' \subseteq \Sigma$.
\end{definition}

Hereafter, we assume a fixed set of qubits, and drop the $\Sigma$ from $\qdist[\Sigma]{\procset}$.
Moreover, we restrict ourselves to well-typed distributions.
The new labelled semantics for lqCCS is given in terms of LTS of quantum distributions of processes, i.e. a triple  $(\qdist{\procset}, \actset, \qmoveto[]{})$, with $\mathord{\qmoveto[]{}} \subseteq \qdist{\procset} \times \actset \times \choiceset \times \qdist{\procset}$.

To define the transition relation $\qmoveto[]{}$, we first consider the smallest transition relation $\qcmoveto[]{}$ over quantum distributions of closed processes that satisfies the related rules in~\autoref{fig:qsemantics}.
Note that these rules are almost in one-to-one correspondence with the ones in~\autoref{semantics}, and that they define the behaviour of distributions with a single element in the support.
The main difference with the configuration semantics is made clear by the \textsc{Meas} rule, where both the resulting quantum state and the probability of observing it are encoded as a single partial density operator.
The transition relation $\qmoveto[]{}$ is finally obtained by applying the last three rules in~\autoref{fig:qsemantics}, which includes a ``deadlock'' rule for distributions that cannot evolve otherwise, and recovers convexity by stating that the components of a quantum distribution can evolve altogether with the same scheduler and label.

\begin{example}
	Take the following process which measures and outputs a qubit, ignoring the classical outcome $x$ of the measurement.
	\begin{align*}
		\proc{P} \Coloneqq t_0 \tags \measstd{q}{x}.t_0 \tags c!q
	\end{align*}
	Consider its evolution with the starting state $\kbpl$, essentially modelling the first source of~\autoref{ex:broken-nondet}.:
	\begin{align*}
		\delem{\kbpl}{\proc{P}}
		\ \qmoveto{t_0} 	 \ \delem{\frac{1}{2}\kbz}{t_0 \tags c!q} \oplus \delem{\frac{1}{2}\kbo}{t_0 \tags c!q}
		\ \qmoveto[c!q]{t_0} \  \delem{\frac{1}{2}\I}{\nil}
	\end{align*}
	The first step follows the \textsc{Meas} rule.
	Note that, in the second distribution, the processes of the two branches coincide, since the measurement outcome $x$ is ignored.
	This means that we can rewrite $\delem{\frac{1}{2}\kbz}{t_0 \tags c!q} \oplus \delem{\frac{1}{2}\kbo}{t_0 \tags c!q}$ as $\delem{\frac{1}{2}\I}{t_0 \tags c!q}$, as they denote the same quantum distribution.
	In other words, our new semantics naturally embodies the equivalence between indistinguishable probability distributions of \autoref{thm:propertyA}. \exqed
\end{example}

\subsection{Labelled Bisimilarity}

We now move to define a suitable labelled bisimilarity relation based on the quantum distribution semantics.
As a first ingredient of our bisimilarity, we extend well-typed relations to quantum distributions.
\begin{definition}[Well-typed Relation]\label{def:welltypedq}
	A relation $\mathord{\rel} \subseteq \qdist{\confset} \times \qdist{\confset}$ is well-typed if $\qd \rel \qt$ implies $(\Sigma, \Sigma') \vdash \qd$ and $(\Sigma, \Sigma') \vdash \qt$ for some $\Sigma, \Sigma'$.
\end{definition}

As it is common, bisimilarity should be closed over transitions, meaning that each bisimilar processes must be capable of matching the moves of the others.

\begin{definition}[Transition-Closed Relation]
	A relation $\mathord{\rel} \subseteq \qdist{\procset} \times \qdist{\procset}$ is \emph{transition-closed} if $\qd \rel \qt$ implies
	that for each $s$ and $\mu$:
	\begin{itemize}
		\item whenever $\qd \qmoveto[\mu]{s} \qd'$, there exists $\qt'$ such that $\qt \qmoveto[\mu]{s} \qt'$ and $\qd' \rel \qt'$;
		\item whenever $\qt \qmoveto[\mu]{s} \qt'$, there exists $\qd'$ such that $\qd \qmoveto[\mu]{s} \qd'$ and $\qd' \rel \qt'$.
	\end{itemize}
\end{definition}

Before defining labelled bisimulation, we introduce some notation and discuss some which properties of quantum distributions can be immediately  distinguished by contexts.
Our goal is finding atomic observable properties that match the discriminating power of a context-closed relation like $\simds$.

Given a quantum distribution $\qd$ of lqCCS processes, we let its mass $\mass{\qd}$ be defined as the sum of its weights $\sum_{P \in \procset} \qd(P)$, as it is for probability distributions.
Note that, if $(\Sigma, \Sigma')$ is the type of $\qd$, then $\mass{\qd}$ is a (partial) density operator in $\sdm{\hilb_{\Sigma}}$.
We are particularly interested in a fragment of $\mass{\qd}$, namely in $\tr_{\Sigma'}(\mass{\qd})$, which we call the \emph{environment} of $\qd$, i.e. the state of the qubits that are unused by the processes of $\qd$ and that the contexts can observe immediately.
The environment of bisimilar quantum distributions must coincide, otherwise a context performing a simple measurement would distinguish them.
\begin{definition}[Environment-preserving Relation]
	The \emph{environment} of an quantum distribution $\qd \in \qdist[\Sigma]{\procset}$ is $\env{\qd} = 0$ if $\qd = \epsilon$,
	otherwise it is $\env{\qd} = \tr_{\Sigma'}(\mass{\qd})$ with $\Sigma'$ the only set of qubit names such that $(\Sigma, \Sigma') \vdash \qd$.

	A relation $\mathord{\rel} \subseteq \qdist{\procset} \times \qdist{\procset}$ is a \emph{environment-preserving} if $\qd \rel \qt$ implies $\env{\qd} = \env{\qt}$.
\end{definition}
%

We allow applying superoperators to configurations and distributions.
Given $(\Sigma, \Sigma') \vdash \qd = \osum_{i \in I} \delem{\rho_i}{P_i}$ and $\E \in \soset{\hilb_{\Sigma}}$, we let $\E(\qd)$ be the quantum distribution $\osum_{i \in I} \delem{\E(\rho_i)}{P_i}$ where all the weights are updated by the superoperator.
In the following, we allow updating $\qd$ also with superoperators that are defined on a subset of the qubits, i.e. with $\E \in \soset{\hilb_{\Sigma''}}$ with $\Sigma'' \subseteq \Sigma$.
In this case, we assume $\E$ to be padded as needed by tensoring its Kraus representation with the identity matrix.
We are interested in relations that are closed over the application of superoperators defined on the environment.
Intuitively, this is needed for matching contexts that modify qubits through trace preserving superoperators and measurements before sending them to processes.

\begin{definition}[Superoperator-closed Relation]
	A superoperator $\E$ is \emph{compatible} with respect to a distribution $\qd$ if $\E \in \soset{\hilb_{\Sigma \setminus \Sigma'}}$ with $(\Sigma, \Sigma') \vdash \qd$.

	A relation $\mathord{\rel} \subseteq \qdist{\procset} \times \qdist{\procset}$ is \emph{closed for compatible superoperators} if $\qd \rel \qt$ implies
	$\E(\qd) \rel \E(\qt)$ for any superoperator compatible with both $\qd$ and $\qt$.
\end{definition}


We can now define \emph{labelled scheduled bisimilations}.
In addition to the usual transition closure, labelled bisimulations are required to satisfy two additional conditions:
paired distributions must share the same environment; and the relation
must be closed for the application of superoperators over the qubits of the environment.
Intuitively, this is needed for proving that labelled bisimilarity has the same observing power of saturated bisimilarity,
as a context $B[\blank] = [\blank] \parallel R$ can read and modify the qubits of the environment.

\begin{definition}[Labelled Scheduled Bisimilarity]\label{def:simdl}
	A relation $\mathord{\rel} \subseteq \qdist{\procset} \times \qdist{\procset}$ is a \emph{labelled scheduled bisimulation} if it is well-typed, transition-closed, environment-preserving and closed for compatible superoperators.
	Let \emph{labelled scheduled bisimilarity}, denoted $\simdl$, be the largest labelled scheduled bisimulation.
\end{definition}

We now provide an example of bisimulation between two simple distributions, highlighting how $\simdl$ encompasses both quantum indistinguishability (as it respects \autoref{thm:propertyA} by construction) and classical equivalences (like the commutativity of the parallel operator).

\begin{example}
	Assume a process $\proc{P} \coloneqq t\tags\measstd{q}{x}.\ite{x = 0}{(t\tags c!q \parallel t'\tags d!1)}{(t'\tags d!1 \parallel t\tags c!q)}$, which takes a qubit $q$ and collapses it in the computational basis, then sends it on a quantum channel $c$ and signals termination on a classical channel $d$ (for "done").
	This process does inspect the classical outcome $x$ of the measurement, but it does nothing meaningful with it: the difference between the "then" branch and the "else" branch is only syntactical, as both branches communicate the same messages.
	For these reasons, if we execute $P$ on the input $\kbpl$ it is functionally the same as setting the qubit to $\oot\I$.

	More in detail, we can define another process $\proc{R} \coloneqq t\tags\mathcal{S}(q).(t\tags c!q \parallel t'\tags d!1)$, which uses the constant "Set" superoperator $\mathcal{S}$ such that $\mathcal{S}(\rho) = \oot\I$ for any $\rho$.
	We can then prove that $\delem{\kbpl}{\proc{P}} \simdl \delem{\kbpl}{\proc{R}}$, highlighting that $\proc{P}$ does not leak the measurement outcome to the environment.

	The transition systems of the two distributions are as follows
	\[
		\begin{tikzcd}
			& \delem{\kbpl}{P} \arrow[d, "t:\tau", two heads, tail]                                                                                       &                                                                 &                                                                          & \delem{\kbpl}{R} \arrow[d, "t:\tau", two heads, tail]                                             &                                                           \\
			& {\qd} \arrow[ld, "t:c!q"', two heads, tail] \arrow[rd, "t':d!1", two heads, tail] &                                                                 &                                                                          & \delem{\oot\I}{R'} \arrow[ld, "t:c!q"', two heads, tail] \arrow[rd, "t':d!1", two heads, tail] &                                                           \\
			\delem{\oot\I}{t'\tags \nil \parallel d!0} \arrow[rd, "t':d!1"', two heads, tail] &                                                                                                                                             & \delem{\oot\I}{t\tags c!q\parallel \nil} \arrow[ld, "t:c!q", two heads, tail] & {\delem{\oot\I}{t'\tags \nil \parallel d!0}} \arrow[rd, "t':d!1"', two heads, tail] &                                                                                                   & {\delem{\oot\I}{t\tags c!q\parallel \nil}} \arrow[ld, "t:c!q", two heads] \\
			& \delem{\oot\I}{\nil\parallel \nil}                                                                                                                        &                                                                 &                                                                          & \delem{\oot\I}{\nil\parallel \nil}                                                                              &
		\end{tikzcd}
	\]
	where $\qd = \delem{\oot\kbz}{P'[0/x]} \oplus \delem{\oot\kbo}{P'[1/x]}$, $P' = \ite{x = 0}{(t\tags c!q \parallel t'\tags d!1)}{(t'\tags d!1 \parallel t\tags c!q)}$ and $R' = (t\tags c!q \parallel t'\tags d!1)$.

	We can build a bisimulation $\rel$ that relates all the states above, and includes also the identity relation.
	We define $\rel$ the smallest relation such that $\delem{\kbpl}{\proc{P}} \ \rel \ \delem{\kbpl}{\proc{R}}$,
	$\qd \ \rel\ \delem{\oot\I}{\proc{R'}}$, and $\delem{\rho}{\proc{Q}} \ \rel \ \delem{\rho}{\proc{Q}}$ for any well-typed $\delem{\rho}{\proc{Q}}$.

	To prove that $\rel$ is a bisimulation, we must show that is well-typed, transition-closed, environment-preserving, and closed for compatible superoperators.
	It is easy to check that it is well-typed, and it is transition-closed as well, simply by inspecting the above transition systems.
	It is environment-preserving: for the first two couples, we have $\env{\delem{\kbpl}{\proc{P}}} = 1 = \env{\delem{\kbpl}{\proc{R}}} = \env{\qd} = \env{\delem{\oot\I}{\proc{R'}}}$; for the identity couples, they have trivially the same environment.
	Finally, the relation $\rel$ is closed for compatible superoperators: for the first two couples, there are no compatible superoperators, as there are no environment qubits; for the identity couples, we have $\delem{\E(\rho)}{\proc{Q}} \rel \delem{\E(\rho)}{\proc{Q}}$ for any superoperator $\E$.

	Thus, we can conclude that $\delem{\kbpl}{\proc{P}} \simdl \delem{\kbpl}{\proc{R}}$.
	Notice how we had to resort to a bisimulation containing infinite couples, due to the requirement of superoperator-closedness.
	Defining a finite $\rel$ relating only the states in the transition system, for example, would have not been sufficient:
	if we relate $\delem{\oot\I}{\nil} \rel \delem{\oot\I}{\nil}$, we need to do the same for all possible compatible superoperator $\E$. \exqed
\end{example}

We prove that
saturated bisimilarity in the probabilistic semantics coincides with labelled bisimilarity in the quantum distribution semantics.
\begin{restatable}[Full Abstraction]{theorem}{demonicfa}\label{thm:demonicfa}
	For any $\rho, \sigma, P, Q$ it holds that
	$$\sconf{\rho, P} \simds \sconf{\sigma, Q} \text{ if and only if }\delem{\rho}{P} \simdl \delem{\sigma}{Q}$$
\end{restatable}
\begin{proofsketch}
	To simplify the proof, we introduce a probabilistic labelled bisimilarity $\simpl$ as a bridge between the probabilistic saturated bisimilarity $\simds$ and the quantum labelled one $\simdl$.
	We prove the equivalence between $\simds$ and $\simpl$ using standard methods for probabilistic systems, like up-to techniques.
	We then show a fully-abstract translation between the probabilistic and quantum distributions that preserves and reflects their defined equivalences ($\simpl$ and $\simdl$ respectively).
\end{proofsketch}

\subsection{Deciding Equivalence of Protocols}

None of the conditions of~\autoref{def:simdl} can be omitted.
As shown by the following counterexample,
the closure over environments superoperators is needed to recover the distinguishing capability of the saturated approach,
making our labelled bisimilarity not adequate for verification.
\begin{example}\label{ex:superopneed}
	Take the following processes which measure a qubit received from the environment in two orthogonal basis
	\begin{align*}
		\proc{P} \Coloneqq t \tags c?x.t \tags \measstd{x}{y}.t \tags c!y.\nil_{x} \qquad\text{ and }\qquad
		\proc{Q} \Coloneqq t \tags c?x.t \tags \meashad{x}{y}.t \tags c!y.\nil_{x}
	\end{align*}
	Consider their evolution with starting state $\kbip$:
	\begin{align*}
		\delem{\kbip}{\proc{P}} & \qmoveto[c?q]{t} \delem{\kbip}{t \tags \measstd{q}{y}.t \tags c!y.\nil_{q}}
		\qmoveto{t} \delem{\frac{1}{2}\kbz}{t \tags c!0.\nil_{q}} \oplus \delem{\frac{1}{2}\kbo}{t \tags c!1.\nil_{q}} \\
		\delem{\kbip}{\proc{Q}} & \qmoveto[c?q]{t} \delem{\kbip}{t \tags \meashad{q}{y}.t \tags c!y.\nil_{q}}
		\qmoveto{t} \delem{\frac{1}{2}\kbpl}{t \tags c!0.\nil_{q}} \oplus \delem{\frac{1}{2}\kbm}{t \tags c!1.\nil_{q}}
	\end{align*}
	It is evident that
	the two configurations are equated if in place of $\simdl$ we consider the largest environment-preserving relation that is closed with respect to $\qmoveto[\mu]{s}$ for all $\mu$ and $s$
	(note that the different resulting quantum state is invisible due to the discard at the end).

	However, since the sent qubit $q$ is received from the environment, it can be transformed by the context before communication.
	Thus, $\delem{\kbip}{\proc{P}}$ is bisimilar to $\delem{\kbip}{\proc{Q}}$ only if $\E(\delem{\kbip}{\proc{P}}) \simdl \E(\delem{\kbip}{\proc{Q}})$ for any compatible $\E$.
	To see that indeed this is not the case, consider the superoperator $\E$ that applies the unitary $SH$ to $q$ (recall that $SH\kip = \kz$ and $SH\kim = \ko$);
	then the semantics become:
	\begin{align*}
		\E(\delem{\kbip}{\proc{P}}) = \delem{\kbz}{\proc{P}} & \qmoveto[c?q]{t} \delem{\kbip}{t \tags \measstd{q}{y}.t \tags c!y.\nil_{q}} \qmoveto{t} \delem{\kbz}{t \tags c!0.\nil_{q}}
		\\
		\E(\delem{\kbip}{\proc{Q}}) = \delem{\kbz}{\proc{Q}} & \qmoveto[c?q]{t} \delem{\kbip}{t \tags \meashad{q}{y}.t \tags c!y.\nil_{q}} \qmoveto{t} \delem{\frac{1}{2}\kbpl}{t \tags c!0.\nil_{q}} \oplus \delem{\frac{1}{2}\kbm}{t \tags c!1.\nil_{q}}
	\end{align*}
	It is evident that the two configurations are not bisimilar: the former can only send $0$ over $c$ while the latter can send $1$ with probability $\sfrac{1}{2}$.

	The capability of modifying the qubits in the environment is evident with saturated bisimilarity, where the two configurations are distinguished by $B[\blank] = [\blank] \parallel SH(q).c!q$. \exqed

\end{example}

There is a practical limitation in using the scheduled bisimilarity of~\autoref{def:simdl} to verify distributed quantum systems and protocols:
the universal quantification over all possible environment superoperators make the relation difficult to prove.
In addition, the~\autoref{ex:superopneed} shows that the closure must be considered, otherwise there is the risk of equating systems that can be distinguished by some context.
Luckily, this closure is only needed because of processes that can receive qubits from the environment, such as those of~\autoref{ex:superopneed}.
Indeed, lots of concrete cases do not need this feature, and if we remove it we can give an equivalent decidable relation, as we show below.

\begin{definition}[Input-Restricted Processes]
	Let $\procsetnr \subset \procset$ be the set of lqCCS \emph{input-restricted processes}, i.e.
	with no unrestricted input channels of type $\chtype{\qtype}$.
	Formally, $P$ is in $\procsetnr$ if $\fcr(P) = \emptyset$ where $\fcr(\blank)$ is defined as:
	\begin{gather*}
		\fcr(\nil_{\tilde{e}} ) = \emptyset
		\qquad
		\fcr(P \setminus c) = \fcr(P) \setminus \{c\}
		\qquad
		\fcr(t \tags c?x . P) = 
		\begin{cases}
		\{c\} \cup \fcr(P)	& \text{ if } c : \chtype{\qtype}\\
		\fcr(P)	& \text{ otherwise }
		\end{cases}	
		\\
		\fcr(t \tags \tau . P) =
		\fcr((t,t) \tags \tau. P ) =
		\fcr(t \tags \E(\tilde{e}).P ) =
		\fcr(t \tags \meas{\tilde{e}}{x}.P ) =
		\fcr(t \tags c!e . P) = \fcr(P)
		\\
		\fcr(\ite{e}{P}{P'}) =
		\fcr(P + P') =
		\fcr(P \parallel P') = \fcr(P) \cup \fcr{P'}
	\end{gather*}
\end{definition}

\begin{definition}[Labelled Ground Scheduled Bisimilarity]\label{def:simgdl}
	A relation $\mathord{\rel} \subseteq \qdist{\procset} \times \qdist{\procset}$ is a \emph{labelled ground scheduled bisimulation} if it is well-typed, transition-closed and environment-preserving.
	Let \emph{labelled ground scheduled bisimilarity}, denoted $\simgdl$, be the largest labelled ground scheduled bisimulation.
\end{definition}

We extend~\autoref{thm:demonicfa} to the case of $\procsetnr$, showing that labelled ground bisimilarity matches saturated bisimilarity if processes do not receive values from the environment.
\begin{restatable}{theorem}{demonicgfa}\label{thm:demonicgfa}
	For any $P, Q \in\procsetnr, \rho, \sigma$
	it holds that
	$\sconf{\rho, P} \simds \sconf{\sigma, Q}$ if and only if $\delem{\rho}{P} \simgdl \delem{\sigma}{Q}$
\end{restatable}
\begin{proofsketch}
	When limited to processes in $\procsetnr$, we can prove that ground bisimilarity is superoperator closed, thus $\simgdl$ and $\simdl$ coincide.
	Then the result follows from \autoref{thm:demonicfa}.
\end{proofsketch}

As a final remark, note that the problem with unconstrained non-determinism and unscheduled bisimilarity is critical also for processes that cannot freely receive input from the environment.
Indeed, both \autoref{ex:broken-nondet} and \autoref{ex:broken-ftl} are in $\procsetnr$, as well as the quantum coin flip introduced in~\autoref{sec:overview}.

\section{Toward Real-World Protocols}\label{sec:real}
We conclude with three real-world examples, showing the expressivity of lqCCS and what kind of properties can be verified thanks to labelled bisimilarity.
We will focus on Superdense Coding~\cite{SDC}, Quantum Teleportation~\cite{qteleportation}, and the already cited quantum Coin Flipping~\cite{bb84}.

\anote{Camera Ready: mettere tutto il codice in figura e formattarlo come nell'overview}

\begin{example}[Superdense Coding]
	The superdense coding protocol allows Alice to communicate a
	two-bit integer to Bob by sending him single qubit, exploiting an entangled pair $\kphip$.
	The protocol is as follows: Alice chooses an integer in $[0,3]$ and encodes it by applying suitable transformations to her  entangled qubit, which is then sent to Bob;
	Bob receives the qubit and decodes
	it
	by performing $\mathit{CNOT}$ and $H \otimes I$ on the pair of qubits (the received qubit and his original one).
	Finally, he measures the qubits in the standard basis, recovering the integer chosen by Alice.

	The actions of Alice (as well as the final result) are dependent on the integer input of the protocol, and are modelled as a non-deterministic choice.
	We can encode the different possible inputs directly in the tags: Alice starts by choosing either $t_0, t_1, t_2$ or $t_3$, which correspond to encoding $0, 1, 2$ or $3$ in the qubit sent by Alice.
	Bob's moves are all tagged with the same tag $t$.

	We consider the following encoding of the protocol $\proc{SDC} = (\proc{A} \parallel \proc{B}) \setminus c$,
	where we assume that Alice ($\proc{A}$) and Bob ($\proc{B}$) already share an entangled pair ($q_1,q_2$), and  we write $t\tags\tau^n$ for the repetition of $n$ consecutive $t$-tagged silent moves.

\[
	\begin{tabular}{c c c}
		{$\begin{aligned}
			\proc{A} \Coloneqq \   
			& (t_0\tags I(q_0). t\tags c!q_0.\nil)\ +\\
			&  (t_1\tags X(q_0). t\tags c!q_0.\nil)\ +\\
			& (t_2\tags Z(q_0). t\tags c!q_0.\nil)\ +\\
			& (t_3\tags ZX(q_0). t\tags c!q_0.\nil) 
		\end{aligned}$}
		&
		{\qquad
		$\begin{aligned}
			\proc{B} \Coloneqq \  
			& t\tags c?x.\\
			& t\tags \mathit{CNOT}(x, q_1).\\
			&t\tags H(x).\\
			& t\tags \meas[\mstd^2]{x,q_1}{y}.\\
			&t\tags out!y.\nil_{x,q_1} 
		\end{aligned}$}
		&
		{\qquad
		$\begin{aligned}
		\proc{Spec}  \Coloneqq \
		& t_0\tags\tau. (t,t)\tags\tau. t\tags\tau^3 . t\tags out!0. \nil_{q_0,q_1}\ +                                    \\
		& t_1\tags\tau. (t,t)\tags\tau. t\tags\tau^3 . t\tags out!1. \nil_{q_0,q_1}\ +                                    \\
		& t_2\tags\tau. (t,t)\tags\tau. t\tags\tau^3 . t\tags out!2. \nil_{q_0,q_1}\ +                                    \\
		& t_3\tags\tau. (t,t)\tags\tau. t\tags\tau^3 . t\tags out!3. \nil_{q_0,q_1}
		\end{aligned}$}
	\end{tabular}
\]

	We let $\qd = \delem{\kbphip}{\proc{SDC}}$, with $\qt = \delem{\kbphip}{\proc{Spec}}$ its specification.
	Note that $\proc{Spec}$ states that the choice of the initial unitary determines the final value to be communicated.
	Since the process has no unrestricted input, we can resort to the ground version of our bisimilarity to prove the adequacy of the protocol.
	We now sketch the proof for $\Delta \simgdl \Theta$.

	The evolutions of $\qd$ and $\qt$ are as follows, for $n = 0,1,2,3$
	\begin{align*}
		 & \qd \qmoveto{t_n} \qmoveto{(t,t)}\qmoveto{t}^3
		\qd_n = \delem{\ketbra{n}}{t\tags out!n.\nil_{q_0,q_1} \setminus c}                                                 \\
		 & \qt  \qmoveto{t_n} \qmoveto{(t,t)}\qmoveto{t}^3 \qt_n = \delem{\kbphip}{t\tags out!n.\nil_{q_0,q_1} \setminus c}
	\end{align*}
	where, abusing notation, we use $\ket{0} = \ket{00}, \ket{1} = \ket{01}, \ket{2} = \ket{10}$ and $\ket{3} = \ket{11}$ when speaking of pairs of qubits.

	The two qubits are always owned by the processes, so
	the environment of all the distributions is trivially the same.
	Note that $\qd_n \simgdl \qt_n$ for any $n$, because both distributions $\qd_n$ and $\qt_n$ sends the number $n$ on the channel $out$, and reduce to deadlock distributions.
	These conditions are enough to construct a ground bisimulation between $\qd$ and $\qt$. \exqed
\end{example}

\begin{example}[Quantum Teleportation]
	The objective of the teleportation protocol is to allow Alice to send quantum information to Bob 
	without a quantum channel.
	Alice and Bob must have each one of the qubits of an entangled pair $\kphip$.
	The protocol works as follows: Alice performs a fixed set of unitaries to the qubit to transfer and to her part of the entangled pair;
	then, she measures the qubits and sends the classical outcome to Bob, which applies different unitaries to his own qubit according to the received information.
	In the end, the qubit of Bob will be in the state of Alice's one, and the entangled pair is discarded.

	This protocol does not involve any non-deterministic choice, thus we have great freedom in the choice of tags.
	We follow a simple policy, where all Alice's and Bob's actions are tagged by $t$.

	Consider the following encoding of the protocol $\proc{Tel} = (\proc{A} \parallel \proc{B}) \setminus c$, where we assume that Alice ($\proc{A}$) and Bob ($\proc{B}$) already share an entangled pair ($q_1,q_2$).
\[
\begin{tabular}{c c c}
	$\begin{aligned}
		\proc{A} \Coloneqq \   
		& t \tags \cnot(q_0,q_1).\\
		& t \tags H(q_0).\\
		& t \tags \meas[\mstd^2]{q_0,q_1}{x}.\\
		& t \tags c!x. \nil_{q_0,q_1}
	\end{aligned}$ 
	&\qquad
	\qquad$\begin{aligned}
		\proc{B} \Coloneqq \  
		& t \tags c?y.\\
		& \ite{y = 0}{t \tags I(q_2).\\&\quad}
		{\ite{y = 1}{t \tags X(q_2). \\&\quad}
		{\ite{y = 2}{t \tags Z(q_2). \\&\quad}
		{t \tags \mathit{ZX}(q_2).}\\&}}
		t \tags \mathit{out}!q_2
	\end{aligned}$
	&\qquad
	$\begin{aligned}
	\proc{Spec}  \Coloneqq \
	& t\tags\tau^3.\\
	& (t,t)\tags\tau.\\
	& t\tags\mathit{SWAP}(q_0,q_2).\\
	& t\tags\mathit{out}!q_2.\nil_{q_0,q_1}
	\end{aligned}$
\end{tabular}
\]
	We write $\mstd^{2}$ to denote the two-qubit measurement in the basis $\{ \ket{00}, \ket{01}, \ket{10}, \ket{11}\}$.
	Note that $\proc{Spec}$ simply states that, after the tagged operations, the state of qubits is swapped and the correct state is communicated over the expected channel.
	We let $\qd = \delem{\ketbra{\psi\Phi^+}}{\proc{Tel}}$, with $\qt = \delem{\ketbra{\psi\Phi^+}}{\proc{Spec}}$ its specification for any $\kp = \alpha\kz + \beta\ko$, and sketch the proof for $\Delta \simgdl \Theta$ below.

	The relevant steps of the evolution of $\qd$ and $\qt$ are
	\begin{align*}
		\qd \qmoveto{t}\qmoveto{t}\qmoveto{t} \
		                                                   & \qd' = \osum\nolimits_{n = 0}^3 \delem{\frac{1}{4}\ketbra{n} \otimes \ketbra{\psi_n}}{{t\tags c!n.\nil_{q_0,q_1} \parallel \proc{B} \setminus c}} \\
		\quad\qmoveto{(t,t)}\qmoveto{t} \
		                                                   & \qd'' = \delem{\sfrac{\I}{4} \otimes \ketbra{\psi}}{{\nil_{q_0,q_1} \parallel t\tags \mathit{out}!q_2 \setminus c}}                               \\
		\qmoveto[out!q_2]{t} \
		                                                   & \qd''' = \delem{\sfrac{\I}{4} \otimes \kbp}{\nil_{q_0,q_1} \setminus c}                                                                           \\
		\qt  \qmoveto{t}^3  \qmoveto{(t,t)}\qmoveto{t}  \  &
		\qt' = \delem{\kbphip\otimes\kbp}{t\tags\mathit{out}!q_2.\nil_{q_0,q_1} \setminus c}                                                                                                                   \\
		\qmoveto[out!q_2]{t} \
		                                                   & \qt'' = \delem{\kbphip\otimes\kbp}{\nil_{q_0,q_1} \setminus c}
	\end{align*}
	where $\ket{\psi_{0}} = \kp$, $\ket{\psi_{1}} = \beta\kz + \alpha\ko$, $\ket{\psi_{2}} = \alpha\kz - \beta\ko$, $\ket{\psi_{3}} = \beta\kz - \alpha\ko$,
	and where, abusing notation, we use $\ket{0} = \ket{00}, \ket{1} = \ket{01}, \ket{2} = \ket{10}$ and $\ket{3} = \ket{11}$ when speaking of pairs of qubits.
	Until the last move, all the qubits are owned by the processes, and so the environments coincide.
	The last move is then a send on the channel $\mathit{out}$ for both distributions
	The bisimilarity then is easily checked as $\qd'''$ and $\qt''$ are in deadlock, and $\tr_{q_0,q_1}\mass{\qd''} = \kbp = \tr_{q_0,q_1}\mass{\qt''}$.
	These conditions are enough to construct a ground bisimulation between $\qd$ and $\qt$.
	Notice how there are no inputs in our processes, and so we do not have to consider a universal quantification over all possible compatible superoperators. \exqed
\end{example}

\begin{example}[Quantum Coin Flip]
	We restate here the Quantum Coin Flip protocol described in \autoref{sec:overview}.
	Similarly to the previous example, we employ the trivial policy of tagging almost everything with $t$, whereas the last operations are tagged by $t_a, t_b$ and $t_c$ to guarantee determinism.
	\[
		\begin{tabular}{c c}
			$\begin{aligned}
					 \proc{Alice} \Coloneqq \  & t\tags\measrand{s}.     \\[-.1cm]
					                           & \ite{s = 0}
					 {t\tags\un{H}(q).t\tags\measstd{q}{w}.              \\[-.1cm]&\quad}
					 {t\tags\un{I}(q).t\tags\meashad{q}{w}.}             \\[-.1cm]
					                           & t\tags\text{AtoB}!q.    \\[-.1cm]
					                           & t\tags\text{guess}?g.   \\[-.1cm]
					                           & t\tags\text{witness}!w. \\[-.1cm]
					                           & t\tags\text{secret}!s.  \\[-.1cm]
					                           & t_a\tags a!(g=s))
				 \end{aligned}$ &
			$\begin{aligned}
					 \proc{Bob} \Coloneqq \  & t\tags\text{AtoB}?x.        \\[-.1cm]
					                         & t\tags\measrand{g}.         \\[-.1cm]
					                         & \ite{g=0}
					 {t\tags\meashad{x}{\mathit{p}}.                       \\[-.1cm]&\quad}
					 {t\tags\measstd{x}{\mathit{p}}.}                      \\[-.1cm]
					                         & t\tags\text{guess}!g.       \\[-.1cm]
					                         & t\tags\text{witness}?w.     \\[-.1cm]
					                         & t\tags\text{secret}?s.      \\[-.1cm]
					                         & (t_b\tags b!(g=s) \parallel
					 t_c\tags \text{cheat}!(g \neq s \land p \neq w) \parallel \nil_x)
				 \end{aligned}$
		\end{tabular}
	\]

	Thanks to lqCCS and to ground labelled bisimilarity, we can investigate both correctness and security of this protocol, proving the properties presented in~\autoref{sec:overview}.

	\paragraph{Correctness}
	A correct execution of the protocol should correspond to tossing a fair coin to select the winner, and having both Alice and Bob agreeing on the winner.
	Since tags do not carry any semantics in this example, we consider a protocol equivalent to its specification when there exists a tagging making them bisimilar.
	Formally, the intended behaviour can be written down as
	$$\proc{FairCoin} \Coloneqq t\tags\measrand{x}.t\tags\tau^2.(t,t)\tags\tau.t\tags\tau^2.(t,t)\tags\tau^3.
		(t_a\tags a!x \parallel t_b\tags b!x \parallel t_c\tags \text{cheat}!0).
	$$
	This tagging of the specification is obtained simply by specifying which steps are synchronization and which are internal action.
		In the general case, there is a finite set of "reasonable" tagging policies to check.

	Since $\proc{QCF}   \Coloneqq (\proc{Alice} \parallel \proc{Bob}) \setminus \{ \text{AtoB}, \text{guess}, \text{secret}, \text{witness} \}$ is in $\procsetnr$,
	we do not need to consider superoperator closedness, and we can prove correctness by showing the bisimilarity below via tedious but simple calculations.
	\[
		\delem{\kbz}{\proc{QCF}} \simgdl \delem{\kbz}{\proc{FairCoin}}.
	\]

	\paragraph{Bob's Attacks}
	To guarantee that the protocol does not favour Bob, we show that he is not able to infer the Alice's secret bit $s$ from the received qubit $q$.
	It is sufficient to take the two versions of $\proc{Alice}$ when $s$ is $0$ or is $1$, and to consider the part of the process that happens before the secret is revealed, $\proc{Alice}_0$ and $\proc{Alice}_1$:
	\begin{align*}
		\proc{Alice}_0 = t\tags\un{H}(q).t\tags\measstd{q}{w}.t\tags\text{AtoB}!q
		\ \quad \
		\proc{Alice}_1 = t\tags\un{I}(q).t\tags\meashad{q}{w}.t\tags\text{AtoB}!q
	\end{align*}

	We have that $\delem{\kbz}{\proc{Alice}_0} \simgdl \delem{\kbz}{\proc{Alice}_1}$, as both of them evolve into $\delem{\frac{1}{2}\I}{t\tags\text{AtoB}!q}$ after two $\tau$ transitions.
	Note that the unscheduled saturated bisimilarity tells the two processes apart, as their evolution coincide with the two qubit sources of~\autoref{ex:broken-nondet}.
	Thus, analyses that do not constraint unfeasible non-determinism implicitly deem that Bob can distinguish the secret value of Alice and always win the toss.

	\paragraph{Alice's Attacks}
	To examine the protocol resistance to a dishonest Alice, we can write down two known attacks from~\cite{bb84}, and study the resulting transition system.
	As detailed in~\autoref{sec:overview}, we will consider two attackers, namely $\proc{Alison}$ and $\proc{Alix}$.

	\[
		\begin{tabular}{c c}
			$\begin{aligned}
					 \proc{Alison} \Coloneqq \  & t\tags\measrand{s}.                                \\[-.1cm]
					                            & \ite{s = 0}
					 {t\tags\un{H}(q).t\tags\meas[\mstd]{q}{w}.                                      \\[-.1cm]&\quad}
					 {t\tags\un{I}(q).t\tags\meas[\mhad]{q}{w}.}                                     \\[-.1cm]
					                            & t\tags\text{AtoB}!q.                               \\[-.1cm]
					                            & t\tags\text{guess}?g.                              \\[-.1cm]
					                            & t\tags\text{witness}!w.                            \\[-.1cm]
					                            & (t\tags\text{secret}!(1-g) \parallel t_a\tags a!0)
				 \end{aligned}$ &
			$\begin{aligned}
					 \proc{Alix} \Coloneqq \  & t\tags\un{H}(q).                                   \\[-.1cm]
					                          & t\tags\un{CNOT}(q,q').                             \\[-.1cm]
					                          & t\tags\text{AtoB}!q.                               \\[-.1cm]
					                          & t\tags\text{guess}?g.                              \\[-.1cm]
					                          & \ite{g = 0}
					 {t\tags\meas[\mhad]{q'}{w}.                                                   \\[-.1cm]&\quad}
					 {t\tags\meas[\mstd]{q'}{w}.}                                                  \\[-.1cm]
					                          & t\tags\text{witness}!w.                            \\[-.1cm]
					                          & (t\tags\text{secret}!(1-g) \parallel t_a\tags a!0)
				 \end{aligned}$
		\end{tabular}
	\]

	Alison's attack is not particularly dangerous, as we can prove that
	$$\delem{\kbz}{(\proc{Alison} \parallel \proc{Bob}) \setminus C} \simgdl \delem{\kbz}{\proc{LeakyUnfairCoin}},$$
	where we write $C$ for the set of channels $\{ \text{AtoB}, \text{guess}, \text{secret}, \text{witness} \}$, and $\proc{LeakyUnfairCoin}$ for the process that always select $0$ (i.e. Alison) as winner and
	goes undetected ($\text{cheat}!0$) three fourth of the times.
	$$\proc{LeakyUnfairCoin} \Coloneqq 	t\tags\meas[\M_{\frac{3}{4}}]{}{x}.t\tags\tau^2.(t,t)\tags\tau.t\tags\tau^2.(t,t)\tags\tau^3.
		(t_a\tags a!0 \parallel t_b\tags b!0 \parallel t_c\tags \text{cheat}!x).$$

	When executing the protocol with multiple qubits, the chances of Alison going undetected increases multiplicatively for each qubit, making it exponentially unlikely.

	On the contrary, Alix's attack goes undetected with probability $1$, and we can prove that
	$$\delem{\kbz}{(\proc{Alix} \parallel \proc{Bob}) \setminus C} \simgdl \delem{\kbz}{\proc{UnfairCoin}},$$
	where we write $\proc{UnfairCoin}$ for the process
	$$\proc{UnfairCoin} \Coloneqq
		t\tags\tau^3.(t,t)\tags\tau.t\tags\tau^2.(t,t)\tags\tau^3.(t_a\tags a!0 \parallel t_b\tags b!0 \parallel t_c\tags \text{cheat}!0).$$ \exqed

\end{example}

\section{Related Works}\label{sec:rel}

We consider approaches that share our objectives: $(i)$ a description language for concurrent or distributed quantum processes; $(ii)$ an adequate operational semantics; and $(iii)$ an approach to verify the expected properties of processes.

\subsection{Description Languages}

Proposed process calculi for quantum systems may have different syntax, but share mostly the same feature set.
The Quantum Process Algebra (QPAlg) \cite{lalireprocess2004, lalireRelationsQuantumProcesses2006} extends value-passing CCS with syntactic operators for unitary transformations and measurements.
As in lqCCS, quantum operations are silent $\tau$-moves.
Qubit names are treated according to an affine policy, in place of the linearity of lqCCS: a process can simply ignore a qubit that it owns, as e.g. in $H(q).\nil$.
In this case, QPAlg assumes that qubits are implicitly discarded, i.e. they cannot be observed by the environment ($H(q).\nil$ is the same as $H(q).\nil_q$).

The Communicating Quantum Processes (CQP) calculus \cite{gaycommunicating2005,gaytypes2006} is inspired by the $\pi$-calculus with primitives for qubit transformations, measurements and creation.
It features an affine type system with the objective of enforcing single ownership for qubits, and it assumes qubits to be implicitly discarded, similarly to QPAlg.
Although name passing is technically possible, it is hardly considered in examples, and it is also explicitly ruled out in a version of the calculus~\cite{davidsonformal2012}.
For technical reasons, the same version also rules out the classical if-then-else conditional composition, which greatly reduce the expressive power of CQP as a modelling language.

The Quantum Calculus of Communicating Systems qCCS \cite{fengprobabilistic2007,fengbisimulation2012} shares the core syntax with lqCCS, but also features recursive processes.
The adherence to the no-cloning theorem is obtained with syntactic requirements in place of a type system.
Like for CQP, qubit names are treated according to an affine policy in place of the linearity of lqCCS;
however, on the contrary of CQP, qCCS assumes that processes make their qubits implicitly visible ($H(q).\nil$ makes the qubit $q$ visible, roughly as if it was sent on some unrestricted channel $H(q).e!q$).

We first developed our description language, lqCCS, in~\cite{ceragioliQuantumBisimilarityBarbs2024}, where we give an asynchronous version of the calculus, equipped with a reduction semantics.
The linear type system is introduced both to comply with the no-cloning theorem and to precisely characterize the set of visible qubits, solving the discrepancy discussed above between the other proposed formalisms: lqCCS can recover the visibility implicitly assumed by QPAlg, CQP and qCCS by applying an adequate discarding discipline.


\subsection{Operational Semantics}

All the proposed process calculi discussed above have an operational semantics in terms of pLTSs, with transitions mapping configurations to distributions of configurations.
Communicating quantum values by sending qubits' name is nowadays the prevalent approach, since labelling send instructions with the sent qubits' reduced density operator, as in~\cite{lalireprocess2004, lalireRelationsQuantumProcesses2006}, misrepresents the effects of entanglement (see~\cite{ceragioliQuantumBisimilarityBarbs2024} for a counterexample).
The pLTS semantics is used \emph{as it is} in most related works~\cite{lalireprocess2004,gaycommunicating2005,fengprobabilistic2007};
however, this state-based approach has been ruled out in the most recent developments of qCCS and CQP, because of the problem of~\autoref{ex:broken-nondet}.

The author of~\cite{davidsonformal2012} proposed a novel semantics for CQP based upon \emph{mixed configurations}, i.e.\ single processes associated with distributions of classical and quantum values.
Mixed configurations encode private invisible information, like  measurement outcomes; they are resolved to classical, visible probabilistic distributions when the measurement outcome is revealed by sending it on some visible channel.
Indeed, this solves the problem of~\autoref{ex:broken-nondet}.
However, the calculus does not include any form of classical control primitive, and an extension of the semantic model in this sense seems challenging because of the formal constraints of mixed-configurations, making the semantics not fully adequate for modelling the behaviour of real-world protocols.

For qCCS, the authors of~\cite{fengtoward2015-1} have proposed an LTS of \emph{full} probability distributions of configurations, obtained by lifting the semantics in a standard manner~\cite{hennessyexploring2012}.
This approach simplifies working with quantum values, since the environment is made explicit by combining the visible values in a weighted sum of density operators.
In addition, lifted semantics forbid composing moves that choose different visible actions (or that choose not to move), with the objective of preventing a form of leakage of quantum data similar to the one we discuss in~\autoref{ex:broken-nondet}.
A known drawback of lifting a pLTS to a LTS of full distributions is that
actions available only in some processes of the support are in deadlock.
This requires ad-hoc solutions when defining behavioural equivalences, like imposing some form of decomposability~\cite{hennessyexploring2012}.
The most recent proposal for qCCS solves this problem by lifting to (not necessarily full) distributions (often named subdistributions in related works)~\cite{dengbisimulations2018}.
However, this works well only for visible actions and in action-deterministic transition systems: the proposal still permits invisible $\tau$-moves to access quantum data without performing measurements, and in general different choices that exhibit the same visible label are unconstrained.
As we will discuss later, the proposed labelled bisimilarity based on this semantics still suffers from some observational problems linked to unconstrained non-determinism.

None of the previous approaches considers schedulers, hence allowing the composition of moves obtained through different choices.
Instead, in our lifting, we obtain a constrained version of decomposability
by removing from the support the configurations that cannot evolve with the given scheduler choice.

The original presentation of lqCCS~\cite{ceragioliQuantumBisimilarityBarbs2024} comes only with a reduction-based semantics, while here we present a labelled one.
In addition, the syntactic locus of the process is used in place of explicit tags, and
only the moves of distinguishing contexts are constrained.
Here, we use syntactic tags and schedulers, which is a standard and more versatile approach, and we
constrain non-determinism in both processes and contexts with the same technique.
A limitation of the current proposal is that the schedulers considered in~\cite{ceragioliQuantumBisimilarityBarbs2024} are more expressive than the ones considered here.
In the present paper, we prefer a minimal approach, favouring well-behaved simple schedulers, and proving the equivalence of labelled and reduction characterizations of indistinguishability.
An extension to more sophisticated scheduling policies is left for future work.

Regarding our quantum-distribution based semantics, other works have proposed similar generalization of probability distributions in the setting of quantum protocols and process algebras.
A symbolic version of the pLTS semantics of qCCS is proposed in~\cite{fengsymbolic2014}, which uses boolean guards as additional labels and superoperators as weights.
Similarly, in~\cite{concur2024,ACT2025}, we give a semantic model for a minimal process calculus by moving from pLTSs to $\mathbb{E}$LTSs with weights from an effect algebra $\mathbb{E}$, discussing different notions of bisimilarity in the case of quantum effects~\cite{heinosaariMathematicalLanguageQuantum2011}.
The objective of these works is to model several instantiations of the same process with different quantum data in a symbolic way, therefore simplifying verification.
Although these promising contributions address the important problem of verifying protocols without assuming a specific initial quantum state, their proposed semantics are not lifted, and they do not consider schedulers nor other solutions to the problem of non-determinism.
As a result, they suffer from the same adequacy problems of the standard probabilistic approach, currently impeding their concrete applications to real-world protocols.

\subsection{Equivalence Relations and Other Verification Approaches}

The first proposed behavioural equivalences for QPAlg and qCCS were based on the standard probabilistic labelled bisimilarity,
adapted to the quantum setting by requiring bisimilar configurations to agree on visible qubits, including the ones sent on open channels.
Different works follow this approach; we focus hereafter on the most recent and well discussed, being the weak bisimilarity for qCCS of~\cite{fengbisimulation2012}.
An equivalence between this labelled bisimilarity and a contextual equivalence is proved in~\cite{dengopen2012}, where the authors also characterize them in terms of a modal logic.
In addition, a symbolic bisimilarity is used to prove equivalence of processes for any possible initial state~\cite{fengsymbolic2014}.
Finally, a probabilistic branching bisimulation \cite{long24} based on $\epsilon$-trees is given, that is a congruence for parallel composition.

This kind of bisimilarity was found too strict to adhere to the prescriptions of quantum theory~\cite{kubotasemi-automated2016,fengtoward2015-1}.
The proposed counterexample is similar to~\autoref{ex:broken-nondet}, where two mixtures of qubits that should be equated are distinguished because each element is considered in isolation.
To solve this problem, the most recent labelled bisimilarities for qCCS are based on a (possibly not full) distribution semantics~\cite{dengbisimulations2018} (or, equivalently, on full distributions with transition-consistent decompositions~\cite{fengtoward2015-1}).
An approximate version of this bisimulation is also discussed, which equates processes if they behave the same with  high enough probability, and it is applied to several real-world protocols~\cite{vitale}.
As previously stated, lifted semantics work well for labelled bisimilarity, because they allow comparing the visible qubits of sub-distributions as a whole, instead of comparing the states of single configurations independently.
The bisimilarity proposed in~\cite{dengbisimulations2018} correctly handles the identified counterexample (as well as~\autoref{ex:broken-nondet}), but falls short with similar cases, like~\autoref{ex:broken-ftl}.
As discussed in~\autoref{sec:sat}, the reason is actually the semantic model, which does not forbid ill-defined moves.
No general property is proved about the adherence with quantum theory,
nor a characterization in terms of distinguishing observers is made explicit.
Finally, the proposed bisimilarity is not a
congruence with respect to the parallel operator, even when reduced to its strong version, as shown in~\cite{ceragioliQuantumBisimilarityBarbs2024}.

Based on the mixed distribution semantics,~\cite{davidsonformal2012} proposed a branching bisimilarity for CQP that is a
congruence for the parallel operator.
The proposed bisimilarity correctly solves the problems of ~\autoref{ex:broken-nondet} and~\autoref{ex:broken-ftl}.
However, the language does not feature any classical control primitive, impeding its application to the cases considered here.

We employ the same sub-distribution approach of~\cite{dengbisimulations2018} for our strong labelled bisimilarity,
but we also discuss the saturated approach, showing consistency and thus congruence with respect to the parallel operator (\autoref{thm:demonicfa}).
Moreover, we prove that our version correctly relates processes acting on indistinguishable quantum states in general (\autoref{thm:propertyA}).
We restrict our focus to strong bisimilarity for simplicity and postpone to future work the investigation on weak and branching versions.

The choice of saturated bisimilarity as the touchstone behavioural equivalence is introduced in~\cite{ceragioliQuantumBisimilarityBarbs2024}, where it is shown that unconstrained non-deterministic contexts do not comply with the observational limitations of quantum theory.
The paper presents a different semantics for processes and contexts, constraining non-determinism only in the latter.
Here we treat them uniformly by adding tags and constraining non-determinism in both of them, obtaining a congruence with respect to parallel composition.
Moreover, we additionally propose an equivalent labelled bisimilarity in terms of a quantum-distributions based semantics,
explicitly representing the observable properties of concurrent systems.

Tagged processes have been introduced in~\cite{chatzikokolakisMakingRandomChoices2007}, and are vastly used in
probabilistic systems for characterizing which choices are ``admissible'' or ``realistic''~\cite{canettiTimeBoundedTaskPIOAsFramework2006,andresInformationHidingProbabilistic2011,songDecentralizedBisimulationMultiagent2015}.
Our tags are reminiscent of the ones used by~\cite{chatzikokolakisMakingRandomChoices2007,chatzikokolakisBisimulationDemonicSchedulers2009} to prevent schedulers from choosing a move based on private data.
In the same works, the authors show that different tagging policies correspond to different classes of schedulers.
Our usage of tags is somewhat similar, as we use them to impose limitations on what can affect the choice of schedulers, but our constraints are motivated by
the physical limitations prescribed by quantum theory instead of secrecy assumptions.
However, our schedulers lack conditionals, which are present in~\cite{chatzikokolakisMakingRandomChoices2007}.
This is intended as we focus on simple, constrained schedulers and not general ones.

Some works address the verification of quantum protocols with techniques different from bisimilarities.
Session types for quantum systems are proposed in~\cite{sessiontypes}.
Their objective is to guarantee communication safety properties, such as deadlock-freedom, as well as conformance with the no-cloning theorem.
Probabilistic model checkers have been applied to verify quantum cryptographic protocols~\cite{Gay2005ProbabilisticMO}.
A model checker designed specifically for quantum systems (QMC) is proposed in~\cite{QMC}, and a translation from CQP processes to the QMC description language is given in~\cite{CQPtoQMC}, allowing for the automatic verification of properties expressed in a quantum version of the Computational Tree Logic.
Both model checker approaches study the execution of protocols on their own, or in the presence of a known attacker.
In contrast, we consider simpler properties, but allow protocols to interact with an unspecified external environment.
The verification techniques for checking equivalence between concurrent processes proposed by~\cite{Gay18} allow comparing protocols that are functional, in the sense of computing a deterministic input-output relation.
By computing the superoperator describing the semantics of each possible interleaving arising from concurrency in the system, the authors reduce equivalence checking for a process calculus similar to qCCS to the comparison of a set of sequential programs.
The equivalence is implemented in a tool called QEC (Quantum Equivalence Checker), which is also applied to several real-world quantum protocols.
The proposed approach does not consider open systems that can receive any value from the environment, and the only non-determinism that is modelled is the pure interleaving of sequential programs.
In contrast, non-determinism and unrestricted communication is central in our investigation, as we consider it essential when reasoning about the properties of protocols.
In some exploratory works, we proposed both a trace semantics and a testing equivalence for lqCCS~\cite{rocco,wadt}.
Although preliminary, these relations demonstrate that ill-defined non-deterministic moves have severe impact on the correctness of behavioural equivalences along all the linear-time--branching-time spectrum~\cite{spectrum}, and that the definition of an operational semantics grounded in quantum theory is a main concern for verifying quantum systems.

\section{Conclusion}\label{sec:conc}
To address the problem of verifying quantum protocols, we introduced a novel, rigorous framework designed specifically for concurrent and distributed quantum processes.
This framework is built upon three core components: an expressive description language in terms of the lqCCS process calculus, an adequate semantics based on feasible schedulers, and a behavioural equivalence capable of verifying expected protocol properties.
Our process calculus seamlessly integrates classical communication, parallelism, and non-determinism with quantum capabilities, making it capable of expressing real-world protocols.

Previous approaches are shown to spuriously discriminate indistinguishable processes, in contrast with the prescriptions of quantum mechanics.
To address this problem, we introduced a labelled version of lqCCS, enriched by tag-based schedulers.
Resorting to simple schedulers allowed us to constrain processes so that they only perform physically admissible choices, i.e. those independent of the quantum states.
This suffices for making our operational semantics compliant with
the limited observational power prescribed by quantum theory,
ensuring that quantum values can only be inspected by causing measurement-induced disturbance.
As a result of that, we were capable of lifting a known result linked to the uncertainty principle, from equivalent qubit values to bisimilar lqCCS processes that cannot be distinguished by any observer (\autoref{thm:propertyA}).

Finally, we characterized the atomic observable properties of lqCCS by deriving a labelled bisimilarity, provably discriminating all and only those processes that can be operationally distinguished by some feasible observer (\autoref{thm:demonicfa}).
This makes our behavioural equivalence also a congruence for the parallel operator.
Moreover, we identified a rich class of lqCCS processes where labelled bisimilarity is checkable, without requiring any universal quantification (\autoref{thm:demonicgfa}).
To showcase our approach, we analysed three real-world protocols: superdense coding, quantum teleportation and quantum coin flip.

\paragraph*{Future Work}

A clear desideratum of a behavioural equivalence is decidability.
Our labelled bisimilarity goes in this direction, thereby avoiding the need to compare processes under every possible context.
However, to compare two protocols, it is still required to consider an infinite number of initial quantum states.
In line with~\cite{concur2024,ACT2025}, we will investigate symbolic versions of our semantics.
The goal is comparing lqCCS processes in a state-independent fashion, to verify if
they are bisimilar when instantiating the quantum input with any value.
Moreover, given that qubits received from the environment may be essentially in any possible state, we expect that this will also help us 
to characterize the superoperators closure in a decidable fashion.

Our schedulers are admittedly simple: all configurations in the support of a distribution must choose the same tag.
A simple extension would allow processes to resolve non-deterministic choices based on classical information, making our system comparable with the one in~\cite{ceragioliQuantumBisimilarityBarbs2024}.
Scheduled bisimilarities like ours, where processes must behave the same when choosing the same tag, are sometimes called \emph{demonic};
in contrast, an \emph{angelic} bisimilarity uses tags only to compute feasible moves.
We will investigate this different flavour of behavioural equivalence, which relieves the user from guessing the correct annotation of processes when tag names do not carry relevant information on the intended behaviour.

\begin{acks}
	Supported by the European Union through the MSCA-SE project QCOMICAL (Grant Agreement ID: 101182520).
\end{acks}
\bibliography{references}

\clearpage
\appendix

\section{Proofs Probability Distributions}

\begin{lemma}\label{lem:llift_normal}
	For any $\Delta, \Theta \in \dist{S}$ and relation $\mathord{\rel} \subseteq S \times \dist{S}$, $\Delta \mathrel{\lift(\rel)} \Theta$ if and only if there is a finite index set $I$ such that $\Delta = \osum_{i \in I} \delem{p_i}{s_i}$, $\Theta = \osum_{i \in I} \delem{p_i}{\Theta_i}$ and $s_i \rel \Theta_i$.
\end{lemma}
\begin{proof}
	($\Leftarrow$) By definition of $\lift$, $\delem{1}{s_i} \mathrel{\lift(\rel)} \Theta_i$.
	Then, by definition of left linear relation it holds that $\osum_{i \in I} \delem{p_i}{s_i} \mathrel{\lift(\rel)} \osum_{i \in I} \delem{p_i}{\Theta_i}$.

	($\Rightarrow$) By rule induction.

	If $\Delta = \delem{1}{s}$, then it follows trivially.

	If $\Delta = \osum_{i \in I} \delem{p_i}{\Delta_i}$ and $\Theta = \osum_{i \in i} \delem{p_i}{\Theta_i}$ then by hypothesis $\Delta_i \mathrel{\lift(\rel)} \Theta_i$ for any $i \in I$ and $\osum_i \delem{p_i}{\Delta_i}$ is defined.
	By inductive hypothesis $\Delta_i = \osum_{j \in J_i} \delem{p_{ij}}{s_{ij}}$ and $\Theta_i = \osum_{j \in J_i} \delem{p_{ij}}{\Theta_{ij}}$ with $s_{ij} \rel \Theta_{ij}$.
	Then $\Delta = \osum_{i \in I}\osum_{j \in J_i}\delem{p_ip_{ij}}{s_{ij}}$ and $\Theta = \osum_{i \in I}\osum_{j \in J_i}\delem{p_ip_{ij}}{\Theta_{ij}}$, therefore it suffices to take the set $IJ = \{ij | i \in I, j \in J_i\}$ as indices with $\{p_ip_{ij} | i \in I, j \in J_i\}$ as probabilities.
	We still need to check definiteness.
	By hypothesis $\Delta$ defined means that $\sum_{i, j \in IJ}p_ip_{ij} \leq 1$.
	We prove that $\mass{\Theta} \leq 1$.
	$\mass{\Theta} = \sum_{i,j \in IJ}p_ip_{ij}\mass{\Theta_{ij}} \leq \sum_{i,j \in IJ}p_ip_{ij}\alpha \leq 1$, where $\alpha = \max_{i,j \in IJ} \mass{\Theta_{ij}}$ that by inductive hypothesis are all less than $1$ and thus $\alpha \leq 1$.
\end{proof}

\lliftDecomp*
\begin{proof}
	Assume $\Delta = \osum_{i \in I} \delem{p_i}{\Delta_i}$ and $\Delta \mathrel{\lift(\rel)} \Theta$.
	By~\autoref{lem:llift_normal} there is a family $J$ such that $\Delta = \osum_{j \in J} \delem{q_j}{s_j}$, $\Theta = \osum_{j \in J} \delem{q_j}{\Theta_j}$ with $s_j \rel \Theta_j$.
	Note that $\Delta_i = \osum_{s \in \supp{\Delta_i}} \delem{\Delta_i(s)}{\osum_{j \in J | s = s_j}\delem{\frac{q_j}{\Delta(s)}}{s_j}}$.
	Thus, take $\Theta_i$ defined as $\Theta_i = \osum_{s \in \supp{\Delta_i}} \delem{\Delta_i(s)}{\osum_{j \in J | s = s_j}\delem{\frac{q_j}{\Delta(s)}}{\Theta_j}}$.
	Since $s_j \rel \Theta_j$, by left linearity it follows that $\Delta_i \mathrel{\lift(\rel)} \Theta_i$ for any $i \in I$.

	We can simply compute that $\Theta = \osum_{i \in I} \delem{p_i}{\Theta_i}$
	\begin{align*}
		\osum_{i \in I} \delem{p_i}{\Theta_i} & = \osum_{s \in \supp{\Delta}}\osum_{i\in I} \delem{p_i\Delta_i(s)}{\osum_{j\in J\mid s=s_j} \delem{\frac{q_j}{\Delta(s)}}{\Theta_j}} \\
		                                      & = \osum_{s \in \supp{\Delta}}\delem{\Delta(s)}{\osum_{j\in J\mid s=s_j} \delem{\frac{q_j}{\Delta(s)}}{\Theta_j}}                     \\
		                                      & = \osum_{s \in \supp{\Delta}}\osum_{j\in J\mid s=s_j} \delem{q_j}{\Theta_j}                                                          \\
		                                      & = \osum_{j\in J} \delem{q_j}{\Theta_j}                                                                                               \\
		                                      & = \Theta\qedhere
	\end{align*}
\end{proof}


%
\section{Proofs about Type and Semantics}
\uniquetype*
\begin{proof}
	By induction on the derivation $\Sigma \vdash P$ and $\Sigma' \vdash P$.
	\begin{description}
		\item[Nil] It must be that $P = \nil_{\tilde{e}}$ for some sequence of qubits $\tilde{e}$, therefore $\tilde{e} \in \tilde{\Sigma}$ and $\tilde{e} \in \tilde{\Sigma}'$.
		      Since $\tilde{e}$ must contain all elements in both $\Sigma$ and $\Sigma'$ it must be that $\Sigma = \Sigma'$.
		\item[QRecv] By inductive hypothesis $\Sigma \cup \{e\} = \Sigma' \cup \{e\}$.
		      However, by $e \notin \Sigma$ and $e \notin \Sigma'$ it must be that $\Sigma = \Sigma'$.
		\item[QSend] It must be that $P = c!q.R$ for a qubit $q$, therefore both $q \in \Sigma$ and $q \in \Sigma'$.
		      By inductive hypothesis $\Sigma \setminus \{q\} = \Sigma' \setminus \{q\}$.
		\item[Par] By hypothesis of the rule $\Sigma = \Sigma_1 \cup \Sigma_2$ and $\Sigma' = \Sigma'_1 \cup \Sigma'_2$ with $\Sigma_1 \cap \Sigma_2 = \emptyset$ and $\Sigma'_1 \cap \Sigma'_2 = \emptyset$.
		      Since it must be that $P = Q \parallel R$, by inductive hypothesis $\Sigma_1 = \Sigma'_1$ and $\Sigma_2 = \Sigma'_2$, therefore $\Sigma = \Sigma'$.
		\item[Otherwise] By application of the inductive hypothesis.
	\end{description}
\end{proof}

\begin{lemma}\label{lem:quantumsubst}
	If $\Sigma \cup \{x\} \vdash P$ and $q \not\in \Sigma$ then $\Sigma \cup \{q\} \vdash P[\sfrac{q}{x}]$.
\end{lemma}
\begin{proof}
	By structural induction on the typing derivation of $\Sigma \cup \{x\} \vdash P$.

	\begin{description}
		\item[Nil] It must be that $P = \nil_{\tilde{e}}$ with $x \in \tilde{e}$.
		      By \textsc{Nil}, trivially $\Sigma \cup \{q\} \vdash \nil_{\tilde{e}}[\sfrac{q}{x}]$.
		\item[Sop] It must be that $P = \E(\tilde{e}).R$ for some process $R$.

		      If $x \not\in \tilde{e}$ then $P[\sfrac{q}{x}] = \E(\tilde{e}).R[\sfrac{q}{x}]$, which is typed by \textsc{Sop} and the inductive hypothesis.

		      If $x \in \tilde{e}$ and $\tilde{e} \in \tilde{\Sigma}'$, then by substitution in the \textsc{Sop} rule, $\tilde{e}[\sfrac{q}{x}] \in \Sigma' \setminus \{x\} \cup \{q\}$.
		      Trivially $\Sigma' \setminus \{x\} \cup \{q\} \subseteq \Sigma \cup \{q\}$, and by inductive hypothesis $\Sigma \cup \{q\} \vdash R[\sfrac{q}{x}]$.
		\item[Meas] It must be that $P = \meas{\tilde{e}}{y}.R$ for some process $R$.
		      The result follows by the same reasoning of the \textsc{Sop} case, noting that the new variable $y$ cannot be the same as the to-be-substituted variable $x$ due to mismatching types.
		\item[QSend] It must be that $P = c!e.R$ for some process $R$.

		      If $e = x$ then by typing the substitution has no effect on $R$, thus the result follows trivially by \textsc{QSend}, since $(c!x.R)[\sfrac{q}{x}] = c!q.R$.

		      If $e \neq x$ then the result follows from the inductive hypothesis.
		\item[Par] It must be that $P = R \parallel Q$ for some processes $R$ and $Q$, with $\Sigma \cup \{x\} = \Sigma_1 \cup \Sigma_2$, $\Sigma_1 \cap \Sigma_2 = \emptyset$ and $\Sigma_1 \vdash R$, $\Sigma_2 \vdash Q$.
		      Without loss of generality assume that $x \in \Sigma_1$.
		      By inductive hypothesis $\Sigma_1 \setminus \{x\} \cup \{q\} \vdash R[\sfrac{q}{x}]$ and $\Sigma_2 \vdash Q[\sfrac{q}{x}]$ since the substitution has no effect on $Q$ because the variable $x$ is not in $\Sigma_2$ and thus does not appear in $Q$.
		      The result follows by the \textsc{Par} rule.
		\item[Otherwise] By application of the inductive hypothesis. \qedhere
	\end{description}
\end{proof}

\quasitypingpreservation*
\begin{proof}
	Consider the function $f$ which takes as input a typing context and an action
	\[
		f(\Sigma, \mu) = \begin{cases}
			\Sigma \cup \{ q \}      & \text{if $\mu = c?q$ with $c : \chtype{\qtype}$} \\
			\Sigma \setminus \{ q \} & \text{if $\mu = c!q$ with $c : \chtype{\qtype}$} \\
			\Sigma                   & \text{otherwise}                                 \\
		\end{cases}
	\]
	Let $(\Sigma'', \Sigma) \vdash \mathcal{C}$ and $\mathcal{C} \moveto[\mu]{s} \Delta$. We will prove that $(\Sigma'', f(\Sigma, \mu)) \vdash \Delta$.

	By induction on the derivation of $\moveto[\mu]{s}$.
	\begin{description}
		\item[ParL] It must be that $\mathcal{C} = \conf{\rho, P_1 \parallel P_2}$, for some processes $P_1$ and $P_2$ such that $\Sigma = \Sigma_1 \cup \Sigma_2$, $\Sigma_1 \cap \Sigma_2 = \emptyset$ and $(\Sigma'', \Sigma_1) \vdash P_1$.
		      By induction $(\Sigma'', f(\Sigma_1,\mu)) \vdash \Delta$, with $\conf{\rho, P_1} \moveto[\mu]{s} \Delta$.
		      We must verify that $(\Sigma'', f(\Sigma_1, \mu) \cup \Sigma_2) \vdash \Delta \parallel P_2$ with $f(\Sigma_1,\mu) \cap \Sigma_2 = \emptyset$.
		      If $\mu = \tau$ or $\mu = c!v,c?v$ with $c$ a classical channel, then $f(\Sigma_1, \tau) = \Sigma_1$, which trivially satisfies the requirements.
		      If $\mu = c!v$ with $c : \chtype{\qtype}$ then $f(\Sigma_1,c!v) = \Sigma_1 \setminus \{v\}$, but by the \textsc{QSend} rule $v \in \Sigma_1$ and by \textsc{Par} $v \not\in \Sigma_2$, thus the requirements hold.
		\item[ParR] Symmetrically to the \textsc{ParL} case.
		\item[SynchL] It must be that $\mathcal{C} = \conf{\rho, P \parallel Q}$, for some processes $P$ and $Q$ such that $\conf{\rho, P} \moveto[c!v]{t} \sconf{\rho, P'}$, $\conf{\rho, Q} \moveto[c?v]{t'} \sconf{\rho, Q'}$, and $\Delta = \sconf{\rho, P' \parallel Q'}$.
		      If $c$ is a classical channel then as for the parallel cases, there exists $\Sigma_1$ and $\Sigma_2$ such that $\Sigma = \Sigma_1 \cup \Sigma_2$, $\Sigma_1 \cap \Sigma_2 = \emptyset$ and $(\Sigma'',\Sigma_1) \vdash P$, $(\Sigma'', \Sigma_2) \vdash Q$.
		      By induction and application of $f$, $(\Sigma'', \Sigma_1 \setminus \{v\}) \vdash P'$ and $(\Sigma'', \Sigma_2 \cup \{v\}) \vdash Q'$.
		      As for the previous case, $v \in \Sigma_1$, thus $v \in \Sigma''$, and by \textsc{Par} $v \not\in \Sigma_2$.
		      Thus, $f(\Sigma_1, c!v) \cup f(\Sigma_2,c?v) = (\Sigma_1 \setminus \{v\}) \cup (\Sigma_2 \cup \{v\}) = \Sigma = f(\Sigma, \tau)$.
		\item[Otherwise] By induction hypothesis and trivial correspondence with their respective typing rules. \qedhere
	\end{description}
\end{proof}

\typingpreservation*
\begin{proof}
	It follows from the proof~\autoref{thm:quasipreservation}, by noticing that for $\mathcal{C} \moveto{s} \Delta$, the unique $\Sigma_\tau$ is exactly $\Sigma_P$.
\end{proof}

\begin{lemma}\label{thm:distquasipreservation}
	Let $\Sigma \subseteq \qtype$ and $\mu \in \actset$. Then there exists $\Sigma'$ such that for all $(\Sigma'', \Sigma) \vdash \Delta$ and $\Theta \in \dist{\confset}$ with $\Delta \dmoveto[\mu]{s} \Theta$ it holds $(\Sigma'', \Sigma') \vdash \Theta$.
\end{lemma}
\begin{proof}
	Let $\Delta = \osum_{i \in I} \delem{p_i}{\mathcal{C}_i}$, where every $\mathcal{C}_i$ is unique with $(\Sigma'', \Sigma) \vdash \Delta$ and $(\Sigma'', \Sigma) \vdash \mathcal{C}_i$ for all $i \in I$.
	If $I$ is empty then $\Delta = \epsilon$ and by definition of $\dmoveto[\mu]{s}$, $\Delta' = \epsilon$, which is always well-typed.
	Otherwise, since $\dmoveto[\mu]{s}$ is left decomposable then for any $i \in I$ there exist a $\Theta_i$ such that $\mathcal{C} \dmoveto[\mu]{s} \Theta_i$ and $\Theta = \osum_{i \in I} \delem{p_i}{\Theta_i}$.
	By~\autoref{thm:quasipreservation}, there exists a $\Sigma'$ such that for any $\mathcal{C}_i \moveto[\mu]{s} \Theta_i$ it holds $(\Sigma'', \Sigma') \vdash \Theta_i$, and thus also $(\Sigma'', \Sigma') \vdash \Theta$.
\end{proof}

We now prove some basic properties on our probabilistic semantics.

\begin{lemma}\label{lem:conf_equal_mu_deterministic}
	For any configuration $\mathcal{C} \in \confset$, $s \in \choiceset$, if $\mathcal{C} \cmoveto[\mu]{s} \Delta$ and $\mathcal{C} \cmoveto[\mu]{s} \Delta'$ then $\Delta = \Delta'$.
\end{lemma}
\begin{proof}
	By \autoref{ass:det} since the two labels are equals it must be that $\Delta = \Delta'$.
\end{proof}

\begin{lemma}\label{lem:bot_deterministic}
	For any configuration $\mathcal{C} \in \confset$, $s \in \choiceset$, if $\mathcal{C} \mathrel{\operatorname{bot}(\moveto[\mu]{s})} \Delta$ and $\mathcal{C} \mathrel{\operatorname{bot}(\moveto[\mu]{s})} \Delta'$ then $\Delta = \Delta'$.
\end{lemma}
\begin{proof}
	Assume $\mathcal{C} \mathrel{\operatorname{bot}(\moveto[\mu]{s})} \Delta$ and $\mathcal{C} \mathrel{\operatorname{bot}(\moveto[\mu]{s})} \Delta'$ and proceed by cases:
	if $\Delta = \epsilon$ then, by definition of $\operatorname{bot}(\blank)$ there is no other $\Delta'$ such that $\mathcal{C} \mathrel{\operatorname{bot}(\moveto[\mu]{s})} \Delta$, thus $\Delta = \Delta' = \epsilon$.
	Otherwise, it holds by \autoref{lem:conf_equal_mu_deterministic}.
\end{proof}

\determdistr*
\begin{proof}
	Let $\Delta = \osum_{i \in I}\delem{p_i}{\confel_i}$ where every $\confel_i$ is unique, and assume $\Delta \dmoveto[\mu]{s} \Theta$ and $\Delta \dmoveto[\mu]{s} \Theta'$.
	Since $\dmoveto[\mu]{s}$ and $\dmoveto[\mu']{s}$ are left decomposable there exists families $\{\Theta_i\}_{i \in I}$ and $\{\Theta'_i\}_{i \in I}$ such that $\Theta = \osum_{i \in I}\delem{p_i}{\Theta_i}$ with $\delem{1}{\confel_i} \dmoveto[\mu]{s} \Theta_i$, and $\Theta' = \osum_{i \in I}\delem{p_i}{\Theta'_i}$ with $\delem{1}{\confel_i} \dmoveto[\mu]{s} \Theta'_i$.
	By \autoref{lem:bot_deterministic} for any $i \in I$ it holds $\Theta_i = \Theta'_i$ and thus $\Theta = \Theta'$.
\end{proof}

\begin{lemma}\label{thm:alwaysmove}
	For any process $\Sigma \vdash P$, and scheduler $s$, there is a (possibly empty) family $I$ of trace non-increasing superoperators $\E_i \in \soset{\hilb_\Sigma}$ and processes $P_i$ such that for all density operator $\rho$, $\conf{\rho, P} \moveto{s} \Delta$ if and only if
	\[
		\Delta = \osum_{i \in I \text{ s.t. } \tr(\E_i(\rho)) \neq 0} \delem{\tr(\E_i(\rho))}{\conf{\frac{\E_i(\rho)}{\tr(\E_i(\rho))}, P_i}}.
	\]
\end{lemma}
\begin{proof}
	We proceed by induction on the rules of the semantics.
	\begin{description}
		\item[Sop] It must be that $P = s \tags \E(x).R$ for some process $R$ and superoperator $\E \in \tsoset{\hilb_\Sigma}$.
		      Then $\Delta = \delem{1}{\conf{\E(\rho), R}}$.
		      Note that since $\E \in \tsoset{\hilb_\Sigma}$, and $\tr(\rho) = 1$, then $\tr(\E(\rho)) = 1$. Take the singleton family with $\E$ and $R$.
		\item[Meas] It must be that $P = s \tags \meas{\tilde{q}}{x}.R$ for some process $R$ and measurement $\M = \{M_m\}_m$.
		      Then $\Delta = \osum_m \delem{\tr(M_m(\rho))}{\conf{\frac{M_m(\rho)}{\tr(\M_m(\rho))},R[\sfrac{m}{x}]}}$.
		      Take the operators $M_m$ such that $\tr(\M_m(\rho)) \neq 0$ and their respective processes $R[\sfrac{m}{x}]$.
		\item[Other non-inductive] The rules \textsc{Tau}, \textsc{TauSynch}, \textsc{Send} and \textsc{Recv} do not modify the quantum state and require $\Delta$ to be of the form $\delem{1}{\conf{\rho,P'}}$.
		      Thus, take the singleton family with the identity superoperator $\I$ and process $P'$.
		\item[Otherwise] The other cases follow trivially from the inductive hypothesis.
	\end{description}
\end{proof}

We define the following notation for the family of superoperators and processes that represents the evolution in a parametric way with respect to density operators.
\begin{definition}
	For any process $P \in \procset$, scheduler $s \in \choiceset$ and action $\mu \in \actset$,
	we say $\{\E_i, P_i\}_{i \in I} \in \nextset(P, s, \mu)$ if $\{\E_i, P_i\}_{i \in I}$ is an $I$-indexed family of pairs of superoperators $\E_i \in \soset{\hilb_{\Sigma_P}}$ with $\sum_{i \in I}\E_i \in \tsoset{\hilb_{\Sigma_P}}$, and processes $P_i \in \procset$ such that either:

	\begin{itemize}
		\item if $I = \emptyset$ then $\conf{\rho, P} \ncmoveto[\mu]{s}$ for any state $\rho$; otherwise
		\item if $I \neq \emptyset$ then $\conf{\rho, P} \cmoveto[\mu]{s} \osum_{i \in I}\delem{\tr(\E_i(\rho))}{\conf{\frac{\E_i(\rho)}{\tr(\E_i(\rho))},P_i}}$ for any state $\rho$.
	\end{itemize}
\end{definition}
Note that under \autoref{ass:det} there is only one family $I$ in $\nextset(P, s, \mu)$ for all $P$, $s$ and $\mu$, and we will write $\nextset(P, s, \mu) = \{\E_i, P_i\}_{i \in I}$ to denote such family.

Moreover, we allow applying superoperators to configurations and distributions.
Given $\E \in \soset{\hilb}$, we define
$\nE : \dist{\confset} \to \dist{\confset}$
as
$$\nE\left(\osum_i \delem{p_i}{\conf{\rho_i, P_i}}\right) = \osum_i \delem{p_i\cdot tr(\E(\rho_i))}{\conf{\frac{\E(\rho_i)}{\tr(\E(\rho_i))}, P_i}}.$$
Roughly, if $\E$ is trace-preserving it is just applied to configurations.
If $\E$ is trace non-increasing, then after its application each $\rho_i$ is normalized and the weight of the $i$-th configuration is updated accordingly.

Finally, given $\Sigma, \Sigma' \vdash \Delta$, if $\E$ is defined on $\tilde{q}$, which are only some of the qubits in $\Sigma$, we write $\E(\Delta)$ for $\E^{\tilde{q}}(\Delta)$ (recall that $\E^{\tilde{q}}$ extends $\E$ by tensoring it with the identity).

\begin{lemma}\label{lem:conf_nonbot_cases}
	For any processes $P, R \in \procset$ with $\tagset(P) \cap \tagset(R) = \emptyset$, we have $\conf{\rho, P \parallel R} \cmoveto[\mu]{s} \Delta$ if and only if one of the following holds:
	\begin{enumerate}
		\item $\Delta = \Delta' \parallel R$ with $\conf{\rho, P} \cmoveto[\mu]{s} \Delta'$ and $\tagset(s) \subseteq \tagset(P)$;
		\item $\Delta = P \parallel \Delta'$ with $\conf{\rho, R} \cmoveto[\mu]{s} \Delta'$ and $\tagset(s) \subseteq \tagset(R)$;
		\item $\Delta = \delem{1}{\conf{\rho, P' \parallel R}}$ where there exists a channel $c$ and value $v$ such that $\conf{\rho, P} \cmoveto[c?v]{t_1}\delem{1}{\rho, P'}$, $\conf{\rho, R} \cmoveto[c!v]{t_2}\delem{1}{\rho, R'}$ with $s = (t_1, t_2)$, $t_1 \in \tagset(P)$ and $t_2 \in \tagset(R)$;
		\item $\Delta = \delem{1}{\conf{\rho, P' \parallel R}}$ where there exists a channel $c$ and value $v$ such that $\conf{\rho, P} \cmoveto[c!v]{t_1}\delem{1}{\rho, P'}$, $\conf{\rho, R} \cmoveto[c?v]{t_2}\delem{1}{\rho, R'}$ with $s = (t_1, t_2)$, $t_1 \in \tagset(P)$ and $t_2 \in \tagset(R)$.
	\end{enumerate}
\end{lemma}
\begin{proof}
		A transition of the form $\conf{\rho, P \parallel R} \cmoveto[\mu]{s} \Delta$ can happen only through one of the rules \textsc{ParL}, \textsc{ParR}, \textsc{SynchL} and \textsc{SynchR}.
		Each of the four cases is a sufficient and necessary condition to apply one of the rules above.
\end{proof}

\begin{lemma}\label{lem:conf_bot_cases}
		For any processes $P$ and $R$ with $\tagset(P) \cap \tagset(R) = \emptyset$, we have $\conf{\rho, P \parallel R} \mathrel{\operatorname{bot}(\cmoveto[\mu]{s})} \epsilon$ if and only if all of the following hold:
		\begin{itemize}
			\item P cannot move, i.e. $\conf{\rho, P} \ncmoveto[\mu]{s}$;
			\item R cannot move, i.e. $\conf{\rho, R} \ncmoveto[\mu]{s}$;
			\item P and R cannot synchronize, i.e. for all $c$ and $v$, if $\mu = \tau$ and $s = (t_1, t_2)$ we have
			      $next(R, t_2, c!v) \neq \emptyset \implies \conf{\rho, P} \ncmoveto[c?v]{t_1}$ and
			      $next(R, t_2, c?v) \neq \emptyset \implies \conf{\rho, P} \ncmoveto[c!v]{t_1}$.
		\end{itemize}
\end{lemma}
\begin{proof}
		For the only if direction, we know that $\conf{\rho, P\parallel R} \ncmoveto[\mu]{s}$, and this implies that $\conf{\rho, P \parallel R}$ can not move with any of the rules \textsc{ParL} (first case), \textsc{ParR} (second case), \textsc{SynchL} and \textsc{SynchR} (third case).
		For the if direction, since all three cases hold, we know that none of the above rule can be applied, and thus $\conf{\rho, P\parallel R} \ncmoveto[\mu]{s}$ and $\conf{\rho, P \parallel R} \mathrel{\operatorname{bot}(\cmoveto[\mu]{s})} \epsilon$.
\end{proof}


\begin{lemma}\label{lem:distr_cases}
	For any non-empty distribution $\Delta \in \sdist{\confset}$ and process $R \in \procset$ with $\tagset(\Delta) \cap \tagset(R) = \emptyset$, we have $\Delta \parallel R \dmoveto[\mu]{s} \Theta \neq \epsilon$ if and only if one of the following holds:
	\begin{enumerate}
		\item $\tagset(s) \subseteq \tagset(\Delta)$ and there exists $\Theta'$ such that $\Theta = \Theta' \parallel R$ and $\Delta \dmoveto[\mu]{s} \Theta'$;
		\item $\tagset(s) \subseteq \tagset(R)$ and $\Theta = \osum_{i \in I}\E_i(\Delta \parallel R_i)$ with $\nextset(R, s, \mu) = \{\E_i, R_i\}_{i \in I}$;
		\item $s = (t_1, t_2)$, $t_1 \in \tagset(\Delta)$, $t_2 \in \tagset(R)$, and $\Theta = \Theta' \parallel R'$, $\exists c,v\ldotp \Delta \dmoveto[c?v]{t_1} \Theta'$ with $\nextset(R, t_2, c!v) = \{\I, R'\}$;
		\item $s = (t_1, t_2)$, $t_1 \in \tagset(\Delta)$, $t_2 \in \tagset(R)$, and $\Theta = \osum_{v \in V}\Theta_v \parallel R_v$, $\exists c,V \ldotp \forall v \in V\ldotp \Delta \dmoveto[c!v]{t_1} \Theta_v$ with $\nextset(R, t_2, c?v) = \{\I, R_v\}$.
	\end{enumerate}
\end{lemma}
\begin{proof}
	We start with the ``only if'' direction.
	Let $\Delta \parallel R \dmoveto[\mu]{r} \Theta$ and let $\Delta = \osum_{i \in I}\delem{p_i}{\conf{\rho_i, P_i}}$ where every $\conf{rho_i, P_i}$.
	By left-decomposability, $\delem{1}{\conf{\rho_i, P_i \parallel R}} \dmoveto[\mu]{s} \Theta_i$ for all $i \in I$ with $\Theta = \osum_{i \in I}\delem{p_i}{\Theta_i}$.
	Since $\Theta \neq \epsilon$, there is a subset $J \subseteq I$ such that $\delem{1}{\conf{\rho_j, P_j \parallel R}} \dmoveto[\mu]{s} \Theta_j$ with $\Theta_j \neq \epsilon$, i.e.\ $\conf{\rho_j, P_j \parallel R} \cmoveto[\mu]{s} \Theta_j$ for all $j \in J$.
	Consider a single $\jhat \in J$.
	By \autoref{lem:conf_nonbot_cases} we have to consider four cases:
	\begin{enumerate}
		\item If $\tagset(s) \subseteq \tagset(P_{\jhat})$, $\conf{\rho_{\jhat}, P_{\jhat}} \dmoveto[\mu]{s} \Theta'_{\jhat}$ and $\Theta_{\jhat} = \Theta'_{\jhat} \parallel R$, then $\tagset(s) \cap \tagset(R) = \emptyset$.
		      Thus, for all $i \in I$ either $\conf{\rho_i, P_i \parallel R} \ndmoveto[\mu]{s}$ so $\Theta_i = \epsilon$ and by \autoref{lem:conf_bot_cases} $\conf{\rho_i, P_i} \ncmoveto[\mu]{s}$, or $\conf{\rho_i, P_i} \cmoveto[\mu]{s} \Theta'_i$ so $\conf{\rho_i, P_i \parallel R} \cmoveto[\mu]{s} \Theta'_i \parallel R = \Theta_i$ by \autoref{lem:conf_nonbot_cases}.
		      Note that the latter are exactly all the configuration in $J$.
		      Therefore, $\Delta = \osum_{i \in I}\delem{p_i}{\conf{\rho_i, P_i}} \dmoveto[\mu]{s} \osum_{j \in J}\delem{p_j}{\Theta'_j} = \Theta'$ and $\Theta = \Theta' \parallel R = \osum_{i \in I}\delem{p_i}{\Theta_i} = \osum_{j \in J}\delem{p_j}{\Theta'_j \parallel R}$.
		\item If $\tagset(s) \subseteq \tagset(R)$, by \autoref{thm:alwaysmove}, $\nextset(R, s, \mu) = \{\E_k, R_k\}_{k \in K}$. Since $\tagset(s) \cap \tagset(\Delta) = \emptyset$, by \autoref{lem:conf_nonbot_cases}, $\conf{\rho_i, R} \cmoveto[\mu]{s} \osum_{k \in K}\delem{\tr(\E_k(\rho_i))}{\conf{\frac{\E_k(\rho_i)}{\tr(\E_k(\rho_i))}, R_k}}$ for all $i$, and thus $\conf{\rho_i, P_i \parallel R} \cmoveto[\mu]{s} \osum_{k \in K}\delem{\tr(\E_k(\rho_i))}{\conf{\frac{\E_k(\rho_i)}{\tr(\E_k(\rho_i))}, P_i \parallel R_k}}$.
		      Therefore, $\Delta \parallel R = \osum_{i \in I}\delem{p_i}{\conf{\rho_i, P_i \parallel R}} \dmoveto[\mu]{s} \osum_{i \in I}\delem{p_i}{\osum_{k \in K}\delem{\tr(\E_k(\rho_i))}{\conf{\frac{\E_k(\rho_i)}{\tr(\E_k(\rho_i))}, P_i \parallel R_k}}} = \osum_{k \in K} \E_k(\Delta \parallel R_k)$.
		\item If $s = (t_1, t_2)$, $t_1 \in \tagset(\Delta)$, $t_2 \in \tagset(R)$ and there are $c$ and $v$ such that $\nextset(R, t_2, c!v) = \{\I, R'\}$ then by \autoref{lem:conf_nonbot_cases} $\conf{\rho_{\jhat}, P_{\jhat} \parallel R} \cmoveto{s} \delem{1}{\conf{\rho_{\jhat}, P'_{\jhat} \parallel R'}}$ with $\nextset(P_{\jhat}, t_1, c?v) = \{\I, P'_{\jhat}\}$.
		      Note that $c$ and $v$ are unique due to the assumption of determinism, and by \autoref{lem:conf_nonbot_cases}, the same holds for all $j \in J$.
		      By left-linearity $\Delta \parallel R \dmoveto{s} \Theta' \parallel R'$ where $\Theta' = \osum_{j \in J}\delem{p_j}{P'_j}$.
		\item If $s = (t_1, t_2)$, $t_1 \in \tagset(\Delta)$, $t_2 \in \tagset(R)$ and there are $c$ and $v$ such that $\nextset(R, t_2, c?v) = \{\I, R_v\}$.
		      Let $J_v$ be the subset of $J$ such that $\conf{\rho_j, P_j} \cmoveto[c!v]{t_1} \delem{1}{\conf{\rho_j, P'_j}}$, and let $\Theta_v = \osum_{j \in J_v}\delem{p_j}{\conf{\rho_j, P'_j}}$.
		      By \autoref{lem:conf_nonbot_cases} and linearity we have $\osum_{j \in J_v}\delem{p_j}{\conf{\rho_j, P_j}} \dmoveto[c!v]{t_1} \Theta_v$ and $\osum_{j \in J_v}\delem{p_j}{\conf{\rho_j, P_j \parallel R}} \dmoveto{s} \Theta_v \parallel R_v$.
		      Let $J^x$ be the subset of processes in deadlock.
		      Therefore, $\Delta \parallel R = \osum_{v \in V}\osum_{j \in J_v}\delem{p_j}{\conf{\rho_j, P_j \parallel R}} \oplus \osum_{j \in J^x}\delem{p_j}{\conf{\rho_j, P_j \parallel R}} \dmoveto{s} \osum_{v \in V}\Theta_v \parallel R_v \oplus \epsilon$.
	\end{enumerate}
	The ``if'' direction is similar: any of the four cases implies $\Delta \parallel R \dmoveto[\mu]{s} \Theta$.
\end{proof}

\begin{lemma}\label{lem:distr_epsilon_cases}
		For any non-empty distribution $\Delta \in \sdist{\confset}$ and process $R \in \procset$ with $\tagset(\Delta) \cap \tagset(R) = \emptyset$, we have $\Delta \parallel R \dmoveto[\mu]{s} \epsilon$ if and only if all of the following hold:
		\begin{itemize}
			\item $\Delta$ cannot move, i.e. $\Delta \dmoveto[\mu]{s} \epsilon$;
			\item R cannot move, i.e. $next(R, s, \mu) = \emptyset$;
			\item $\Delta$ and R cannot synchronize, i.e. for all $c$ and $v$, if $\mu = \tau$ and $s = (t_1, t_2)$ we have
			      $next(R, t_2, c!v) \neq \emptyset \implies \Delta \dmoveto[c?v]{t_1} \epsilon$ and
			      $next(R, t_2, c?v) \neq \emptyset \implies \Delta \dmoveto[c!v]{t_1} \epsilon$.
		\end{itemize}
\end{lemma}
\begin{proof}
		Let $\Delta = \osum_{i \in I}\delem{p_i}{\conf{\rho_i, P_i}}$ where every $\conf{\rho_i, P_i}$ is unique.
		By left-decomposability $\delem{1}{\conf{\rho_i, P_i \parallel R}} \dmoveto[\mu]{s} \epsilon$ for all $i \in I$.
		For all $i \in I$ we can apply \autoref{lem:conf_bot_cases}, and thus we know that \begin{itemize}
			\item For all $i \in I $, $\conf{\rho_i, P_i} \ncmoveto[\mu]{s}$;
			\item For all $i \in I $, $\conf{\rho_i, R_i} \ncmoveto[\mu]{s}$;
			\item For all $i$, $c$ and $v$, if $\mu = \tau$ and $s = (t_1, t_2)$ we have
			      $next(R, t_2, c!v) \neq \emptyset \implies \conf{\rho_I, P_i} \ncmoveto[c?v]{t_1}$ and
			      $next(R, t_2, c?v) \neq \emptyset \implies \conf{\rho_I, P_i} \ncmoveto[c!v]{t_1}$.
		\end{itemize}
		The first point is equivalent to saying that $\Delta \dmoveto[\mu]{s} \epsilon$.
		The second point holds if and only if $next(R, s, \mu) = \emptyset$, thanks to \autoref{thm:alwaysmove}.
		If we move the quantification on all the $i$s into the implication, the third point is equivalent to
		$next(R, t_2, c!v) \neq \emptyset \implies \Delta \dmoveto[c?v]{t_1} \epsilon$ and
		$next(R, t_2, c?v) \neq \emptyset \implies \Delta \dmoveto[c!v]{t_1} \epsilon$.
\end{proof}

\determpar*
\begin{proof}
	Let $\rho$ be a quantum state and assume $\conf{\rho, P \parallel R} \cmoveto[\mu]{s} \Delta$ and $\conf{\rho, P \parallel R} \cmoveto[\mu']{s} \Delta'$.
	By cases on $s$.

	If $\tagset(s) \subseteq \tagset(P)$, by \autoref{thm:alwaysmove} there exists $\Delta_P$ and $\Delta'_P$ such that $\Delta = \Delta_P \parallel R$ with $\conf{\rho, P} \cmoveto[\mu]{s} \Delta_P$ and $\Delta' = \Delta'_P \parallel R$ with $\conf{\rho, P} \cmoveto[\mu]{s} \Delta'_P$.
	By \autoref{ass:det} either $\mu = \mu'$ and $\Delta_P = \Delta'_P$ or $\mu = c?v$ and $\mu' = c?v'$.
	In either case the requirements of \autoref{ass:det} are trivially satisfied by $P \parallel R$.

	The case for $\tagset(s) \subseteq \tagset(R)$ follows the same line of reasoning.

	If $\tagset(s) \nsubseteq \tagset(P \parallel R)$ then the process must be in deadlock.

	If $\tagset(s) \subseteq \tagset(P \parallel R)$ but $\tagset(s) \nsubseteq \tagset(P)$ and $\tagset(s) \nsubseteq \tagset(R)$, then $s = (t_1, t_2)$ with $t_1 \in \tagset(P)$ and $t_2 \in \tagset(R)$.
	By the form of the process and the form of the scheduler the only applicable semantic rules are either \textsc{SynchL} or \textsc{SynchR}.
	Without loss of generality assume $\conf{\rho, P \parallel R} \cmoveto[\mu]{s} \Delta$ by \textsc{SynchL}.
	Therefore, it must be that $\mu = \tau$, $\conf{\rho, P} \cmoveto[c!v]{t_1} \conf{\rho, P_{\text{cont}}}$, $\conf{\rho, R} \cmoveto[c?v]{t_2} \conf{\rho, R_{\text{cont}}}$ and $\Delta = \conf{\rho, P_{\text{cont}} \parallel R_{\text{cont}}}$ for some channel $c$, value $v$ and processes $P_{\text{cont}}$, $R_{\text{cont}}$.
	Then it must be that also $\conf{\rho, P \parallel R} \cmoveto[\mu']{s} \Delta'$ by \textsc{SynchL}, because if it moves by \textsc{SynchR} then $\conf{\rho, P} \cmoveto[d?w]{t_1} \Delta''$, meaning that with the process $P$ can perform two moves with distinct labels with the same scheduler $t_1$, violating \autoref{ass:det}.
	Thus, $\mu' = \tau$, $\conf{\rho, P} \cmoveto[c'!v']{t_1} \conf{\rho, P'_{\text{cont}}}$, $\conf{\rho, P} \cmoveto[c'?v']{t_2} \conf{\rho, R'_{\text{cont}}}$ and $\Delta' = \conf{\rho, P'_{\text{cont}} \parallel R'_{\text{cont}}}$ for some channel $c'$, value $v'$ and processes $P'_{\text{cont}}$, $R'_{\text{cont}}$.
	By \autoref{ass:det} and since $c!v$ and $c'!v'$ are both not receptions, $P_{\text{cont}} = P'_{\text{cont}}$ and $c!v = c'!v'$, thus it must be that $c?v = c'?v'$ and $R_{\text{cont}} = R'_{\text{cont}}$.
\end{proof}

\section{Proofs about Saturated Bisimulations}
\begin{definition}
	Let $\mathord{\rel} \subseteq \sdist{\confset} \times \sdist{\confset}$.
	The \emph{weak linear closure} $\lcl(\rel)$ of $\rel$ is the least relation satisfying the following rule
	\begin{mathpar}
		\inferrule
		{\forall i \in I\ldotp \Delta_i \rel \Theta_i \\
			|\textstyle\osum_{i \in I} \delem{p_i}{\Delta_i}| \leq 1 \\
			|\textstyle\osum_{i \in I} \delem{p_i}{\Theta_i}| \leq 1}
		{(\textstyle\osum_{i \in I}\delem{p_i}{\Delta_i}) \mathrel{\lcl(\rel)} (\textstyle\osum_{i \in I}\delem{p_i}{\Theta_i})}
	\end{mathpar}
\end{definition}

Notice how this definition is based on "weak linearity", meaning that the composite distributions must be related only if they are both defined.
Instead, in a left-linear relation, $|\textstyle\osum_{i \in I} \delem{p_i}{\Theta_i}| \leq 1$ is a consequence and not a requirement.
Nonetheless, if $\rel$ is mass-preserving, the weak linear closure is indeed linear.

\anote{chiusura lineare di R equimassa e' lineare, chiusura lineare di R mass non-increasing e' left-linear (rivedere simbnolo: l o lc).}
\begin{proposition}\label{thm:lcl-leftlinear}
	Let $\mathord{\rel} \subseteq \sdist{\confset} \times \sdist{\confset}$, where $\mass{\Delta} \geq \mass{\Theta}$ whenever $\Delta \rel \Theta$.
	The weak linear closure $\lcl(\rel)$ is left-linear.
\end{proposition}
\begin{proof}
	Let $I$ be the index of families $\{\Delta_i\}$ and $\{\Theta_i\}$ such that $\Delta_i \mathrel{\lcl(\rel)} \Theta_i$ for all $i \in I$.
	By hypothesis for any $i \in I$ it also holds that $\mass{\Delta_i} \geq \mass{\Theta_i}$.
	Take $\Delta = \osum_{i \in I}\delem{p_i}{\Delta_i}$ and $\Theta = \osum_{i \in I}\delem{p_i}{\Theta_i}$ with $\mass{\Delta} \leq 1$.
	By definition of mass:
	\begin{align*}
		1 \geq \mass{\Delta}
		= \mass{\osum_{i \in I} \delem{p_i}{\Delta_i}}
		= \sum_{i \in I}{p_i}{\mass{\Delta_i}}
		\geq \sum_{i \in I}{p_i}{\mass{\Theta_i}}
		= \mass{\osum_{i \in I} \delem{p_i}{\Theta_i}}
		= \mass{\Theta}
	\end{align*}
	Since $\Theta$ is defined, we have that $\Delta \mathrel{\lcl(\rel)} \Theta$.
\end{proof}

\begin{proposition}\label{thm:lcl-linear}
	Let $\mathord{\rel} \subseteq \sdist{\confset} \times \sdist{\confset}$, where $\mass{\Delta} = \mass{\Theta}$ whenever $\Delta \rel \Theta$.
	The linear closure $\lcl(\rel)$ is linear.
\end{proposition}
\begin{proof}
	Let $I$ be the index of families $\{\Delta_i\}$ and $\{\Theta_i\}$ such that $\Delta_i \mathrel{\lcl(\rel)} \Theta_i$ for all $i \in I$.
	By hypothesis for any $i \in I$ it also holds that $\mass{\Delta_i} = \mass{\Theta_i}$.
	Take $\Delta = \osum_{i \in I}\delem{p_i}{\Delta_i}$ and $\Theta = \osum_{i \in I}\delem{p_i}{\Theta_i}$ with $\mass{\Delta} \leq 1$.
	By definition of mass:
	\begin{align*}
		1 \geq \mass{\Delta}
		= \mass{\osum_{i \in I} \delem{p_i}{\Delta_i}}
		= \sum_{i \in I}{p_i}{\mass{\Delta_i}}
		= \sum_{i \in I}{p_i}{\mass{\Theta_i}}
		= \mass{\osum_{i \in I} \delem{p_i}{\Theta_i}}
		= \mass{\Theta}
	\end{align*}
	Since $\Theta$ is defined, we have that $\Delta \mathrel{\lcl(\rel)} \Theta$.
	The same reasoning also holds if we assume that $\mass{\Theta} \leq 1$, meaning that $\lcl(\rel)$ is also right-linear and thus linear.
\end{proof}

\begin{definition}
	Let $\mathord{\rel} \subseteq \sdist{\confset} \times \sdist{\confset}$.
	The \emph{context closure} $\ccl(\rel)$ of $\rel$ is the least relation satisfying the following rule
	\begin{mathpar}
		\inferrule{\Delta \rel \Theta}{B[\Delta] \mathrel{\ccl(\rel)} B[\Theta]}
	\end{mathpar}
	where $\Sigma \vdash \Delta, \Theta$ and $B[\blank]$ has a hole of type $\Sigma$.
\end{definition}

\begin{definition}
	A relation $\mathord{\rel} \subseteq \sdist{\confset} \times \sdist{\confset}$ is a \emph{saturated bisimulation up-to} $\lcl \circ \ccl$ if $\Delta \rel \Theta$ implies $(\Sigma, \Sigma') \vdash \Delta, \Theta$ for some $\Sigma$, $\Sigma'$, $\mass{\Delta} = \mass{\Theta}$ and for any context $B[\blank]$ it holds
	\begin{itemize}
		\item whenever $B[\Delta] \dmoveto{s} \Delta'$, there exists $\Theta'$ such that $B[\Theta] \dmoveto{s} \Theta'$ and $\Delta' \mathrel{\lcl(\ccl(\rel))} \Theta'$;
		\item whenever $B[\Theta] \dmoveto{s} \Theta'$, there exists $\Delta'$ such that $B[\Delta] \dmoveto{s} \Delta'$ and $\Delta' \mathrel{\lcl(\ccl(\rel))} \Theta'$.
	\end{itemize}
\end{definition}

We now prove that bisimulation up-to $\lcl \circ \ccl$ is \emph{valid}, i.e.\ that we can use it to prove bisimilarity.
We define the function $b$ over relations, of which saturated bisimilarity is the greatest fixpoint.
\[
	b(\rel) \coloneq \left\{(\Delta, \Theta)\ \middle| \begin{array}{c}
		\mass{\Delta} = \mass{\Theta}                                                                                           \\
		B[\Delta] \dmoveto{s} \Delta' \Rightarrow \exists\Theta'\ldotp B[\Theta] \dmoveto{s} \Theta' \land \Delta' \rel \Theta' \\
		B[\Theta] \dmoveto{s} \Theta' \Rightarrow \exists\Delta'\ldotp B[\Delta] \dmoveto{s} \Delta' \land \Delta' \rel \Theta'
	\end{array} \right\}
\]
Observe that $b$, $\lcl$ and $\ccl$ are monotone functions on the lattice of relations.

\begin{lemma}[$\lcl$ is $b$-compatible]\label{lem:ll_compatible}
	Let $\mathord{\rel} \subseteq \sdist{\confset} \times \sdist{\confset}$. $\lcl(b(\rel)) \subseteq b(\lcl(\rel))$.
\end{lemma}
\begin{proof}
	Assume $(\Delta, \Theta) \in \lcl(b(\rel))$.
	By definition of $\lcl$ and $b$, $\Delta = \osum_{i \in I}\delem{p_i}{\Delta_i}$, $\Theta = \osum_{i \in I}\delem{p_i}{\Theta_i}$ with $\Delta_i \mathrel{b(\rel)} \Theta_i$ and $\mass{\Delta_i} = \mass{\Theta_i}$ for any $i \in I$ and $\mass{\Delta} \leq 1$.
	Then $\mass{\Delta} = \mass{\Theta}$.
	Take any context $B[\blank]$, $B[\Delta] = \osum_{i \in I}\delem{p_i}{B[\Delta_i]} \dmoveto{s} \Delta'$.
	By decomposability of $\dmoveto{s}$, for any $i \in I$, $B[\Delta_i] \dmoveto{s} \Delta'_i$, and $\osum_{i \in I}\delem{p_i}{\Delta'_i} = \Delta'$ with $\mass{\Delta'} \leq 1$.
	Since $\Delta_i \rel \Theta_i$, for all $i \in I$ there exists $\Theta'_i$ such that $B[\Theta_i] \dmoveto{s} \Theta'_i$.
	By linearity of $\dmoveto{s}$ and by definition of $\lcl$, $B[\Theta] \dmoveto{s} \osum_{i \in I}\delem{p_i}{\Theta'_i} = \Theta'$ where $\mass{\Theta'}$ must be less than one and $\Delta' \mathrel{\lcl(\rel)} \Theta'$.
	The argument is symmetric for the last condition.
\end{proof}

\begin{lemma}[$\ccl$ is $b$-compatible]
	Let $\mathord{\rel} \subseteq \sdist{\confset} \times \sdist{\confset}$. $\ccl(b(\rel)) \subseteq b(\ccl(\rel))$.
\end{lemma}
\begin{proof}
	Assume $(B[\Delta], B[\Theta]) \in \ccl(b(\rel))$.
	If $B[\blank] = [\blank]$ then the proof is trivial ($(\Delta, \Theta) \in b(\rel) \subset b(\ccl(\rel))$ and $b(\rel) \subseteq b(\ccl(\rel))$ as $b$ is monotone and $\rel \subseteq \ccl(\rel)$).
	Let $B[\blank] = [\blank] \parallel P$.
	Since the contexts is applied linearly, for any $\Delta$, $\mass{B[\Delta]} = \mass{\Delta}$, therefore $\mass{B[\Delta]} = \mass{\Delta} = \mass{\Theta} = \mass{B[\Theta]}$.
	For the other requirements of $b$, let $B'[\blank]$ be any context.
	If $B'[\blank] = [\blank]$ then the proof is again trivial.
	Thus, let $B'[\blank] = [\blank] \parallel R$.
	We have that $B'[B[\Delta]] = \Delta \parallel P \parallel R$ and that $B'[B[\Theta]] = \Theta \parallel P \parallel R$.
	By assumption, we have that $\Delta \mathrel{b(\rel)} \Theta$.
	Therefore with the context $B''[\blank] = [\blank] \parallel P \parallel R$,
	if $B''[\Delta] \dmoveto{s} \Delta'$ then there exists a $\Theta'$ such that $B''[\Theta] \dmoveto{s} \Theta'$ and $\Delta' \rel \Theta'$.
	However, by definition of $\ccl$, $\Delta' \rel \Theta'$ implies that $\Delta' \mathrel{\ccl(\rel)} \Theta'$.
\end{proof}

\begin{theorem}\label{thm:lcb-valid}
	Bisimulation up-to $\lcl \circ \ccl$ is a valid proof technique for $\simds$, meaning that if $\Delta \rel \Theta$ for a saturated bisimulation up-to $\lcl \circ \ccl$, then $\Delta \simds \Theta$.
\end{theorem}
\begin{proof}
	From \cite{sangiorgienhancements2011}, we know that if two functions $f_1$ and $f_2$ are $b$-compatible, then also $f_1 \circ f_2$ is $b$-compatible.
	Furthermore, bisimulations up-to $f$ are a sound proof technique whenever $f$ is $b$-compatible.
\end{proof}

\linearityCongruence*
\begin{proof}
	By requirement of saturated bisimulations, the bisimilarity $\simds$ must associate only distributions with equal mass, therefore by \autoref{thm:lcl-linear} $\lcl(\simds)$ is linear.
	However, as a corollary of \autoref{thm:lcb-valid} $\mathord{\lcl(\simds)} = \mathord{\simds}$ therefore $\simds$ is linear.
\end{proof}

\propertyA*
\begin{proof}
	We will prove that $\delem{1}{\conf{p\cdot\rho + (1 - p)\cdot\sigma, P}} \simds \delem{p}{\conf{\rho, P}} \oplus \delem{1 - p}{\conf{\sigma, P}}$ for any operators $\rho$, $\sigma$, process $P$ and $p \in [0,1]$.
	By linearity of $\simds$ we have that
	\[
		\osum_{i} \delem{p_i}{\conf{\rho_i,P}} \simds \delem{1}{\conf{\sum_{i}p_i\rho_i,P}} \simds \delem{1}{\conf{\sum_{j}q_j\sigma_j,P}} \simds \osum_{i} \delem{q_j}{\conf{\sigma_j,P}}
	\]
	We define the relation $\rel$ and prove that is a bisimulation up-to $\lcl$.
	\[
		\mathord{\rel} = \left\{(\delem{1}{\conf{p\cdot\rho + (1 - p)\cdot\sigma, P}}, \delem{p}{\conf{\rho, P}} \oplus \delem{1 - p}{\conf{\sigma, P}})\ \middle| \rho, \sigma, p, P \right\}
	\]

	Take $(\Delta, \Theta) \in \mathord{\rel}$ and a compatible context $B[\blank]$.
	Let $\Delta = \delem{1}{\conf{\nu, P}}$, $\Theta = \delem{p}{\conf{\rho, P}} \oplus \delem{1-p}{\conf{\sigma, P}}$, with $\nu = p\rho + (1-p)\sigma$.
	Note that $\delem{1}{\conf{\nu, B[P]}} = \dmoveto{s} \Xi$ iff $\conf{\nu, B[P]} \mathrel{\operatorname{bot}(\moveto{s})} \Xi$ (and similarly for $\rho$ and $\sigma$).
	Thus, by~\autoref{thm:alwaysmove}, there exists a family $next(B[P], s)$ of trace non-increasing superoperators $\E_i$ and processes $P_i$ such that
	$B[\Delta] \dmoveto{s} \Delta'$ if and only if
	\begin{gather*}
		\Delta' = \osum_{{i \in I \text{ s.t. } \tr(\E_i(\nu)) \neq 0}} \delem{\tr(\E_i(\nu))}{\conf{\frac{\E_i(\nu)}{\tr(\E_i(\nu))}, P_i}}
		= \osum_{{i \in I \text{ s.t. } \tr(\E_i(\nu)) \neq 0}} \delem{\tr(\E_i(\rho)\psum{p}\E_i(\sigma))}{\conf{\frac{\E_i(\rho)\psum{p}\E_i(\sigma)}{\tr(\E_i(\rho)\psum{p}\E_i(\sigma))}, P_i}}
	\end{gather*}
	Where $\_ \posum{p} \_ = p\_ \oplus (1-p) \_$ and $\_ \psum{p} \_ = p\_ + (1-p)\_$.

	Moreover, by~\autoref{thm:alwaysmove}, left decomposability and left linearity of $\dmoveto{s}$, $B[\Theta] \dmoveto{s} \Theta'$ if and only if
	\begin{gather*}
		\Theta' = {\osum_{{i \in I \text{ s.t. } \tr(\E_i(\rho)) \neq 0}} \delem{\tr(\E_i(\rho))}{\conf{\frac{\E_i(\rho)}{\tr(\E_i(\rho))}, P_i}}}
		\posum{p} {\osum_{{i \in I \text{ s.t. } \tr(\E_i(\sigma)) \neq 0}} \delem{\tr(\E_i(\sigma))}{\conf{\frac{\E_i(\sigma)}{\tr(\E_i(\sigma))}, P_i}}}
	\end{gather*}

	Recall that both superoperators and the trace are linear operators.
	Also by linearity, $\tr(\E_i(\rho \psum{p} \sigma)) \neq 0$ iff $\tr(\E_i(\rho)) \psum{p} \tr(\E_i(\sigma)) \neq 0$, and $\Theta'$ can be rewritten in the following form
	\begin{gather*}
		\Theta' = \osum_{{i \in I \text{ s.t. } \tr(\E_i(\nu)) \neq 0}} \delem{\tr(\E_i(\rho))\psum{p}\tr(\E_i(\sigma))}{\left(\conf{\frac{\E_i(\rho)}{\tr(\E_i(\rho))}, P_i} \posum{\frac{p\tr(\E_i(\rho))}{\tr(\E_i(\rho)) \psum{p} \tr(\E_i(\sigma))}} \conf{\frac{\E_i(\sigma)}{\tr(\E_i(\sigma))}, P_i}\right)}
	\end{gather*}
	By performing some simple calculations, one can prove that $\Delta'$ and $\Theta'$ have the same weights, and thus $\Delta' \mathrel{\lcl(\rel)} \Theta'$.
\end{proof}

%

%
\section{Proofs about Probabilistic Labelled Bisimulations}

As an intermediate step for proving the full abstraction of~\autoref{thm:demonicfa}, we introduce a probabilistic version of the labelled bisimilarity, and prove its complete adherence with respect to the saturated equivalence.

We need some auxiliary definitions.
First, given a distribution $\Delta$ of lqCCS processes, we let $env(\Delta)$ be its visible quantum state $\sum_{\conf{\rho, P} \in \supp{\Delta}} \Delta(\conf{\rho, P}) \cdot tr_{\Sigma'}(\rho)$, where $\Sigma,  \Sigma' \vdash \Delta$.
Notice that $env(\Delta)$ is a partial density operator in $\sdist{\hilb_{\Sigma \setminus \Sigma'}}$, with $\tr(env(\Delta)) = \mass{\Delta}$.

%

We can now define our labelled bisimilarity.

\begin{definition}[Labelled Bisimilarity]\label{def:labBisim}
	A relation $\mathord{\rel} \subseteq \sdist{\confset} \times \sdist{\confset}$ is a
	\emph{labelled bisimulation} if $\Delta\,\rel\,\Theta$ implies
	$(\Sigma, \Sigma') \vdash \Delta$ and $(\Sigma, \Sigma') \vdash \Theta$ for some $\Sigma, \Sigma'$, and it holds that
	\begin{itemize}
		\item $env(\Delta) = env(\Theta)$;
		\item $\nE(\Delta)\ \rel\ \nE(\Theta)$ for any superoperator $\mathcal{E} \in \soset{\hilb_{\Sigma \setminus \Sigma'}}$;
		\item whenever $\Delta \dmoveto[\mu]{s} \Delta'$, there exists $\Theta'$
		      such that $\Theta \dmoveto[\mu]{s} \Theta'$ and $\Delta'\;\rel\;\Theta'$;
		\item whenever $\Theta \dmoveto[\mu]{s} \Theta'$, there exists $\Delta'$
		      such that $\Delta \dmoveto[\mu]{s} \Delta'$ and $\Delta'\;\rel\;\Theta'$.
	\end{itemize}

	Let \emph{labelled bisimilarity}, denoted $\simpl$, be the largest labelled bisimulation.
\end{definition}

We will now introduce some lemmas, useful to prove that $\simds$ and $\simpl$ coincide.

\begin{lemma}\label{thm:alphamove}
	For any $\Delta$ and $t$, and for any tag $t'$ fresh in $\Delta$,
	\[
		\Delta \dmoveto{s} \Delta' \quad \text{ if and only if }\quad \Delta[\sfrac{t'}{t}] \dmoveto{s[\sfrac{t'}{t}]} \Delta'[\sfrac{t'}{t}]
	\]
\end{lemma}
\begin{proof}
	By induction on the operational semantics.
\end{proof}

\begin{lemma}\label{thm:alpha}
	For any $\Delta$, $\Theta$ and $t$, and for any tag $t'$ fresh in $\Delta$ and $\Theta$,
	\[
		\Delta \simds \Theta \quad\text{ if and only if }\quad \Delta[\sfrac{t'}{t}] \simds \Theta[\sfrac{t'}{t}]
	\]
\end{lemma}
\begin{proof}
	Let $\rel$ be the smallest relation such that $\mathord{\simds} \subseteq \mathord{\rel}$ and $\mathord{\rel} \subseteq \{ (\Delta[\sfrac{t'}{t}], \Theta[\sfrac{t'}{t}]) \mid  \Delta \rel \Theta \text{ and } t' \text{ is fresh in }\Delta, \Theta\}$.
	We will show that $\rel$ is a bisimulation.

	In the following, we write $[\sfrac{t_i'}{t_i}]_{i \in I}$ for the substitution $[\sfrac{t_{x_0}'}{t_{x_0}}][\sfrac{t_{x_1}'}{t_{x_1}}] \dots  [\sfrac{t_{x_n}'}{t_{x_n}}]$, with $I = \{x_0, x_1, \dots, x_n\}$.

	Take $\Delta[\sfrac{t_i'}{t_i}]_{i \in I}$ and $\Theta[\sfrac{t_i'}{t_i}]_{i \in I}$ in $\rel$, with $\Delta \simds \Theta$, and take a generic context $B[\blank] = [\blank] \parallel R$.
	Assume that $$\Delta[\sfrac{t_i'}{t_i}]_{i \in I} \parallel R \dmoveto{s} \Delta'.$$

	Take then a collection of distinct fresh tags $t_i''$.
	By~\autoref{thm:alphamove},
	\begin{align*}
		(\Delta[\sfrac{t_i'}{t_i}]_{i \in I} \parallel R)
		[\sfrac{t_i''}{t_i},\sfrac{t_i}{t_i'}]_{i \in I}
		 & =
		\Delta[\sfrac{t_i'}{t_i}, \sfrac{t_i''}{t_i},\sfrac{t_i}{t_i'}]_{i \in I} \parallel R[\sfrac{t_i''}{t_i},\sfrac{t_i}{t_i'}]_{i \in I}
		=                                                                                                                                                                                            \\
		 & = \Delta \parallel R[\sfrac{t_i''}{t_i},\sfrac{t_i}{t_i'}]_{i \in I} \dmoveto{s[\sfrac{t_i''}{t_i},\sfrac{t_i}{t_i'}]_{i \in I}} \Delta'[\sfrac{t_i''}{t_i},\sfrac{t_i}{t_i'}]_{i \in I}.
	\end{align*}
	Since $\Delta \simds \Theta$, then it must be that $\Theta \parallel R[\sfrac{t_i''}{t_i},\sfrac{t_i}{t_i'}]_{i \in I} \dmoveto{s[\sfrac{t_i''}{t_i},\sfrac{t_i}{t_i'}]_{i \in I}} \Theta'$, with $\Delta'[\sfrac{t_i''}{t_i},\sfrac{t_i}{t_i'}]_{i \in I} \simds \Theta'$.
	By~\autoref{thm:alphamove}, $(\Theta \parallel R[\sfrac{t_i''}{t_i},\sfrac{t_i}{t_i'}]_{i \in I})[\sfrac{t_i'}{t_i},\sfrac{t_i}{t_i''}]_{i \in I} = \Theta  [\sfrac{t_i'}{t_i}]_{i \in I} \parallel R \dmoveto{s} \Theta'[\sfrac{t_i'}{t_i},\sfrac{t_i}{t_i''}]_{i \in I}$.
	We conclude by noting that $\Delta'\ \rel\ \Theta'[\sfrac{t_i'}{t_i},\sfrac{t_i}{t_i''}]_{i \in I}$ because $\Delta' = \Delta'[\sfrac{t_i''}{t_i},\sfrac{t_i}{t_i'}]_{i \in I}[\sfrac{t_i'}{t_i},\sfrac{t_i}{t_i''}]_{i \in I}$.
\end{proof}

\begin{lemma}\label{thm:freshtag}
	For any $\Delta$, $\Delta'$, and $B[\blank]$, and for any $s$ with only tags that are fresh in $B[\blank]$,
	\[
		B[\Delta] \dmoveto{s} \Delta' \quad \text{ if and only if }\quad \Delta \dmoveto{s} \Delta'' \text{ and } \Delta' = B[\Delta'']
	\]
\end{lemma}
\begin{proof}
	By the rules for the parallel composition.
\end{proof}

\begin{lemma}\label{thm:ignoreSending}
	For any $\Sigma, \Sigma' \vdash \Delta, \Theta$, for any $q$ in $\Sigma \setminus \Sigma'$, we have that
	$\Delta \simds \Theta$ if and only if
	$\Delta \parallel t\tags c!q \simds \Theta \parallel t\tags c!q$ for some $t$ fresh both in $\Delta$ and $\Theta$.
\end{lemma}
\begin{proof}
	$\Delta \simds \Theta$ implies $\Delta \parallel t\tags c!q \simds \Theta \parallel t\tags c!q$ because saturated bisimilarity is a congruence by definition.
	For the other side, it suffices showing that
	$\mathord{\rel} = \{(\Delta, \Theta)  \mid \Delta \parallel t\tags c!q \simds \Theta \parallel t\tags c!q\}$ is a bisimulation.

	Take $(\Delta,\Theta) \in \rel$ and $\Sigma'' \vdash B[\blank]_{\Sigma'} = [\blank] \parallel R$.
	We consider two cases depending on $q$ being in $\Sigma''$ or not.

	If $q \in \Sigma''$, then let $\Sigma'' \vdash B'[\blank]_{\Sigma' \cup \{q\}} = [\blank] \parallel t'\tags c?x.R[\sfrac{x}{q}][\sfrac{t''}{t}]$, with $t'$ and $t''$ fresh tags in $\Delta, \Theta$ and $R$.
	We have then that $B'[\Delta \parallel t\tags c!q] \dmoveto{{{(t,t')}}} \Delta \parallel R[\sfrac{t''}{t}]$, and
	$B'[\Theta \parallel t\tags c!q] \dmoveto{{{(t,t')}}} \Theta \parallel R[\sfrac{t''}{t}]$.
	Since $\Delta \parallel t\tags c!q \simds \Theta \parallel t\tags c!q$ by assumption, we know that $\Delta \parallel R[\sfrac{t''}{t}] \simds \Theta \parallel R[\sfrac{t''}{t}]$, and the property follows by~\autoref{thm:alpha} since $t''$ is chosen free in $\Delta, \Theta$ and $R$.

	If $q \notin \Sigma''$, then let $\Sigma'' \vdash B'[\blank]_{\Sigma'} = [\blank] \parallel R[\sfrac{t'}{t}]$, with $t' \neq t$ fresh in $\Delta, \Theta$ and $R$.
	Then, assume $B[\Delta] \dmoveto{s} \Delta'$.
	Notice that $B[\Delta][\sfrac{t'}{t}] = B'[\Delta]$ since $t$ is fresh in $\Delta$.
	Hence, $B'[\Delta] \dmoveto{s'} \Delta'[\sfrac{t'}{t}]$, with $s' = s[\sfrac{t'}{t}]$, by~\autoref{thm:alphamove}.
	Clearly, $t$ does not appear in $s$, hence, by~\autoref{thm:freshtag}, $B'[\Delta \parallel t\tags c!q] \dmoveto{s'} \Delta'[\sfrac{t'}{t}] \parallel t\tags c!q$ (where we exploit the associativity of the parallel operator).
	Since $\Delta \parallel t\tags c!q \simds \Theta \parallel t\tags c!q$, we know that
	$B'[\Theta \parallel t\tags c!q] \dmoveto{s'} \Theta'$, with $\Delta'[\sfrac{t'}{t}] \parallel t\tags c!q \simds \Theta'$.
	Moreover, by~\autoref{thm:freshtag}, $\Theta'$ is of the form $\Theta'' \parallel t\tags c!q$, and $B'[\Theta] \dmoveto{s'} \Theta''$.
	By~\autoref{thm:alphamove}, we then know that $\Theta'' = \Theta'''[\sfrac{t'}{t}]$ for some $\Theta$ such that $B[\Theta] \dmoveto{s} \Theta'''$.
	We know that $\Delta'[\sfrac{t'}{t}] \parallel t\tags c!q \simds \Theta'''[\sfrac{t'}{t}] \parallel t\tags c!q$, hence, by~\autoref{thm:alphamove},
	\begin{align*}
		(\Delta'[\sfrac{t'}{t}] \parallel t\tags c!q)[\sfrac{t}{t'}, \sfrac{t'}{t}] = \Delta' \parallel t'\tags c!q \simds
		\Theta''' \parallel t\tags c!q =
		(\Theta'''[\sfrac{t'}{t}] \parallel t\tags c!q)[\sfrac{t}{t'}, \sfrac{t'}{t}].
	\end{align*}
	Therefore, $(\Delta', \Theta''') \in \rel$.
\end{proof}

In the following we will write $\tilde{t}\tags c!\tilde{q}$ for the process
\[
	t_0\tags c!q_0 \parallel t_1\tags c!q_1 \parallel \dots \parallel t_n\tags  c!q_n,
\]
with $c$ any channel, and $\tilde{t} = t_0, t_1, \dots t_n$ and $\tilde{q} = q_0, q_1, \dots q_n$ any sequences of (distinct) tags and qubit names.

\begin{lemma}\label{lem:cinesi}
	Let $\Delta \simds \Theta$. Then
	\begin{itemize}
		\item $env(\Delta) = env(\Theta)$;
		\item $\nE(\Delta) \simds \nE(\Theta)$ for any trace non-increasing superoperator over the environment.
	\end{itemize}
\end{lemma}
\begin{proof}
	For the first point, suppose that $\env{\Delta} \neq \env{\Theta}$, we will show that $\Delta \not\simds \Theta$.
	Since $\env{\Delta}  \neq \env{\Theta}$, then there exists a positive linear operator $E$ such that $\tr(E \cdot \env{\Delta}) \neq \tr(E \cdot \env{\Theta})$~\cite{heinosaariMathematicalLanguageQuantum2011}.
	We can construct the measurement $\M_{E}$ with measurement operators $\{\sqrt{E}, \sqrt{\I - E}\}$ and the context
	$B[\blank] = [\blank] \parallel t_0\tags\meas[\M_E]{\tilde{q}}{x}.R$, where $R = \ite{x = 0}{t_1\tags\tau.\nil_{\tilde{q}}}{\nil_{\tilde{q}}}$, $t_0, t_1$ are fresh tags and $\tilde{q}$ is the tuple of all the qubits outside of $\Delta$.
	For $\Delta = \osum_i \delem{p_i}{\conf{\rho_i, P_i}}$ and $\Theta = \osum_j \delem{p_j}{\conf{\rho_j, P_i}}$ we have
	\begin{align*}
		B[\Delta]
		 & \dmoveto{{t_0}} \osum_i \delem{p_i}{\left(\delem{q_i}{\conf{\rho_i^E, P_i \parallel R[0/x]}} \oplus \delem{(1-q_i)}{\conf{\rho_i^{\not E}, P_i \parallel R[1/x]}}\right)} \\
		 & \dmoveto{{t_1}} \osum_i \delem{p_iq_i}{\conf{\rho_i^E, P_i \parallel \nil_{\tilde{q}}}} = \Delta'
	\end{align*}
	and
	\begin{align*}
		B[\Theta]
		 & \dmoveto{{t_0}} \osum_j \delem{p_j}{\left(\delem{q_j}{\conf{\rho_j^E, P_j \parallel R[0/x]}} \oplus \delem{(1-q_j)}{\conf{\rho_j^{\not E}, P_j \parallel R[0/x]}}\right)} \\
		 & \dmoveto{{t_1}} \osum_j \delem{p_jq_j}{\conf{\rho_j^E, P_j \parallel \nil_{\tilde{q}}}} = \Theta'
	\end{align*}
	where $\rho_n^E$ (respectively $\rho_n^{\not E}$) is the state resulting from the successful (failed) measurement of $E$ on $\rho_n$, $q_n$ is the probability $\tr(E\otimes \mathbb{I} \cdot \rho_n)$.
	We have that $\tr(E \cdot \env{\Delta}) = \tr(E\otimes I \cdot \sum_ip_i\rho_i) = \mass{\Delta'}$, and thus $\mass{\Delta'} \neq \mass{\Theta'}$ and $\Delta, \Theta$ are not bisimilar.

	For the second point, we consider separately the case of a trace-preserving superoperator $\mathcal{E}$.
	In this case, we take the distinct fresh tags $t$ and $\tilde{t}$, and we build a context $B[\blank] = [\blank] \parallel t\tags\E(\tilde{q}). \tilde{t}\tags c!\tilde{q}$.
	After a $\xlongrightarrow{\tau}_{t}$ transition, $\Delta$ goes in $\nE(\Delta) \parallel \tilde{t}\tags c!\tilde{q}$ and $\Theta$ goes in $\nE(\Theta) \parallel \tilde{t}\tags c!\tilde{q}$, from which we know $\nE(\Delta) \simds \nE(\Theta)$ thank to \autoref{thm:ignoreSending}.

	Finally,
	Suppose $\E = \{K_0, \ldots, K_{n-1}\}$ is a trace non-increasing superoperator with $n$ Kraus operators, we define $n$ superoperators $\mathcal{M}_i = \{K_i\}$ for $i = 0,\ldots, n-1$.
	We have that, for any $\rho$, $\E(\rho) = \sum_{i = 0}^{n-1} \mathcal{M}_i(\rho)$ and $\tr(\E(\rho)) = \sum_{i = 0}^{n-1} \tr(\mathcal{M}_i(\rho))$.
	Since $\E$ is trace non-increasing, we know by definition that $\I - \sum_{i = 0}^{n-1} K_i^\dagger K_i$ is a positive matrix, and we call this difference $M$.
	We build a measurement $\M$ with measurement operators $\{K_0, \ldots, K_{n-1}, \sqrt{M}\}$.
	This is an $n+1$-outcome measurement, one for each Kraus operator of $\E$ and an additional one signifying that $\E$ does not happen.
	Take then the context $B[\blank] = [\blank] \parallel t\tags\meas{\tilde{q}}{x}.R$ where $R=\ite{x \neq n}{t'\tags\tau.\tilde{t}\tags c!\tilde{q}}{\tilde{t}\tags c!\tilde{q}}$ and $t, t', \tilde{t}$ are distinct fresh tags.
	Suppose that $\Delta = \osum_j \delem{p_j}{\conf{\rho_j, P_j}}$, then a possible sequence of transition of $B[\Delta]$ is
	\begin{align*}
		 & B[\Delta] \dmoveto{t} \osum_j \delem{p_j}{\left(
			\osum_{i = 0}^{n} \delem{\tr(\M_i(\rho_j))}{\conf{\frac{\M_i(\rho_j)}{\tr(\M_i(\rho_j))}, P_j \parallel R[i/x]}}
		\right)}                                            \\
		 & \dmoveto{t'}\osum_j \delem{p_j }{\left[
				\tr(\E(\rho_j))\left(
				\osum_{i=0}^{n-1} \delem{\frac{\tr(\M_i(\rho_j))}{\tr(\E(\rho_j))}}{ \conf{\frac{\M_i(\rho_j)}{\tr(\M_i(\rho_j))}, P_j \parallel \tilde{t}\tags c!\tilde{q} }}
				\right)
				\right]} = \Delta'
	\end{align*}
	Applying \autoref{thm:propertyA} to such $\Delta'$, we get by linearity
	\begin{align*}
		\Delta' & \sim \osum_j \delem{p_j}{\left[
				\delem{\tr(\E(\rho_j))}{\conf{\frac{\sum_{i=0}^{n-1} \M_i(\rho_j)}{\tr(\E(\rho_j))}, P_j \parallel \tilde{t}\tags c!\tilde{q}}
				}
		\right]}                                                                                                                             \\
		        & = \osum_j \delem{p_j \tr(\E(\rho_j))}{\conf{\frac{\E(\rho_j)}{\tr(\E(\rho_j))}, P_j \parallel \tilde{t}\tags c!\tilde{q}}} \\
		        & = \nE(\Delta) \parallel \tilde{t}\tags c!\tilde{q}.
	\end{align*}
	This sequence of transition must be replicated also by $B[\Theta]$, and so we know that
	\begin{align*}
		B[\Delta] \dmoveto{t}\dmoveto{t'} \  & \Delta' \simds \nE(\Delta) \parallel \tilde{t}\tags c!\tilde{q} \\
		                                     & \ \rotatebox{90}{$\simds$}                                      \\
		B[\Theta] \dmoveto{t}\dmoveto{t'} \  & \Theta' \simds \nE(\Theta) \parallel \tilde{t}\tags c!\tilde{q}
	\end{align*}
	We thus know that $\nE(\Delta) \simds \nE(\Theta)$ thanks to \autoref{thm:ignoreSending}.
	Notice how all the contexts we used in this proof use fresh tags, and thus can be legally applied to $\Delta$ and $\Theta$ due to \autoref{lem:deterPar}.
\end{proof}

\begin{restatable}{theorem}{complete}\label{thm:complete}
	$\mathord{\simds} \subseteq \mathord{\simpl}$.
\end{restatable}
\begin{proof}
	We have to prove that $\simds$ is a labelled bisimulation. The first two conditions are already proven in \autoref{lem:cinesi}, we must check the last two.
	Suppose $\Delta \simds \Theta$ and $\Delta \dmoveto[\mu]{s} \Delta'$. We proceed by cases on $\mu$, and build contexts capable of bringing $\Delta$ to $\Delta'$. All the contexts we will build use  fresh tags, and thus can be legally applied to $\Delta$ and $\Theta$ due to \autoref{lem:deterPar}.

	If $\mu = \tau$, then $B[\Delta] \dmoveto{s} \Delta'$ for the empty context $B[\blank] = \blank$, thus there exists a $\Theta'$ such that $B[\Theta] = \Theta \dmoveto{s} \Theta'$ and $\Delta' \simds \Theta'$.


	If $\mu = c?v$, then $\Delta \parallel t \tags c!v \dmoveto{(s, t)} \Delta' \parallel \nil$. By $\Delta \simds \Theta$, there exists $\Theta'$ such that $\Theta \parallel t \tags c!v \dmoveto{(s,t)} \Theta' \parallel \nil$.
	From the semantics of synchronization we also know that $\Theta \dmoveto[c?v]{s} \Theta'$ and $\Delta' \parallel \nil \simds \Theta' \parallel \nil$.
	It is possible to show that $\Delta \parallel \nil \simds \Delta$ for any $\Delta$, and thus we get $\Delta' \simds \Theta'$.


	If $\mu = c!v$, when $v$ is a quantum name, we build the context $B[\blank] = [\blank] \parallel t \tags c?x.R$ with $R = \ite{x = v}{t' \tags \tau.t'' \tags c!x}{t'' \tags c!x}$.
	Then we have that $B[\Delta] \dmoveto{(s, t)} \Delta' \parallel R[\sfrac{v}{x}]$, and thus $B[\Theta] \dmoveto{(s, t)} \Theta' \parallel R[\sfrac{u}{x}]$ and $\Theta \dmoveto[c!u]{s} \Theta'$ for some quantum name $u$ possibly different from $v$.
	Thanks to the bisimilarity between $\Delta' \parallel R[\sfrac{v}{x}]$ and $\Theta' \parallel R[\sfrac{u}{x}]$, we can prove $u = v$:
	$\Delta' \parallel R[\sfrac{v}{x}] \dmoveto{t'} \Delta' \parallel t'' \tags c!v$, and also $\Theta' \parallel R[\sfrac{u}{x}]$ must have the same action available, so it must be $u = v$.
	To sum up, we know that $\Theta \dmoveto[c?v]{s} \Theta'$ and that $\Delta' \parallel t'' \tags c!v \simds \Theta' \parallel t'' \tags c!v$.
	This implies $\Delta' \simds \Theta'$ thanks to \autoref{thm:ignoreSending}.

	The case for $\mu = c!v$ for a classical value $v$ is the same, using the context $B[\blank] = \blank \parallel t \tags c?x.\ite{x = v}{t'\tags \tau. \nil}{\nil}$.
\end{proof}

\begin{restatable}{theorem}{correct}\label{thm:correct}
	$\mathord{\simpl} \subseteq \mathord{\simds}$.
\end{restatable}
\begin{proof}
	It is sufficient to show that $\simpl$ is a saturated bisimulation up-to $\lcl  \circ \ccl$.
	Assume $\Delta \simpl \Theta$.
	To see that $\mass{\Delta} = \mass{\Theta}$ it is sufficient to notice that $\env{\Delta} = \env{\Theta}$, by definition of $\simpl$, and that $\mass{\Delta} = tr(\env{\Delta})$ for any $\Delta$.

	We now show that for all $B[\blank], \Delta, \Theta, \Delta'$, if $\Delta \simpl \Theta$ and $B[\Delta] \dmoveto{s} \Delta'$ then
	$\Theta'$ exists such that $B[\Theta] \dmoveto{s} \Theta'$ and $\Delta' \mathrel{\lcl(\ccl(\simpl))} \Theta'$.
	If $B[\blank] = \blank$, then from $\Delta \simpl \Theta$ we have $\Delta' \simpl \Theta'$ for some $\Theta'$ such that $\Delta' \simpl \Theta'$, and $\simpl \subseteq \lcl(\ccl(\simpl))$.

	When $\Delta = \Theta = \epsilon$, we have that $B[\Delta] = B[\Theta]= \epsilon  \dmoveto[\mu]{s} \epsilon$ for any $B$, $\mu$ and $s$, and it is possible to prove that $\epsilon \simds \epsilon$.

	Otherwise, we know that $\Delta \neq \epsilon \neq \Theta$, and that $B[\blank] = \blank \parallel R$ for some $R$.
	If $\Delta\parallel R \dmoveto{s} \epsilon$, then we know from \autoref{lem:distr_epsilon_cases} that:
	\begin{enumerate}
		\item $\Delta \dmoveto[]{s} \epsilon$, which implies that $\Theta \dmoveto[]{s} \epsilon$, since $\Delta \simpl \Theta$;
		\item $next(R, s, \tau) = \emptyset$;
		\item For all $c$ and $v$, if $s = (t_1, t_2)$ we have
		      $$next(R, t_2, c!v) \neq \emptyset \implies \Delta \dmoveto[c?v]{t_1} \epsilon \implies  \Theta \dmoveto[c?v]{t_1} \epsilon$$
		      and
		      $$next(R, t_2, c?v) \neq \emptyset \implies \Delta \dmoveto[c!v]{t_1} \epsilon \implies \Theta \dmoveto[c!v]{t_1} \epsilon.$$
	\end{enumerate}
	In other words, we know all the three points necessary to apply \autoref{lem:distr_epsilon_cases} on $\Theta$ in the reverse direction, obtaining that $\Theta \parallel R \dmoveto[\tau]{s} \epsilon$.

	Suppose instead that $\Delta\parallel R \dmoveto{s} \Delta' \neq \epsilon$, then we know from \autoref{lem:distr_cases} that one of the following is true:
	\begin{enumerate}
		\item $\tagset(s) \subseteq \tagset(\Delta)$ and there exists $Delta''$ such that $\Delta' = \Delta'' \parallel R$ and $\Delta \dmoveto[\tau]{s} \Delta''$.
		      The latter implies $\Theta \dmoveto[\tau]{s} \Theta'$ for some $\Theta' \simpl \Delta''$, since $\Delta \simpl \Theta$.
		      But then, it must be $\tagset(s) \subseteq \tagset(\Theta)$, and from \autoref{lem:distr_cases} we know $\Theta\parallel R \dmoveto[\tau]{s} \Theta' \parallel R$, and $\Delta'' \parallel R \ \ccl(\simpl)\  \Theta' \parallel R$.

		\item $\tagset(s) \subseteq \tagset(R)$ and $\Delta' = \osum_{i \in I}\E_i(\Delta \parallel R_i)$ with $\nextset(R, s, \mu) = \{\E_i, R_i\}_{i \in I}$.
		      But then we know from \autoref{lem:distr_cases} that $\Theta \parallel R \dmoveto[\tau]{s}  \osum_{i \in I}\E_i(\Theta \parallel R_i)$.
		      Thanks to superoperator closure, we know for all the $i$s that $\E_i(\Delta) \simpl \E_i(\Theta)$, and thus $\E_i(\Delta \parallel R_i \ \ccl(\simpl)\  \Theta \parallel R_i$ (notice how $\E_i(\Delta) \parallel R = \E_i(\Delta \parallel R)$).
		      This proves that $\osum_{i \in I}\E_i(\Delta \parallel R_i) \mathrel{\lcl(\ccl(\simpl))} \osum_{i \in I}\E_i(\Theta \parallel R_i)$.

		\item $s = (t_1, t_2)$, $t_1 \in \tagset(\Delta)$, $t_2 \in \tagset(R)$, and $\Delta' = \Delta'' \parallel R'$, $\exists c,v\ldotp \Delta \dmoveto[c?v]{t_1} \Delta''$ with $\nextset(R, t_2, c!v) = \{\I, R'\}$.
		      This implies $\Theta \dmoveto[c?v]{t_1} \Theta'$ for some $\Theta' \simpl \Delta''$, since $\Delta \simpl \Theta$, and thus $\Theta \parallel R \dmoveto[]{s} \Theta' \parallel R'$ thanks to \autoref{lem:distr_cases}, where we have $\Delta'' \parallel R \mathrel{\ccl(\simpl)} \Theta' \parallel R$.

		\item $s = (t_1, t_2)$, $t_1 \in \tagset(\Delta)$, $t_2 \in \tagset(R)$, and $\exists c,V \ldotp \Delta' = \osum_{v \in V}\Delta_v \parallel R_v$, and $ \forall v \in V\ldotp \Delta \dmoveto[c!v]{t_1} \Delta_v$ with $\nextset(R, t_2, c?v) = \{\I, R_v\}$.
		      This implies $\Theta \dmoveto[c!v]{t_1} \Theta_v$ for some $\Theta_v \simpl \Delta_v$ for any $v \in V$, and thus $\Theta \parallel R \dmoveto[]{s} \osum_{v \in V} \Theta_v \parallel R_v = \Theta'$ thanks to \autoref{lem:distr_cases}, where we have $\Delta' \mathrel{\ccl(\simpl)} \Theta'$.
	\end{enumerate}

\end{proof}

\begin{theorem}\label{thm:probcc}
	For any $\Delta$ and $\Theta$, $\Delta \simpl \Theta$ if and only if $\Delta \simds \Theta$.
\end{theorem}
\begin{proof}
	By~\autoref{thm:complete} and~\ref{thm:correct}.
\end{proof}


%
\section{Proofs about Quantum Labelled Bisimulations}
\begin{definition}\label{def:alphaGamma}
	We define the following:
	\begin{itemize}
		\item The function $\alpha : \sdist{\confset} \to \qdist{\procset}$ is given by
		      \begin{align*}
			      \alpha\left(\osum_{i \in I}\osum_{j \in J_i} \delem{p_j}{\conf{\rho_j, P_i}}\right)
			      = \osum_{i \in I} \delem{\left(\sum_{j \in J_i} p_j\rho_j\right)}{P_i}
		      \end{align*}
		      where processes are indexed by $i \in I$, and the states and probabilities associated to a process $P_i$ are indexed by $j \in J_i$.
		\item The function $\gamma : \qdist{\procset} \to \sdist{\confset}$ is given by
		      \begin{align*}
			      \gamma\left(\osum_{i \in I} \delem{\rho_i}{P_i}\right) = \osum_{i \in I} \delem{\tr(\rho_i)}{\conf{\frac{\rho_i}{\tr(\rho_i)}, P_i}}
		      \end{align*}
		      where processes are indexed by $i \in I$.
	\end{itemize}
\end{definition}

\begin{lemma}\label{lm:alphaGammaInverse}
	$\alpha$ is the left inverse of $\gamma$, and $\gamma$ is the left inverse of $\alpha$ up-to bisimilarity.
	\[
		\qd = \alpha(\gamma(\qd)) \qquad
		\Delta \simds \gamma(\alpha(\Delta))
	\]
\end{lemma}
\begin{proof}
	The first point follows from definition, the second from \autoref{thm:propertyA} and \autoref{thm:probcc}.
\end{proof}

\begin{lemma}\label{lm:alphaGammaSem}
	The function $\alpha$ and $\gamma$ respect transitions, i.e.:
	\begin{itemize}
		\item For any $\Delta, \Delta' \in \sdist{\confset}$, if $\Delta \dmoveto[\mu]{s} \Delta'$ then $\alpha(\Delta) \qmoveto[\mu]{s} \alpha(\Delta')$.
		\item For any $\qd, \qd' \in \qdist{\procset}$, if $\qd \qmoveto[\mu]{s} \qd'$ then $\gamma(\qd) \dmoveto[\mu]{s} \Delta'$ for some $\Delta' \in \sdist{\confset}$ such that $\qd' = \alpha(\Delta')$.
	\end{itemize}
\end{lemma}
\begin{proof}

	By definition of $\alpha$ and since the rules for $\cmoveto[\mu]{s}$ and $\qcmoveto[\mu]{s}$ match one-to-one, it is trivial to show that $\Delta \cmoveto[\mu]{s} \Delta'$ if and only if $\alpha(\Delta) \qcmoveto[\mu]{s} \alpha(\Delta')$.

	Regarding the first point, if $\Delta' = \epsilon$ then every configuration in $\Delta$ is in deadlock, and thus also the corresponding quantum distribution mapped by $\alpha$ are in deadlock.
	If $\Delta' \neq \epsilon$, then $\Delta'$ is obtained via linear combinations of $\Delta'_j$ with weights $p_j$ where $\conf{\rho_j, P_i} \cmoveto[\mu]{s} \Delta'_j$.
	But then $\alpha(\delem{p_j}{\conf{\rho_j, P_i}}) \qcmoveto[\mu]{s} \alpha(\delem{p_j}{\Delta'_j})$ and the result follows from applications of \textsc{Id} and \textsc{Conv}.

	Regarding the second point, let $\qd = \osum_{i \in I}\delem{\rho_i}{P_i}$.
	If $\delem{\rho_i}{P_i} \qmoveto[\mu]{s} \epsilon$ and thus $\delem{\rho_i}{P_i} \qcnmoveto[\mu]{s}$ then also $\conf{\frac{\rho_i}{\tr(\rho_i)}, P_i} \ncmoveto[\mu]{s}$ and $\alpha(\epsilon) = \epsilon$.
	If $\delem{\rho_i}{P_i} \qcmoveto[\mu]{s} \qd'_i$, by \autoref{lm:alphaGammaInverse} and the equivalence noted earlier, $\alpha(\delem{1}{\conf{\rho_i, P_i}}) \qcmoveto[\mu]{s} \alpha(\gamma(\qd'_i))$ if and only if $\conf{\rho_i, P_i} \cmoveto[\mu]{s} \gamma(\qd'_i)$.
	Take $\Delta' = \osum_{i \in I}\gamma(\qd'_i)$, thus trivially $\qd' = \alpha(\Delta')$.
\end{proof}

\begin{lemma}\label{lm:alphaEnv}
	The function $\alpha$ preserves environment, i.e. $\env{\Delta} = \env{\alpha(\Delta)}$
\end{lemma}
\begin{proof}
	\begin{align*}
		\env{\Delta} & = \env{\osum_{i \in I}\delem{p_i}{\conf{\rho_i, P_i}}} = \tr_{\Sigma}\mass{\osum_{i \in I}\delem{p_i}{\conf{\rho_i, P_i}}}                                                        \\
		             & = \tr_{\Sigma}\left(\sum_{i \in I}p_i \cdot \rho_i\right)                                                                                                                         \\
		             & = \tr_{\Sigma}\mass{\osum_{i \in I}\delem{(p_i \cdot \rho_i)}{P_i}} = \tr_{\Sigma}\mass{\alpha\left(\osum_{i \in I}\delem{p_i}{\conf{\rho_i, P_i}}\right)} = \env{\alpha(\Delta)}
	\end{align*}
\end{proof}

\begin{lemma}\label{lm:alphaSop}
	The function $\alpha$ commutes with superoperator application, i.e. $\alpha(\E(\Delta)) = \E(\alpha(\Delta))$
\end{lemma}
\begin{proof}
	\begin{align*}
		\alpha(\E(\Delta)) & = \alpha(\E(\osum_{i \in I}\delem{p_i}{\conf{\rho_i, P_i}}))                                                            \\
		                   & = \alpha\left(\osum_{i \in I}\delem{(p_i \cdot \tr(\E(\rho_i)))}{\conf{\frac{\E(\rho_i)}{\tr{\E(\rho_i)}}, P_i}}\right) \\
		                   & = \osum_{i \in I}\delem{(p_i \cdot \E(\rho_i))}{P_i}                                                                    \\
		                   & = \E\left(\osum_{i \in I}\delem{(p_i \cdot \rho_i)}{P_i}\right)                                                         \\
		                   & = \E(\alpha(\Delta))
	\end{align*}
\end{proof}

\begin{lemma}\label{lm:gammaSop}
	The function $\gamma$ commutes with superoperator application up to bisimilarity, i.e. $\gamma(\E(\qd)) \simds \E(\gamma(\qd))$
\end{lemma}
\begin{proof}
	Follows by \autoref{lm:alphaSop} and \autoref{lm:alphaGammaInverse}:
	\[\gamma(\E(\Delta)) = \gamma\circ\E\circ\alpha \circ \gamma(\Delta) = \gamma \circ \alpha \circ \E \circ \gamma(\Delta) \simds \E(\gamma(\Delta)).\qedhere\]
\end{proof}

We can now prove full-abstraction.
\begin{lemma}\label{lm:fullAbstrFwd}
	For any two probability distributions $\Delta, \Theta \in \sdist{\confset}$, if $\Delta \simpl \Theta$, then $\alpha(\Delta) \simdl \alpha(\Theta)$
\end{lemma}
\begin{proof}
	We define the relation $\mathord{\rel} \subseteq \qdist{\procset} \times \qdist{\procset}$ as
	\[
		\mathord{\rel} = \left\{(\alpha(\Delta), \alpha(\Theta)) \mid \Delta \simpl \Theta\right\}
	\]
	and prove it to be a labelled quantum bisimilarity.
	It is environment preserving thanks to \autoref{lm:alphaEnv}, as
	\[ \env{\alpha(\Delta)} = \env{\Delta} = \env{\Theta} = \env{\alpha(\Theta)}.\]
	It is superoperator closed thanks to \autoref{lm:alphaSop}, as
	\begin{align*}
		\alpha(\Delta) \rel \alpha(\Theta)
		 & \Rightarrow
		\Delta \simpl \Theta \Rightarrow
		\E(\Delta) \simpl \E(\Theta) \\
		 & \Rightarrow
		\alpha(\E(\Delta)) \rel \alpha(\E(\Theta)) \Rightarrow
		\E(\alpha(\Delta)) \rel \E(\alpha(\Theta)).
	\end{align*}
	All is left to prove is that it is transition closed.
	Suppose that $\alpha(\Delta) \qmoveto[\mu]{s} \qd'$, we will prove that $\alpha(\Theta) \qmoveto[\mu]{s} \qt'$ and $\qd' \rel \qt'$.
	The other direction is symmetrical.

	\[
		\begin{tikzcd}[ampersand replacement=\&]
			\alpha(\Delta) \arrow[dd, "s:\mu"', two heads, tail] \& \gamma(\alpha(\Delta)) \arrow[dd, "s:\mu", tail] \arrow[r, "\simpl", phantom]  \& \Delta \arrow[dd, "s:\mu", tail] \arrow[r, "\simpl", phantom] \& \Theta \arrow[dd, "s:\mu"', tail]             \& \alpha(\Theta) \arrow[dd, "s:\mu", two heads, tail] \\
			{} \arrow[r, "\Rightarrow", phantom]                 \& {} \arrow[r, "\Rightarrow", phantom]                                          \& {} \arrow[r, "\Rightarrow", phantom]                          \& {} \arrow[r, "\Rightarrow", phantom]         \& {}                                                  \\
			\qd'                                                 \& \Delta' \arrow[r, "\simpl", phantom] \arrow[l, "\alpha", dashed, maps to]     \& \Delta'' \arrow[r, "\simpl", phantom]                         \& \Theta' \arrow[r, "\alpha", dashed, maps to] \& \alpha(\Theta')
		\end{tikzcd}
	\]
	From left to right, we have used \autoref{lm:alphaGammaSem}, \autoref{lm:alphaGammaInverse}, definition of $\rel$, \autoref{lm:alphaGammaSem}, and we use the dashed alpha arrow to indicate when a probability distribution gets transformed in a density distribution.
\end{proof}

\begin{lemma}\label{lm:fullAbstrBkwd}
	For any two density distributions $\qd, \qt \in \qdist{\procset}$, if $\qd \simdl \qt$, then $\gamma(\qd) \simpl \gamma(\qt)$
\end{lemma}
\begin{proof}
	We define the relation $\mathord{\rel} \subseteq \sdist{\confset} \times \sdist{\confset}$ as
	\[
		\mathord{\rel} = \left\{(\Delta, \Theta) \mid \qd \simdl \qt, \Delta \simpl \gamma(\qd), \Theta \simpl \gamma(\qt) \right\}
	\]
	and prove it to be a labelled bisimilarity.
	It is environment preserving thanks to \autoref{lm:alphaEnv}, as
	\[ \env{\Delta} = \env{\gamma(\qd)} = \env{\gamma(\qt)} = \env{\Theta}.
	\]
	It is superoperator closed thanks to \autoref{lm:gammaSop} and the fact that $\simpl$ is superoperator closed, as
	\begin{align*}
		\Delta \rel \Theta
		 & \Rightarrow
		\Delta \simpl \gamma(\qd) \wedge \Theta \simpl \gamma(\qt)                                                               \\
		 & \Rightarrow
		\E(\Delta) \simpl \E(\gamma(\qd)) \simpl \gamma(\E(\qd)) \wedge \E(\Theta) \simpl \E(\gamma(\qt)) \simpl \gamma(\E(\qt)) \\
		 & \Rightarrow \E(\Delta) \rel \E(\Theta)
	\end{align*}
	where in the last row we have used the fact that $\qd \simdl \qt$ implies $\E(\qd) \simdl \E(\qt)$.
	All is left to prove is that it is transition closed.
	Suppose that $\Delta \dmoveto[\mu]{s} \Delta'$, we will prove that $\Theta \dmoveto[\mu]{s} \Theta''$ and $\Delta' \rel \Theta''$.
	The other direction is symmetrical.

	\[
		\begin{tikzcd}[ampersand replacement=\&]
			\Delta \arrow[dd, "s:\mu", tail] \arrow[r, "\simpl", phantom] \& \gamma(\qd) \arrow[dd, "s:\mu", tail]         \& \alpha(\gamma(\qd)) \arrow[dd, "s:\mu", two heads, tail] \arrow[r, "=", phantom] \& \qd \arrow[dd, "s:\mu", two heads, tail] \arrow[r, "\simdl", phantom] \& \qt \arrow[dd, "s:\mu", two heads, tail] \& \gamma(\qt) \arrow[dd, "s:\mu", tail] \arrow[r, "\simpl", phantom]        \& \Theta \arrow[dd, "s:\mu", tail] \\
			{} \arrow[r, "\Rightarrow", phantom]                         \& {} \arrow[r, "\Rightarrow", phantom]          \& {} \arrow[r, "\Leftrightarrow", phantom]                                         \& {} \arrow[r, "\Rightarrow", phantom]                                  \& {} \arrow[r, "\Rightarrow", phantom]     \& {} \arrow[r, "\Rightarrow", phantom]                                     \& {}                               \\
			\Delta' \arrow[r, "\simpl", phantom]                          \& \Delta'' \arrow[r, "\alpha", dashed, maps to] \& \alpha(\Delta'') \arrow[r, "=", phantom]                                         \& \alpha(\Delta'') \arrow[r, "\simdl", phantom]                         \& \qt'                                     \& \Theta' \arrow[r, "\simpl", phantom] \arrow[l, "\alpha", dashed, maps to] \& \Theta''
		\end{tikzcd}
	\]
	From left to right, we have used the definition of $\rel$, \autoref{lm:alphaGammaSem}, \autoref{lm:alphaGammaInverse}, definition of $\rel$, \autoref{lm:alphaGammaSem}, definition of $\rel$,
	and we use the dashed alpha arrow to indicate when a probability distribution gets transformed in a density distribution.
	To show that $\Delta' \rel \Theta''$, it is sufficient to notice that $\alpha(\Delta'') \simdl \qt$, and if we apply $\gamma$ to both we get by \autoref{lm:alphaGammaInverse} $\gamma(\alpha(\Delta'')) \simpl \Delta'' \simpl \Delta'$ and $\gamma(\qt) = \gamma(\alpha(\Theta')) \simpl \Theta' \simpl \Theta''$.
\end{proof}

\begin{lemma}\label{lm:fullAbstr}
	For any $\Delta, \Theta \in \sdist{\confset}$, $\Delta \simds \Theta$ if and only if $\alpha(\Delta) \simdl \alpha(\Theta)$.
\end{lemma}
\begin{proof}
	Let $\Delta \simds \Theta$, by \autoref{thm:probcc} it holds if and only if $\Delta \simpl \Theta$.
	It consequently holds if and only if $\alpha(\Delta) \simdl \alpha(\Theta)$ from \autoref{lm:fullAbstrFwd} and \autoref{lm:fullAbstrBkwd}.
\end{proof}

\demonicfa*
\begin{proof}
	By \autoref{lm:fullAbstr} $\delem{1}{\conf{\rho, P}} \simds \delem{1}{\conf{\sigma, Q}}$ if and only if $\alpha(\delem{1}{\conf{\rho, P}}) \simdl \alpha(\delem{1}{\conf{\sigma, Q}})$ and by definition of $\alpha$, $\delem{\rho}{P} \simdl \delem{\sigma}{Q}$.
\end{proof}

\begin{lemma}\label{thm:move_to_sop_move}
	For any $\qd, \qd' \in \qdist{\procset}$ with $(\Sigma, \Sigma') \vdash \qd$, if $\qd \moveto[\mu]{s} \qd'$ then for all $\E \in \soset{\Sigma \setminus \Sigma'}$ $\E(\qd) \moveto[\mu]{s} \E(\qd')$.
\end{lemma}
\begin{proof}
	First, note that transitions are uniquely determined by the process syntax and thus, all the transitions available for $\qd$ are also available for $\E(\qd)$, although possibly with a different resulting distribution for the application of either a superoperator or a measurement.
	However, it is sufficient to note that, the operator $\E$ and the applications of superoperators and measurements are on distinct qubits and thus they commute.
\end{proof}

\begin{lemma}\label{thm:sop_move_exist}
	For all $\qd, \qd' \in \qdist{\procset}$ with $(\Sigma, \Sigma') \vdash \qd$, and $\E \in \soset{\Sigma \setminus \Sigma'}$, if $\E(\qd) \moveto[\mu]{s} \qd'$ then there exists $\qd'' \in \qdist{\procset}$ such that $\qd \moveto[\mu]{s} \qd''$ and $\qd' = \E(\qd'')$.
\end{lemma}
\begin{proof}
	By induction on the transition $\E(\qd) \moveto[\mu]{s} \qd'$.
\end{proof}

\begin{lemma}\label{thm:g_fullabs_rel}
	For all $\qd, \qt \in \qdist{\procsetnr}$, if $\qd \sim_g \qt$ with $(\Sigma, \Sigma') \vdash \qd$ and $(\Sigma, \Sigma') \vdash \qt$, then $\E(\qd) \sim_g \E(\qt)$ for all $\E \in \soset{\Sigma \setminus \Sigma'}$.
\end{lemma}
\begin{proof}
	We will prove that the relation
	\[
		\soset{\sim_g} = \{(\E(\qd), \E(\qt)) \mid (\Sigma, \Sigma') \vdash \qd, (\Sigma, \Sigma') \vdash \qt,  \qd \sim_g \qt, \E \in \soset{\Sigma \setminus \Sigma'} \}
	\]
	is a ground bisimulation.

	The relation $\soset{\sim_g}$ is trivially well-typed because superoperator application does not affect the support of distributions. 
	Moreover, it is environment preserving, since the superoperator $\E$ only affects the qubits in the environment and thus commutes with the partial trace, i.e. $\tr_{\Sigma'}(\qs(\E(\qd))) = \E(\tr_{\Sigma'}(\qs(\qd)) = \E(\tr_{\Sigma'}(\qs(\qt)) = \tr_{\Sigma'}(\qs(\E(\qt)))$.

	We show transition closure by considering the case of reception separately.
	Assume a transition $\E(\qd) \moveto[\mu]{s} \qd'$ exists with $\mu \neq c?e$.
	By~\autoref{thm:sop_move_exist} there exists a $\qd''$ such that $\qd \moveto[\mu]{s} \qd''$ and $\qd' = \E(\qd'')$.
	However, $\qd \sim_g \qt$, therefore, there exists a $\qt''$ such that $\qt \moveto[\mu]{s} \qt''$ and $\qd'' \sim_g \qt''$.
	Since the move is not a reception, the environment of $\qd''$ and $\qt''$ contains the same qubit or one more, thus by~\autoref{thm:move_to_sop_move} $\E(\qt) \moveto[\mu]{s} \E(\qt'')$, and by definition of $\soset{\sim_g}$ we have $\E(\qd'') \mathrel{\soset{\sim_g}} \E(\qt'')$.

	Assume a transition $\E(\qd) \moveto[c?e]{s} \qd'$ exists.
	As in the previous case, $\qd \sim_g \qt$ together with~\autoref{thm:sop_move_exist} and~\autoref{thm:move_to_sop_move} imply that there exists a $\qt'$ such that $\E(\qt) \moveto[c?e]{s} \qt'$ with $\qd' \sim_g \qt'$.
	Note that $c$ cannot be a quantum channel because of the syntactic constraints of $\procsetnr$, and the type of $\qd$ is the same of the one of $qd'$.
	Then, since $\I(\qd') = \qd'$ and $\I(\qt') = \qt'$ we also have $\qd' \mathrel{\soset{\sim_g}} \qt'$.
\end{proof}

\begin{corollary}
	Let $\qd, \qt \in \qdist{\procset}$.
	If $\qd \simgdl \qt$ then $\qd \simdl \qt$.
\end{corollary}

\begin{lemma}\label{lm:GroundfullAbstr}
	For any $\Delta, \Theta \in \sdist{\procsetnr}$, $\Delta \simds \Theta$ if and only if $\alpha(\Delta) \simgdl \alpha(\Theta)$.
\end{lemma}
\begin{proof}
	Follows immediately from \autoref{thm:g_fullabs_rel}.
\end{proof}

\demonicgfa*
\begin{proof}
	Follows immediately from \autoref{lm:GroundfullAbstr}.
\end{proof}

\end{document}